\documentclass[authoryear,review]{article}

\makeatletter
\def\ps@pprintTitle{%
	\let\@oddhead\@empty
	\let\@evenhead\@empty
	\def\@oddfoot{}%
	\let\@evenfoot\@oddfoot}
\makeatother

\usepackage[utf8]{inputenc}
\usepackage{fontenc}[t1]
\usepackage{tikz}

\usepackage{natbib}
\usepackage{subcaption}
\usepackage{float}
\usepackage{verbatim}
\usepackage{wrapfig}
\usepackage{listings}

\usepackage{geometry}                
\usepackage{amssymb}
\usepackage{amsmath}
\usepackage{amsthm}
\usepackage{bigints}
\usepackage{placeins} 
\usepackage{lineno}
\usepackage{enumitem}
\usepackage{lipsum} 

\usepackage{comment}
\usepackage{color, colortbl}
\usepackage{xcolor}

\usepackage{graphicx}
\usepackage{epstopdf}
\definecolor{magenta}{rgb}{.5,0,.5}
\definecolor{black}{rgb}{1.0,1.0,1.0}
\definecolor{magenta}{rgb}{.1,0,.3}
\definecolor{gruen}{rgb}{0.2,0.5,.5}
\definecolor{light}{rgb}{ 0.992, 0.961,  0.902}
\definecolor{Tan}{rgb}{ 0.992, 0.9,  0.902}

\newcommand{\field}[1]{\mathbb{#1}}
\newcommand {\R}        {\field{R} }

\newcommand {\e}  {{\mathbf e}}
\newcommand {\eps}      {\varepsilon}
\theoremstyle{plain}

\newtheorem{prop}{Proposition}[section]

\newtheorem{rem}[prop]{Remark}

\newtheorem{theorem}[prop]{Theorem}
\newtheorem{lemma}[prop]{Lemma}

\newcommand{\komment}[1]{}

\title{Quiescence wins: The Discovery Of Slowness}
\author{Leonard Braun, Noah Risse, Aur\'elien Tellier, Johannes M\"uller\\TU Munich}

\begin{document}
	\maketitle

\begin{abstract}
We investigate the evolution of quiescence within the framework of Adaptive Dynamics for an SIQS (Susceptible - Infected - Quiescent) model with constant environment. In the first part of the paper, the competition of two strains which have the same basic fitness (same reproduction number) but different timing in quiescence is analyzed. Thereto, the complexity of the model is reduced: By a time scale argument, we approximate the SQIS model by an SIS model under the assumption of rapid switching between I and Q. Furthermore, using dimension reduction methods for the van Kampen expansion of the models, we replace the multi-dimensional SDE by an effective one dimensional SDE on the center manifold of the models. Finally, the fixation probabilities of the strains are derived.\\
In the second part of the paper, we use concepts from adaptive dynamics to analyze the stochastic random walk on the trait space. We find that quiescence is favored by intrinsic stochastic effects. In the end, a more intuitive explanation is added to the mathematical arguments.\\ 
The analysis suggests a new paradigm for the evolution of quiescence in parasites: In the present context, quiescence is not a bed-hedging strategy to escape detrimental conditions in a fluctuating environment, but a simple and efficient method to continuously control and slow down the time scale of life-history traits. It turns out that evolution favors slowness, with the analysis suggesting that this effect may be widespread in systems where species interact only indirectly through competition for resources. 
\end{abstract}

\section{Introduction}

Why should pathogens exhibit dormant or quiescent phases? The answer was developed about 30 years ago~\cite{venable1988selective}: Quiescent individuals 
are protected. They are protected against sporadic adverse environmental conditions such as antibiotics, drought, and lack of nutrients, to name but a few. Quiescence is a bet-hedging strategy. Really? Is that the whole answer? 
The paradigm is cracking, quiescence can also develop in a constant environment~\cite{Blath.2020}. But then, why be slow? 
Faster runners are expected to win the race~\cite{lee2022slow}! But that is not the whole story. Slowness has its own strength. 
This is the topic of the present article: The Discovery Of Slowness~\cite{Nadolny2005}.\par\medskip

Quiescence is a trait which is widespread among fungi and microorganisms~\cite{lennon2011microbial}, plants~\cite{tellier2019,verin2018host,blath2021branching} and also pathogens~\cite{sorrell2009evolution,rittershaus2013normalcy}. 
As quiescence seemingly decreases the arithmetic fitness -- only the active part of the population reproduces, while the quiescent 
fraction does not contribute to offspring -- the interpretation of quiescence is a challenge in the context of 
evolution. A classic explanation is that quiescence can be considered as a bet-hedging strategy~\cite{venable1988selective,stumpf2002herpes,blath2021branching,Browning2021}. 
Subpopulations following different behaviors are optimized for different environments. In the case of pathogens, 
the proliferating subpopulation successfully reproduces when infecting an individual, while it will likely be 
wiped out by antibiotics. On the other hand, the quiescent subpopulation cannot contribute to the population 
growth during an ongoing infection, but has a chance to survive antibiotics as their metabolism is largely switched off. 
A population consisting of subpopulations of active and quiescent individuals is obviously best prepared for 
a switching environment. In that, the inclusive fitness of the population is optimized, or at least 
increased. Of course, there has to be a mechanism that allows individuals to change their type 
once in a while to guarantee a genetically homogeneous population. The switching rates will determine the size of the subpopulation 
fractions and are optimized to the time scale of the switching environment~\cite{gaal2010exact,Mueller2013}. 
This behavior is one of the widespread varieties of division 
of labour~\cite{claessen2014bacterial,zhang2016understanding}.
\par\medskip 

Dormant phases are not only observed in individual pathogens, but also in infected hosts. Individuals can alternate between infectious and dormant phases (sometimes referred to as a silent infection or covert infection~\cite{sorrell2009evolution}). Here, too, a changing environment can be a possible cause. For example, the change between dry and rainy seasons has a strong effect on the spread of malaria as the vector is not available during the dry season. The pathogen needs to persist within a host.
{\it Plasmodium vivax}~\cite{Gowda2018}, {\it Plasmodium ovale}~\cite{merrick2021hypnozoites} or {\it Plasmodium cynomolgi}~\cite{Joyner2016} (which mainly infects apes) show dormant phases where the pathogen retracts into the liver. These dormant phases can perhaps be considered as a reaction to the switching environment. On the other hand, the timing of these 
episodes (e.g.\ for {\it Plasmodium cynomolgi}, the infectious period is about one week, with a relapse after 2-4 weeks~\cite{Joyner2016}), are much faster than the annual rhythm of the climate, in contrast to the theoretical predicitons~\cite{gaal2010exact,Mueller2013}. Furthermore, there are other plasmodium strains, such as {\it Plasmodium falciparum}, which do not show a significant dormant behavior~\cite{phillips2017malaria}. In that, it can be discussed if, indeed, the annual rhythm can be identified as the evolutionary force that leads to quiescent phases in malaria.\\
Furthermore, some viral infections show dormant phases, such as chickenpox, which goes into a dormant state and withdraws into the nervous system; the virus can become active again and cause {\it herpes zoster}~\cite{Hadley2016}. More viruses with a quiescent phase on the host level are, for example, Herpes viruses in general~\cite{stumpf2002herpes}, or the  Epstein-Barr Virus~\cite{thorley2013pathogenesis,Chakravorty2022}.\\
A further example of infections with quiescent phases is bacteria infected by bacteriophages. 
Bacteriophage infections might be lytic or lysogenic, i.e.~they either multiply actively in the bacterium or are dormant~\cite{Spriewald2020}.\\
In the present work, we do not deal with a specific life cycle of a specific pathogen but study a general SIQS model to investigate the effects of dormancy at the host level in a generic framework. The model, of course, requires adaptation if a specific pathogen is of interest.\par\medskip 

A recent finding by Blath et al.~\cite{Blath.2020} provides a first indication that a changing environment is not always necessary for dormancy to be selected by evolution. The authors consider competing species in a stochastic model within constant conditions and find that intrinsic fluctuations in population dynamics are sufficient to trigger quiescence as a bet-hedging strategy. This observation can be seen as a first step away from the switching/fluctuating environment paradigm. \\
Like Blath et al., also  the present paper is not dealing with a switching environment, but the only perturbations come from the intrinsic noise caused by a finite population size. We emphasize the difference between the two scenarios: A switching environment induces a clear fitness advantage for quiescence if its timing is in resonance with that of the environment. This result holds even in the deterministic limit~\cite{Gaal.2010}. In the present work, we only consider intrinsic noise due to a finite population, which vanishes in the deterministic limit. Consequently, in this limit, we neither find an advantage nor a disadvantage of quiescence. This observation can be easily understood by considering the effective reproduction number in a constant environment: Quiescence divides the active/infectious time period of an individual into several time intervals, which are interrupted by resting phases. Hereby, the total active time is neither increased nor decreased. Therefore, the effective reproduction number in a constant environment is not affected by quiescence. And indeed, the analysis of the deterministic model shows the occurrence of a line of steady states.\\
However, when considering a finite population, the intrinsic fluctuations lead to a weak advantage for the trait with more quiescence (see~\cite{Blath.2020} and this paper). Adaptive Dynamics indicates that the tendency towards more quiescence does not reach a balance, but instead becomes stronger and stronger. One might argue that this unbounded tendency towards more and more quiescence arises from the absence of a timescale in intrinsic fluctuations, which exhibit characteristics akin to Wiener noise with a frequency density of $1/f^2$.
However, we present an additional argument that demonstrates that bet-hedging is not the primary driver of quiescence in our set-up: Within an appropriate range of parameters, the quiescence model can be well approximated by a model that lacks quiescence, but with appropriately scaled parameters. In this case, the strains (without quiescence) remain neutral in the deterministic limit, resulting in a continuum of stationary states. The key distinction lies not in the reproduction number, but in the timescale of the strains. As the work of Kogan et al.~\cite{Kogan.2014} shows, the intrinsic noise gives an advantage to the slower time scale. 
That is, we argue that in the current constellation, quiescence is just a "clever" way to slow down the pace of a strain without affecting the reproduction number, and in that, obtaining a weak advantage if the intrinsic noise due to a finite population is present.\par\medskip 

The article is structured as follows: We first set up the stochastic SIQS model, and indicate the evolutionary scenario we 
intend to analyze: Two strains of equal fitness, that is, with identical reproduction numbers, compete. The model incorporates 
four compartments (active and quiescent compartments for the two competing strains). As this model is mathematically challenging, we reduce the complexity in Section~2: On the one hand, if the switching rates between I and Q are large, the I and the Q compartment can be replaced by an effective I compartment with averaged infection and recovery rates. In that, the resulting SIS model still incorporates the switching rates of the SIQS model, 
but merely in order to adapt the time scale. On the other hand, it turns out that both models (the SIQS and the SIS model) 
exhibit a one-dimensional attractive manifold, caused by the fact that both strains only differ in the quiescence switching rates, 
but not in the reproduction number. An effective, one-dimensional stochastic differential equation on the center manifold can be derived, and from that, 
also the fixation probabilities of the competing strains. In consequence, a kind of ``competitive exclusion'' principle is obtained. Based on this, we develop in Section~3 
a variant of Adaptive Dynamics suited to handle the resulting random walk in the trait space. Herein, we find that quiescence 
is favored. As the mathematical arguments that lead to this result are quite technical, we also re-discuss the mechanisms on an informal 
level to give some intuition. The results are put into context in the discussion, Section~4. Most of the proofs are moved to the supplementary information (SI).

\section{Model}
\begin{figure}[tb]
    \centering
    \begin{minipage}{0.9\textwidth}
        \centering
	    \begin{tikzpicture}[->,>=latex,auto,node distance=2.5cm,thick,main node/.style={circle,draw,font=\sffamily\Large\bfseries}]
      \node[main node] (S) {$S$};
      \node[main node] (I_1) [left of=S, xshift=-1cm] {$I_1$};
      \node[main node] (I_2) [right of=S, xshift=+1cm] {$I_2$};
        \node[main node] (Q_1) [below of=I_1] {$Q_1$};
      \node[main node] (Q_2) [below of=I_2] {$Q_2$};
      \path[every node/.style={font=\sffamily\small}]
        (S) edge [bend right] node [above] {$\beta_1 S I_1/N$} (I_1)
        (I_1) edge [bend right] node [below] {$\gamma_1 I_1$} (S)
        (S) edge [bend left] node [above] {$\beta_2 S I_2/N$} (I_2)
        (I_2) edge [bend left] node [below] {$\gamma_2 I_2$} (S)
        (I_1) edge [bend left] node [right] {$\hat{\omega}_1 I_1$} (Q_1)
        (Q_1) edge [bend left] node {$\hat{c}_1 Q_1$} (I_1) 
        (I_2) edge [bend left] node {$\hat{\omega}_2 I_2$} (Q_2)
        (Q_2) edge [bend left] node {$\hat{c}_2 Q_2$} (I_2)
        (Q_1) edge [bend right] node [below] {$\hat{\mu}$} (S)
        (Q_2) edge [bend left] node {$\hat{\mu}$} (S)
        ;
    \end{tikzpicture}
        \caption{Illustration of the transitions in the one-strain SIQS model.}
        \label{figDiagramSIQS}
    \end{minipage}
\end{figure}
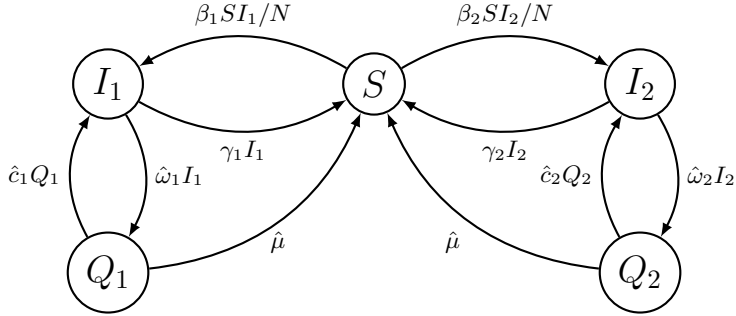
This study aims to investigate the circumstances under which quiescence develops. Thereto, the competition of two rather similar strains is investigated: The difference between the two strains is the degree of quiescence, while all other properties are chosen so that the strain's fitness (reproduction number) is identical. With the standard notation of a two-strain SIQS model with quiescence, we have (also see 
Fig.~\ref{figDiagramSIQS})\\[4pt]
{    \begin{center}
        \begin{tabular}{||c | c | c||} 
         \hline
         Event & Type of transition & Rate \\ [0.5ex] 
         \hline\hline
         Infection with strain 1 &  $I_1 \to I_1+1$ & $\beta S I_1/N$ \\ 
         \hline
         Infection with strain 2 &  $I_2 \to I_2+1$ & $\beta S I_2/N$ \\ 
         \hline
         Recovery of $I_1$ &  $I_1 \to I_1-1$ & $\gamma_1 I_1$ \\
         \hline
         Recovery of $I_2$ & $I_2 \to I_2-1$ & $\gamma_2 I_2$  \\
          \hline
         Deactivation of $I_1$ & $(I_1,Q_1) \to (I_1-1,Q_1+1)$ & $\hat{\omega}_1 I_1$  \\
           \hline
         Deactivation of $I_2$ & $(I_2,Q_2) \to (I_2-1,Q_2+1)$ & $\hat{\omega}_2 I_2$  \\
          \hline
         Reactivation of $Q_1$ & $(I_1,Q_1) \to (I_1+1,Q_1-1)$ & $\hat{c}_1 Q_1$  \\
           \hline
         Reactivation of $Q_2$ & $(I_2,Q_2) \to (I_2+1,Q_2-1)$ & $\hat{c}_2 Q_2$  \\ 
         \hline
		 Death/recovery of $Q_1$ & $Q_1 \to Q_1-1$ & $\hat \mu Q_1$  \\ 
           \hline
         Death/recovery of $Q_2$ & $Q_2 \to Q_2-1$ & $\hat \mu Q_2$  \\ [1ex] 
         \hline
        \end{tabular}
		\end{center} 
		}\par\medskip\noindent 
where always $S=N-I_1-Q_1-I_2-Q_2$. Note that we particularly assume that also quiescent individuals 
are allowed to recover (or die). It is well known that if $\hat\mu=0$, the reproduction number of a strain does not depend 
on the quiescent phase/the parameters for quiescence; if, however, we choose $\hat\mu>0$, the reproduction number will decrease by the net time an infected individual spends in the quiescent state. The fitness of the strains can be measured by the respective reproduction number.

\begin{theorem}\label{reproNumber}
The reproduction number $R_{0,i}$ of strain $i\in\{1,2\}$ is given by 
$$ R_{0,i} = \frac{\beta_i}{\gamma_i+\hat \mu\,\frac{\hat \omega_i}{\hat c_i+\hat \mu}}. 
$$
\end{theorem}
(Proof: SI, Section~1.1). As expected, we have $R_{0,i}=\beta/\gamma_i$ if $\hat\mu=0$ 
and the reproduction number is decreasing in $\hat\mu$. We aim to compare the effect of quiescence 
in a ``fair'' way -- not the fitness, but the timing due to quiescence is the driving force for the 
dynamics. Therefore, our standing hypothesis is 
$$ R_{0,1} \equiv  R_{0,2} = R_0.$$
which reduces the dimension of the parameter space. If $\hat\omega_i$ are given, 
then $\hat c_i$ are determined. \par\medskip 

We require formulations of the model on different scales for the population size $N$: If $N$ is small 
(around 100, say), the individual-based process as described is adequate. We can use the master equations  (as stated in SI, Section~1.2) to discuss the fate of the model at 
this population size. Technically, this description has the 
disadvantage that the state space is discrete. That limits the available tools to reduce the complexity of the model. 
If, however, 
the population size is larger ($N\approx 1000$, to give a number), the diffusion limit 
(van Kampen system size expansion~\cite{Gardiner.2009}) of the particle 
system is an appropriate description. Herein, we consider relative frequencies 
$x = I_1/N$, $u=Q_1/N$, $y=I_2/N$ and $v=Q_2/N$. 
We obtain a stochastic differential equation (SDE), where the state space 
is continuous. The Fokker-Planck equation of this SDE is stated in SI, Section~A.2. 
More techniques, particularly for dimension reduction are available. 
If, however,  $N$ tends to infinity, the noise terms cancel out, and we end up with the 
ODE ($s=1-(x+y+u+v)$)
\begin{eqnarray*}\label{eq: deterministic SIIQQS rescaled}
x' &=& \beta_1 \, s\, x-\hat\omega_1\, x+\hat c_1 u  -\gamma_1 x\\
u' &=& \hat\omega_1\, x-\hat c_1  u-\hat\mu u\\
y' &=& \beta_2 \, s\, y -\hat\omega_2\, y+\hat c_2 v -\gamma_2 y\\
v' &=& \hat\omega_2\, y-\hat c_2  v-\hat\mu v
\end{eqnarray*}
As the reproduction numbers for both strains are the same, a short computation 
indicates that this ODE system has a line of stationary states, the coexistence line (CL). We only state this 
theorem for the core model, where all parameters are identical but the switching rates between I and Q, and furthermore $\hat\mu=0$. 

\begin{theorem} Let $\gamma_1=\gamma_2$ and $\hat\mu=0$. 
If $R_{0,1}=R_{0,2}=R_0$ the ODE has a line of stationary points, given by 
$$ (x,u,y,v) = 
\bigg(x^\ast,\,\, 
\vartheta_1 x^\ast,\,\, 
\frac{r-(1+\vartheta_1)\, x^\ast}{1+\vartheta_2},\,\, 
\vartheta_2\frac{r-(1+\vartheta_1)\, x^\ast}{1+\vartheta_2}\bigg),
\qquad x^\ast \in[0,r/(r+\vartheta_1)]$$
where $r = 1-1/R_0$, $\vartheta_i = \hat\omega_i/\hat c_i$. This coexistence line is transversally stable. 
\end {theorem}
The result follows by direct computations and an application of the centre manifold theorem~\cite{guckenheimer2013nonlinear} (all eigenvalues
of the linearisation have negative real part but one which is zero).\par\medskip 

For the case of two strains with an equal reproduction number, a standard situation often discussed is obtained: 
The drift of the stochastic model drives the state into the vicinity of the centre manifold/line of stationary 
points on a fast time scale; once the state is there, the drift term becomes small, and the noise takes over (Fig.~\ref{simulSIS}). 
The system performs a random walk along this coexistence line until one strain goes into fixation (the other strain dies out). 
The approximation results target the characterization of this one-dimensional random walk to get a grip on the absorption/fixation probabilities.

\section{Reduction of complexity}

\begin{figure}[tb]
\begin{center}
\includegraphics[width=12cm]{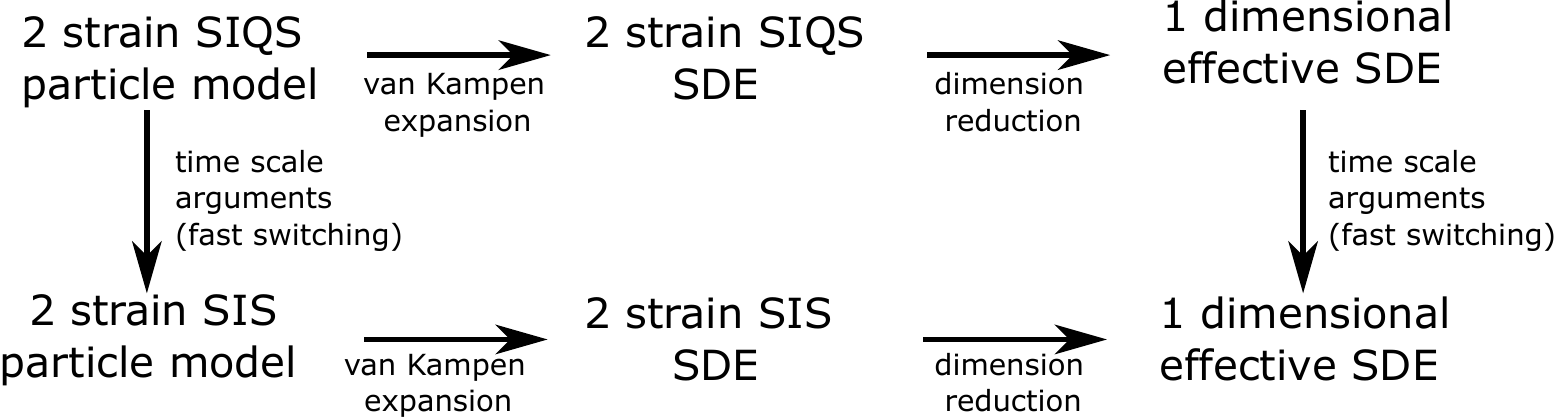}
\end{center}
\caption{Reduction of the model's complexity by time scale arguments and by dimension reduction. The diagram commutes..}\label{complexReduct}
\end{figure}
Let's outline the steps we take to reduce complexity. Our starting point is a SIQS particle model. Because the state space is discrete, we primarily rely on a time-scale argument. If we assume that 
the switching rates between the I and the Q states are large, the particle system quickly approaches a quasi-steady state. This allows us to replace the SIQS model with an SIS model with re-scaled parameters (see Fig.~\ref{complexReduct}). That is, we have two particle systems describing the competition of two strains, one SQIS and one SIS model. These systems exhibit similar behavior, as we will see in section~\ref{timeScaleReduction}.\par\medskip 

To further reduce the state space, we use a standard van Kampen expansion to derive a Fokker-Planck equation. 
The resulting state space of the Fokker-Planck equation is 4 dimensional (SIQS model) resp.\ 2 dimensional (SIS system). The advantage 
of the transition from a discrete to a continuous state space is that the drift term of the Fokker-Planck equations 
coincides a deterministic ODE, and drives the state towards a 1-dimensional  line of (deterministic) stationary states. \\
In order to understand the long-term dynamics of the systems, it is sufficient to understand the one-dimensional dynamics along this 
line of stationary states. In this dimension reduction, we follow the method of Kogan, and obtain 
one-dimensional Fokker Planck equations (one for the SIQS and one for the SIS model, Section~\ref{dimReducitonSection}).\\
It should be noted that the limit based on rapid switching rates transfers the one-dimensional Fokker-Planck equation for the  SIQS model exactly to the one-dimensional Fokker-Planck equation for the SIS model. In this respect, the picture closes (end of Section~\ref{dimReducitonSection}). \par\medskip 

The SDE's derived with the help of dimension reduction allow to work out the fixation probabilities of the two strains (Section~\ref{fixProbab}). These fixation probabilities will be a central building block used by Adaptive Dynamics.\par\medskip 

We describe the steps below in a rather informal way; the (partially lengthy) technical derivations can be found in the SI.

\subsection{Rapid switching: From the branching SIQS to the branching SIS model}
\label{timeScaleReduction}
We assume that $\hat c$ and $\hat \omega$ are large. We might formally express this fact by replacing 
$$ \hat c\mapsto \hat c/\eps,\qquad  \hat \omega\mapsto\hat\omega/\eps$$
where $\eps$ is small. For a one-strain model, this limit has been discussed in~\cite{usman:arxiv}. In that, we will be able to remove the Q state. The idea is based on Markov processes on 
two time scales, as e.g.~explained in the monography~\cite{Yin.2013}: At a given point in time, 
the total number of  individuals infected by strain $i\in\{1,2\}$ is given by 
$\tilde I_i=I_i+Q_i$. 
Any infected individual rapidly oscillates between $I$ and $Q$, such that it is dormant  (state $Q_i$) with probability $\zeta_i := \hat c_i/(\hat c_i + \hat\omega_i)$, and infectious (state $I_i$) with probability $1-\zeta_i$. We can remove the stratification of infected individuals into an infectious and a dormant class. Instead,  we introduce effective infection/removal rates, 
$$ \beta\mapsto \beta\frac{\hat c_i}{\hat c_i + \hat\omega_i}, \qquad 
\gamma_i\mapsto \gamma_i\frac{\hat c_i}{\hat c_i + \hat\omega_i} + \hat\mu   \frac{\hat \omega_i}{\hat c_i + \hat\omega_i}$$
which weight the respective rates by the probability (or fraction of time) an infected individual spends in the 
infectious/quiescent state. The proof of the following theorem can be found in SI, Section~2.1.
\begin{theorem}\label{rapidSwitchTheo}
Assume that $\hat\omega_i$, $\hat c_i\rightarrow\infty$ such that $\frac{\hat \omega_i}{\hat c_i + \hat\omega_i}$ converge to $\zeta_i\in(0,1)$.
Then, the 2-strain SIQS model converges in distribution to a 2-strain SIS model with transition rates
{    \begin{center}
        \begin{tabular}{||c | c | c||} 
         \hline
         Event & Type of transition & Rate \\ [0.5ex] 
         \hline\hline
         Infection with strain 1 &  $\tilde I_1 \to \tilde I_1+1$ & $\beta\,(1-\zeta_1) S \tilde I_1/N$ \\ 
         \hline
         Infection with strain 2 &  $\tilde I_2 \to \tilde I_2+1$ & $\beta\,(1-\zeta_2) S \tilde I_2/N$ \\ 
         \hline
         Recovery of an infected individual (strain 1) &  $\tilde I_1 \to \tilde I_1-1$ & $(\gamma_1(1-\zeta_1)+\hat\mu\zeta_1) \tilde I_1$ \\
         \hline
         Recovery of an infected individual (strain 2)  & $\tilde I_2 \to \tilde I_2-1$ & $(\gamma_2(1-\zeta_2)+\hat\mu\zeta_2) \tilde I_2$  \\
          \hline
        \end{tabular}
		\end{center} 
		}\par\medskip\noindent 
(where $S=N-\tilde I_1-\tilde I_2$) in the sense that  $I_i+Q_i \overset{d}{\to} \tilde I_i$.
\end{theorem}

\subsection{Dimension reduction for the van Kampen expansion of the SIQS and SIS model}
\label{dimReducitonSection}
We do not go into the van Kampen expansion but simply state the result of this diffusion limit (SI, Section~1.2) for the SIQS model. We briefly describe the dimension reduction method 
which goes back to an idea provided in a paper by Kurtz~\cite{Kurtz.1973} and was carried forward by many authors (also see \cite[chapter~1, corollary 7.8]{ethier2009markov} or the review article~\cite[Section~6]{Givon2004extracting}).  
 We sketch the general method, following 
the nice explanations in \cite[chapter 4.4]{Majda2001}, also see~\cite{Heinrich.2018}; the details, and particularly the idea of Kogan et al.~\cite{Kogan.2014} to avoid some involving steps 
are carried out in the SI, Section~3.
\par\medskip 

\begin{figure}
    \centering
        \includegraphics[width=7cm]{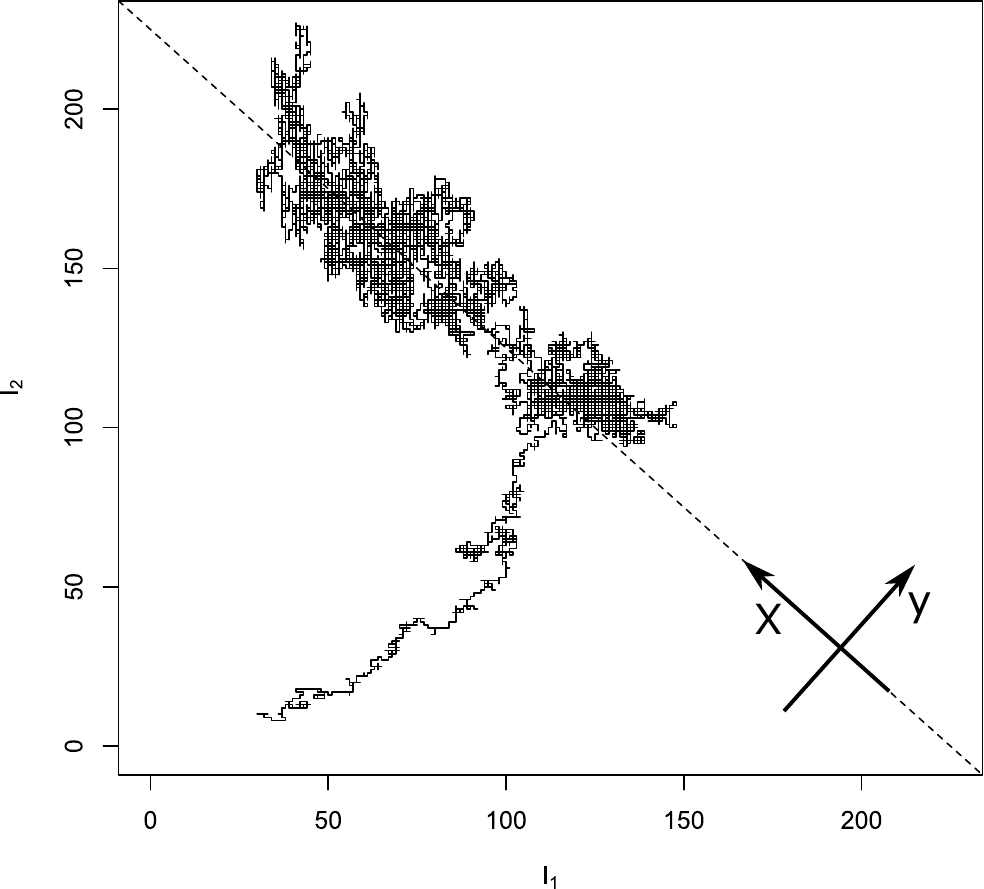}
        \caption{Realization of the individual-based 2-strain SIS model 
        ($\beta_1=1$, $\beta_2=1.5$, $R_0=2$ for both strains, $N=450$, initial value $I_1=30$, $I_2=10$). Dashed line: Coexistence line. $(X,y)$: new coordinate system (note that $Y=y/\sqrt{N}$). }
        \label{simulSIS}
\end{figure}

{\it Preparations -- appropriate coordinate system:}\\
The observations show that the state of the stochastic system moves towards the coexistence line, and performs  
a random walk there (Fig.~\ref{simulSIS}). In the first step, we choose a new coordinate system, 
where $X$ points along this coexistence line, 
and $y_i$ are directions perpendicular to $X$. As the process will stay $1/\sqrt{N}$-close to this line, we blow up 
the $y$ coordinates, $y_i=Y_i\eps$ with $\eps^2=1/N$. The r.h.s.\ of the resulting Fokker-Planck equation depends on $\eps$. 
We intend to investigate the long-term behaviour, and hence choose the time $t = \eps^{2} \tilde\tau$, such that our Fokker-Planck equation reads 
$$\eps^2\partial_{\tilde\tau} \rho^\eps(x,y,\tilde \tau) = L^\eps \rho^\eps(x,y,\tilde\tau).$$
The operator $L^\eps$ can be expanded w.r.t~$\eps$, 
$$ L^\eps \rho^\eps  = L^{(0)} \rho^\eps + \eps L^{(1)} \rho^\eps + \eps^2 L^{(2)} \rho^\eps +\cdots. $$
With $L^\eps$ also the solution $\rho^\eps$ of $\eps^2\partial_t \rho^\eps(x,y,t) = L^\eps \rho^\eps(x,y,t)$ depends on $\eps$.
\par\medskip 

{\it Approximation theorem:}\\
In order to state the approximation theorem, we consider the backward equation, 
$$ - \eps^2 \partial_s w^\eps = L^{(0)+} w^\eps + \eps L^{(1)+} w^\eps + \eps^2 L^{(2)+} w^\eps +\cdots $$
where $L^{(i)+}$ are the formal adjoint operators for $L^{(i)}$. $w^\eps$ is the solution of the backward equation (which depends on $\eps$). 
Formally the expansion of $w^\eps$ in terms of $\eps$ reads
$$ w^\eps = w^{(0)}+\eps w^{(1)}+\eps^2w^{(2)}+\cdots.$$
Inserting $w^\eps$ into the backward equation and  equating the coefficients 
of equal power in $\eps$ yields 
\begin{eqnarray*}
	L^{(0)+} w^{(0)} &=& 0\\
	L^{(0)+} w^{(1)} &=& - L^{(1)+} w^{(0)}\\
	L^{(0)+} w^{(2)} 
	&=& 
	-\partial_s w^{(0)} 
	- L^{(2)+}w^{(0)}
	- L^{(1)+}w^{(1)}
	\\
	&\ldots &
\end{eqnarray*}
Any feasible solution is element of the range of  $L^{(0)+}$.
A natural assumption is that 
 the operator $L^{(0)+}$ generates a semigroup that converges to a projection operator for $t\rightarrow\infty$, 
 $$ \lim_{t\rightarrow\infty} e^{L^{(0)+}\, t} \cdot = {\mathbf P} \,\, \cdot$$
As $L^{(0)+} w^{(0)} = 0$, we conclude ${\mathbf P} w^{(0)}=w^{(0)}$. Therewith, we understand the consequence of the first equation.
\par\medskip 

Also the second equation 
$L^{(0)+} w^{(1)} = - L^{(1)+} w^{(0)}$ implies a solvability condition, 
$${\mathbf P}\,\,L^{(1)+} w^{(0)} = 0.$$
Therefore, there is a (pseudo)-inverse of $L^{(0)+}$, and 
$$ w^{(1)} = -  (L^{(0)+})^{-1} L^{(1)+} w^{(0)}.$$
Kogan's method (see SI, Section~3) 
centrally suggests a way of avoiding the explicit computation of these pseudo-inverses.  \par\medskip 

Also, the solvability of the third equation gives us information, 
$$0= {\mathbf P}\,\,(	-\partial_s w^{(0)} 
- L^{(2)+}w^{(0)}
- L^{(1)+}w^{(1)}).$$
We replace $w^{(1)}$ by our result from above, and find 
$$-\partial_s w^{(0)} = {\mathbf P}\,L^{(2)+}w^{(0)} - 
{\mathbf P} (L^{(1)+})(L^{(0)+})^{-1} L^{(1)+}w^{(0)}.$$
Of course, the arguments presented here have only been formally, but it is possible to extend them to a proof of the 
following theorem~\cite[chapter 4.4]{Majda2001}. 
The proof of the following theorem can be found in SI, Section~3.

\begin{theorem}\label{approxSIQS}
	Let $w^{\eps}(s,x,y |t)$ satisfy 
	$$ - \eps^2 \partial_s w^\eps = L^{(0)+} w^\eps + \eps L^{(1)+} w^\eps + \eps^2 L^{(2)+} w^\eps, \qquad w^\eps(t,x,y|t) = f(x) $$
	where $L^{(0)+}$, $L^{(1)+}$, $L^{(2)+}$ are backward Fokker-Planck 
	operators, $L^{(0)+}$ generates a stationary process such that 
	 $$ \lim_{t\rightarrow\infty} e^{L^{(0)+} t} \cdot = {\mathbf P} \,\, \cdot,$$
	and ${\mathbf P}L^{(2)+}{\mathbf P} = 0$. Assume ${\mathbf P} f=f$. Then, in the limit as $\eps\rightarrow 0$, $w^{\eps}(s,x,y|t)$ tends to $w^{(0)}(s,x|t)$ 
	for $-T<s\leq t$, $T<\infty$, uniformly in $x$ and $y$ on compact sets, 
	where $w^{(0)}$  satisfies ${\mathbf P}w^{(0)}=0$ and solves the backward equation 
	$$ -\partial_s w^{(0)} = {\hat L}^+w^{(0)},\qquad w^{(0)}(t,x|t)=f(x)$$
	with 
	${\hat L}^+ 
	= {\mathbf P}\,(L^{(2)+}){\mathbf P} - 
	{\mathbf P} L^{(1)+}(L^{(0)+})^{-1} L^{(1)+}{\mathbf P}$.
\end{theorem}
\medskip

Of course, we return to the forward equation. If we use this reasoning for our SIQS model, we find that the long-term projector ${\mathbf P}$ 
``lives'' on the one-dimensional line, and the resulting approximate Fokker-Planck equation is one-dimensional. For the full 2-strain SIQS model, 
we limit ourselves to the case $\gamma_1=\gamma_2$ and $\hat\mu=0$ to reduce the computational complexity.

\begin{theorem} Consider the van Kampen 
expansion of the 2-strain SIQS model for $\gamma_1=\gamma_2=\gamma$ and $\hat\mu=0$.  
With rescaled time  $\tau  = \gamma \varepsilon^2 t/r$ the one-dimensional approximate Fokker-Planck equation reads 
\begin{equation}
    \begin{aligned}\label{FP1dim}
     \partial_{\tau} f(X,\tau) = -\partial_X v(X) f(X,\tau) + \frac{1}{2} \partial_{X}^2 D(X) f(X,\tau), \qquad X\in[0,1]
    \end{aligned}
\end{equation}
with effective drift and diffusion terms are given by (recall $\zeta_i = \hat\omega_i/(\hat c_i+\hat\omega_i)$, 
and introduce $\omega_i = \hat\omega_i/(\beta r)$)
\begin{align}\label{Drift1dim}
\alpha(X) &= 1-\zeta_2+(\zeta_2-\zeta_1)X\\
    v(X) &= -(1-\zeta_1)(1-\zeta_2)\frac{\zeta_1\zeta_2 (\omega_2-\omega_1) \alpha(X)
    +\omega_1\omega_2(\zeta_2-\zeta_1) }{ \alpha(X)^2 \left[\omega_1+\zeta_1 \alpha(X) \right]\left[\omega_2+\zeta_2 \alpha(X) \right]}  X(1-X)\label{siqsDtift}\\
    D(X) &= \frac{2(1-\zeta_1)(1-\zeta_2)}{ \alpha(X)}X(1-X).\label{siqsNoise}
\end{align}
\end{theorem}

Also, the SIS model exhibits the crucial coexistence line of stationary points - as the parameters for the SIQS model are chosen in such a 
way that no strain has an advantage in the deterministic limit, this also is the case for the SIS model we did derive. As a consequence, 
we precisely are in the setup of Kogan's paper~\cite{Kogan.2014} which uses a two-strain SIS model to demonstrate his method. We directly
can use his result (SI, Section 3.2). 
Note that we adapt the rate constant to the present notation of the SIS model above and that we rescale the variable $X$ (as defined in~\cite{Kogan.2014}) such that $X\in[0,1]$. In the next theorem, we allow for $\gamma_1\not=\gamma_2$ and $\hat\mu\geq 0$. 

\begin{theorem} \label{approxSIS}
The one-dimensional approximation of the two-strain SIS SDE model is given by 
\begin{equation}
    \begin{aligned}\label{FP1dim2strain}
     \partial_{\tau} f(X,\tau) = -\partial_X\bigg( v_\text{SIS}(X) f(X,\tau)\bigg) + \frac{1}{2} \partial_{X}^2 \bigg(D_\text{SIS}(X) f(X,\tau)\bigg), \qquad X\in[0,1]
    \end{aligned}
\end{equation}
where $\tau = \frac{(1-\zeta_1)\,\gamma_1}{r}\eps^2t$ and with effective drift and diffusion terms are given by 
    \begin{align}
	a &= \frac{\gamma_2(1-\zeta_2) + \hat\mu\zeta_2}{\gamma_1(1-\zeta_1) + \hat\mu\zeta_1}\\
        v_\text{SIS}(X) &=-\frac{\,a\,(1-a)}{(a+(1-a)X)^2}\,X(1-X)\\
     D_\text{SIS}(X) &= \frac{2\,a}{a+(1-a)X}\,X(1-X)
    \end{align}
\end{theorem} 

\begin{rem}
Let  $\hat\omega_i,\hat c_i\rightarrow\infty$ while $\hat\omega_i/(\hat\omega_i+\hat c_i)\equiv \zeta_i$, 
 $\hat\mu=0$ and $\gamma_1=\gamma_2$, s.t.\ $a=(1-\zeta_2)/(1-\zeta_1)$. Then, 
\begin{eqnarray*}
D(X) &=& \frac{2(1-\zeta_1)(1-\zeta_2)}{ \alpha(X)}X(1-X) 
=  \frac{2(1-\zeta_1)(1-\zeta_2)}{1-\zeta_2+(\zeta_2-\zeta_1)X}X(1-X) = (1-\zeta_1) D_{SIS}(X)
\end{eqnarray*}
and in the limit 
\begin{eqnarray*}
 v(X) &=& -(1-\zeta_1)(1-\zeta_2)\frac{\zeta_1\zeta_2 (\omega_2-\omega_1) \alpha(X)
    +\omega_1\omega_2(\zeta_2-\zeta_1) }{ \alpha(X)^2 \left[\omega_1+\zeta_1 \alpha(X) \right]\left[\omega_2+\zeta_2 \alpha(X) \right]}  X(1-X)\\
	&\rightarrow & 
-(1-\zeta_1)(1-\zeta_2)\frac{(\zeta_2-\zeta_1) }{ \alpha(X)^2 }  X(1-X) = (1-\zeta_1) v_{SIS}(X)
\end{eqnarray*}
The factor $1-\zeta_1$ is due to the different time scaling in Theorem~\ref{approxSIQS} and Theorem~\ref{approxSIS}, s.t.\ 
indeed the rapid switching limit transfers the SIQS approximate Fokker-Planck equation into the approximate SIS Fokker-Planck equation.
\end{rem}

\subsection{Effective dynamics along the Coexistence Line: Fixation probabilities}
\label{fixProbab}
\label{sec: Dynamics at CL}
We first discuss the sign of the drift term in equation~\eqref{siqsDtift} and then proceed to the computation of the absorption probabilities.

\begin{figure}[t!]
    \begin{subfigure}[b]{0.3\textwidth}
        \includegraphics[width=\textwidth]{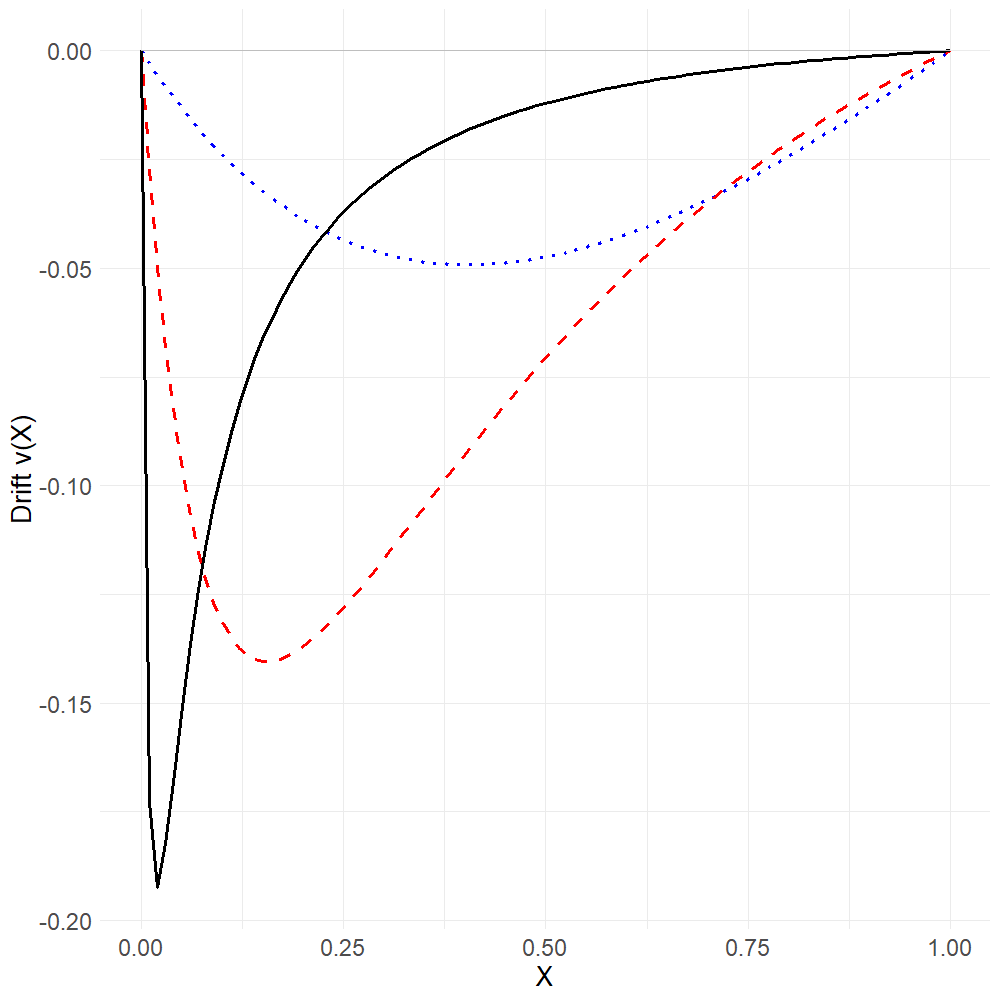}
        \caption{Parameters: $\omega_1=\omega_2=2$, $\zeta_1=0.2$ and $\zeta_2= 0.44,0.84, 0.984$}
        \label{fig:Drift zeta}
    \end{subfigure}
    \hfill
    \begin{subfigure}[b]{0.3\textwidth}
        \includegraphics[width=\textwidth]{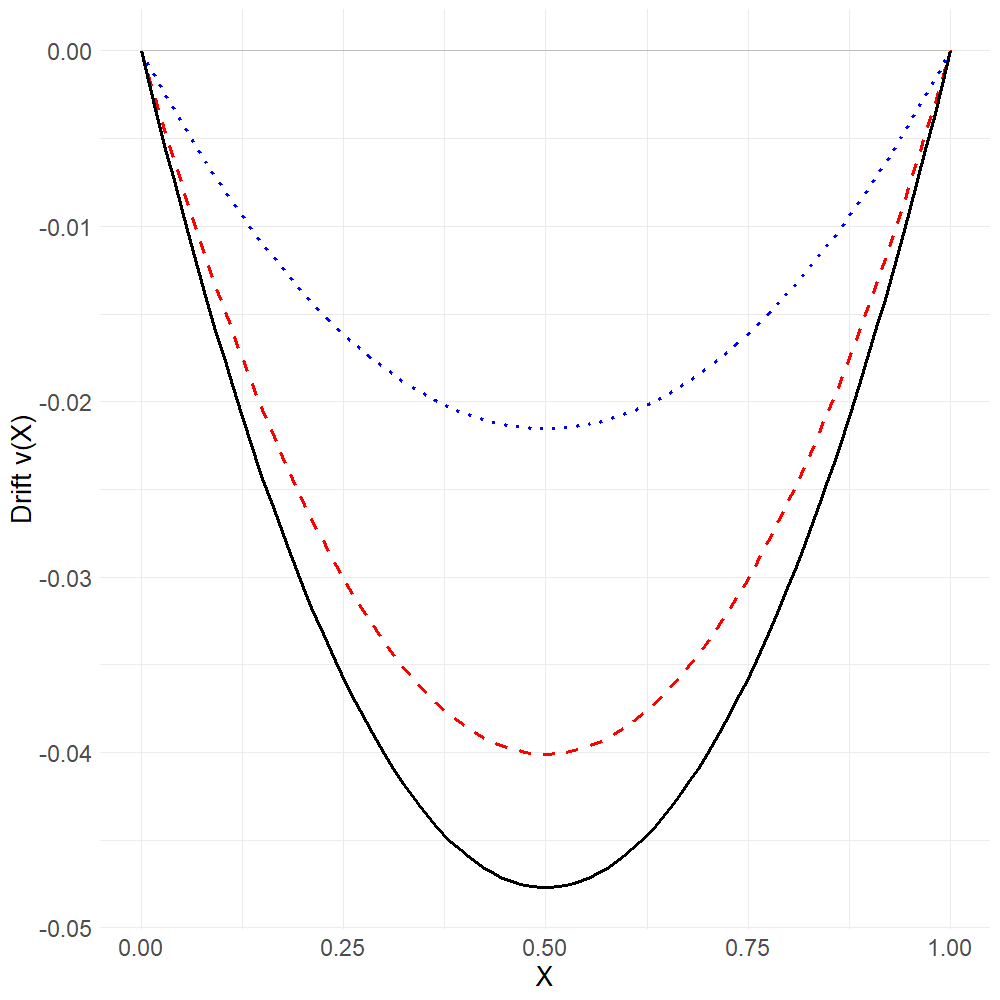}
        \caption{Parameters: $\zeta_1=\zeta_2=0.75$, $\omega_1=0.5$ and $\omega_2=1,3,10$}
        \label{fig:Drift omega}
    \end{subfigure}
    \hfill
    \begin{subfigure}[b]{0.3\textwidth}
        \includegraphics[width=\textwidth]{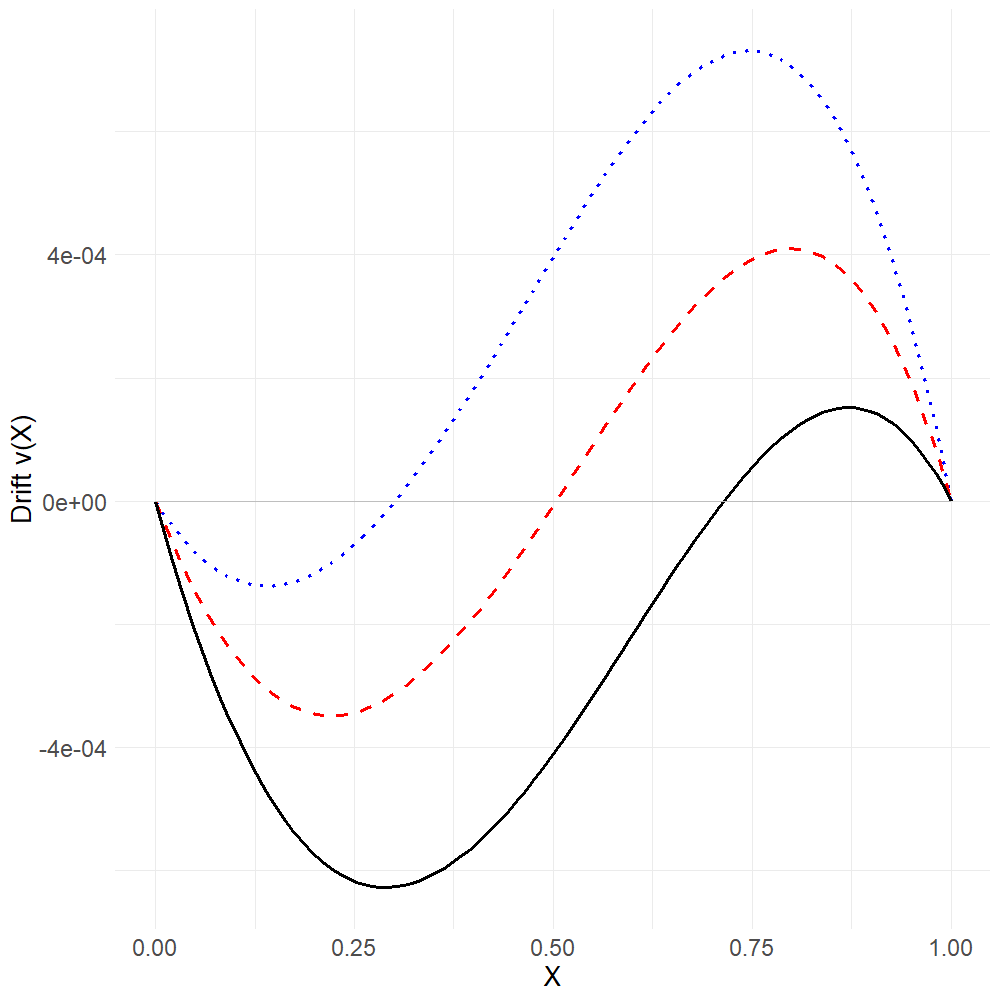}
        \caption{Parameters: $\omega_1=1, \omega_2=3,\zeta_1=0.5$\\ and $\zeta_2=0.422,0.424,0.426$}
        \label{fig: Drift disruptive}
    \end{subfigure}
    \caption{Drift $v(X)$ from (\ref{Drift1dim}). (a) for $\omega_1=\omega_2$,
    	the larger $\zeta$ has an advantage (b) for $\zeta_1=\zeta_2$, the larger $\omega$ has an advantage (c) disruptive dynamics in case of $\zeta_1>\zeta_2$ and $\omega_1<\omega_2$. }
    \label{fig:Drift_SIQS}
\end{figure}

Our ultimate aim is to understand how quiescence develops under evolutionary forces. Thereto we investigate the competition of different quiescent traits. 
A quiescent trait can be characterized by two different aspects: How often an individual becomes quiescent ($\omega$) and how long the individual stays quiescent ($\zeta$). The drift term $v(X)$ in eqn.~\eqref{siqsDtift} nicely splits up into two parts induced by these different aspects  ($\omega_2-\omega_1$ resp.\ $\zeta_2-\zeta_1$). In Fig.~\ref{fig:Drift_SIQS}, the drift term is plotted in different scenarios. For $1-\zeta_1 \leq \alpha(X) \leq 1-\zeta_2$ the sign of the drift is determined by the differences $\omega_2-\omega_1$ and $\zeta_2-\zeta_1$ (see Prop.~\ref{Prop: selectiontypes} below). If $\omega_1=\omega_2$ is fixed, then the drift is positive for $\zeta_2>\zeta_1$ and conversely, if $\zeta_1=\zeta_2$ is fixed, then the drift is positive for $\omega_2>\omega_1$. This means that strains with higher values in $\omega$ (higher jump frequency between active and inactive state) and higher values in $\zeta$ (mean proportion of time inactive) gain a competitive advantage. Herein, we already find a first fundamental result: More quiescence (higher $\omega$, higher $\zeta$) is favorable. \\
If $(\omega_2-\omega_1)$ and $(\zeta_2-\zeta_1)$ have different signs, the drift term may also have a sign conversion. The direction of selection forces depends on the state of the system. By the next proposition, the drift is always directed to either of the boundaries of the CL. The model can exhibit weak directional or disruptive selection but not balancing selection (Proof: SI, Section~4.1).

\begin{prop}\label{Prop: selectiontypes}
	Assume $\zeta_2 > \zeta_1$. If and only if $(1-\zeta_2) < -\frac{(\zeta_2-\zeta_1)\omega_1\omega_2}{\zeta_1\zeta_2(\omega_2-\omega_1)} < (1-\zeta_1)$, there exists $X_0 \in (0,1)$ with $v(X_0)=0$. For any such $X_0$ it is $v'(X_0)>0$. In this sense, the model may exhibit directional or disruptive selection but not balancing selection. 
\end{prop}

\begin{figure}[htbp]
	\begin{subfigure}[b]{0.45\textwidth}
		\includegraphics[width=\textwidth]{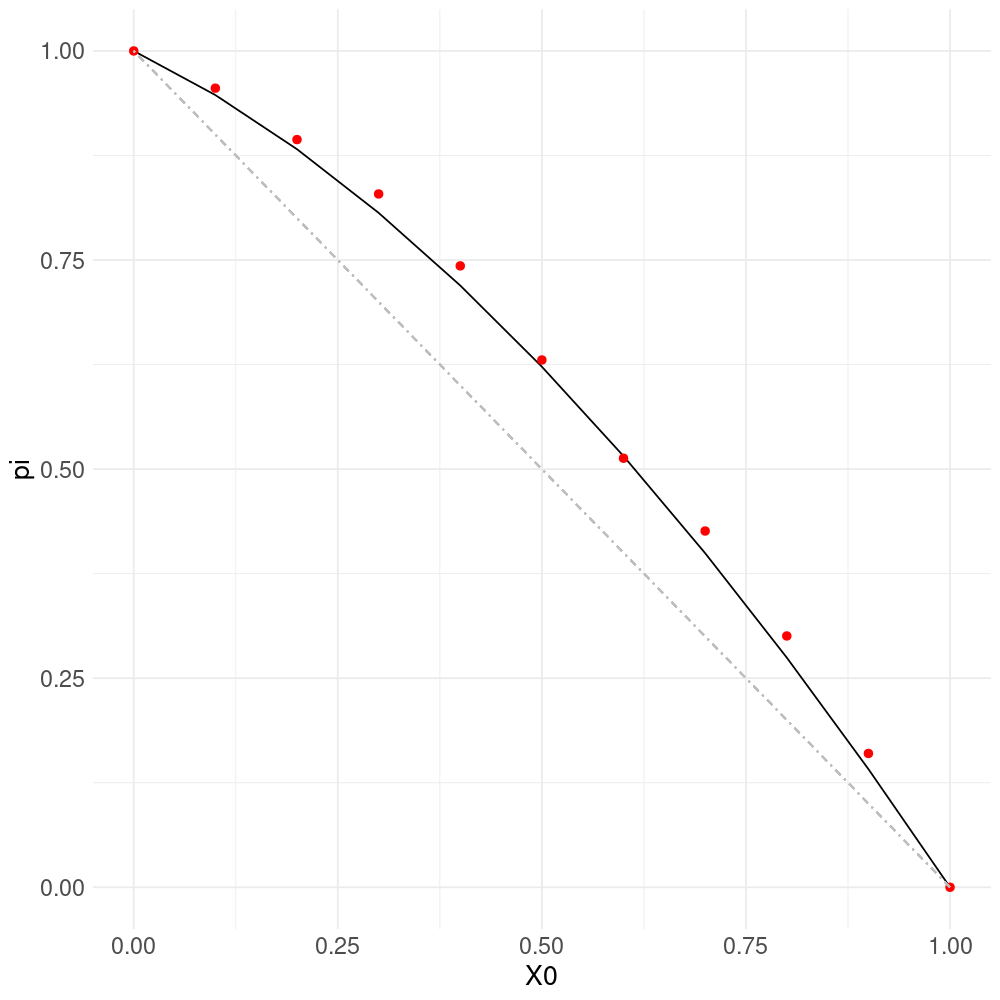}
		\caption{Parameters: $\omega_1=\omega_2=2$, $\zeta_1=0.2$, $\zeta_2=0.8$}
		\label{fig: Pi zeta empiric}
	\end{subfigure}
	\hfill
	\begin{subfigure}[b]{0.45\textwidth}
		\includegraphics[width=\textwidth]{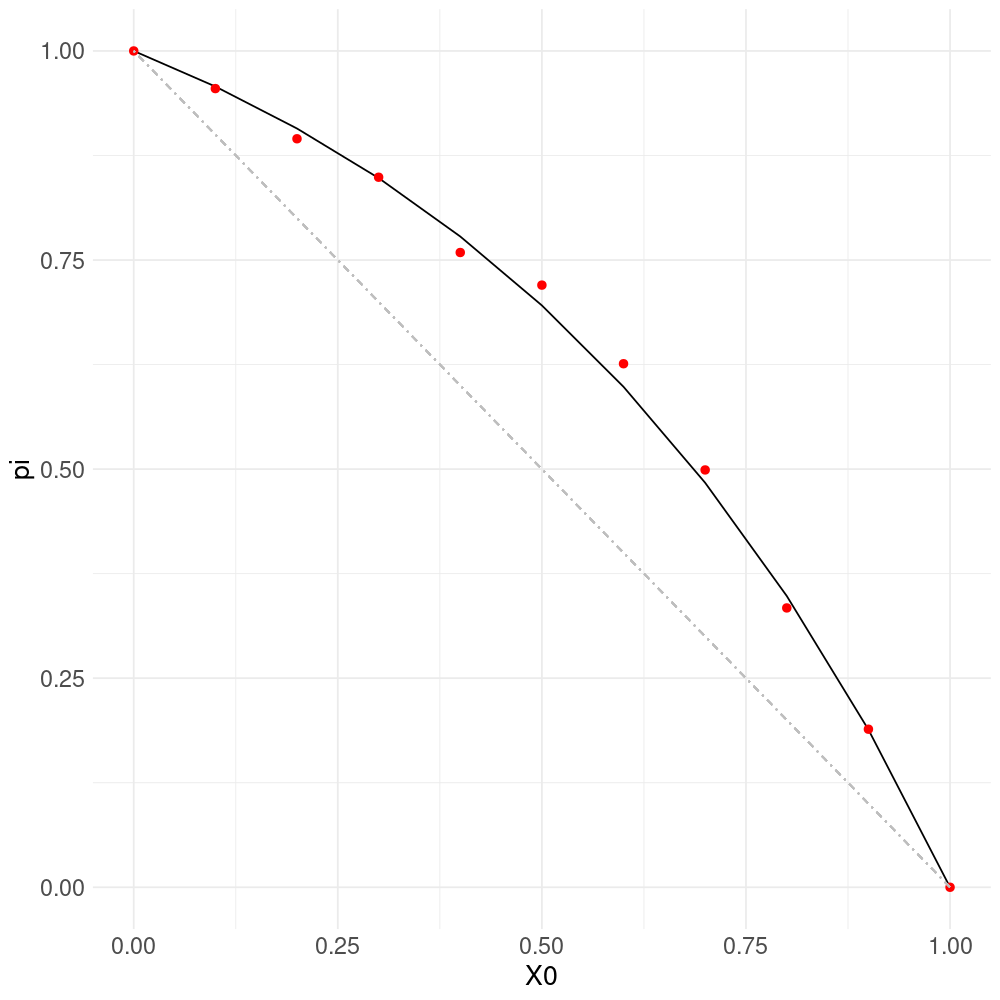}
		\caption{Parameters: $\omega_1=0.2, \omega_2=5$, $\zeta_1=\zeta_2=0.8$.}
		\label{fig: Pi w empiric}
	\end{subfigure}
	\caption{Absorption probabilities; Solid lines: Theory based on the approximative one-dimensional SIQS Fokker-Planck equation, dots:  $10^3$ realizations of the individual-based $SIQS$-model with $N=250$ using Monte-Carlo-simulation. }
	\label{fig: Pi empiric}
\end{figure}

The Fokker-Planck-equation (\ref{FP1dim})  describes a one-dimensional Markov process $(X_\tau)_{\tau \geq 0}$ in terms of it's probability density function $f(X,\tau)$. It has absorbing boundaries at $X=0$ and $X=1$, corresponding to fixation of either the fast first strain $I_1$ (if $X=1$) or slow second strain $I_2$ (if $X=0$). Within a finite time, it will a.s.\ reach one of these boundary points. We aim at the probability of fixation of each strain. 
It is well-known how to compute the absorption probability~\cite[Chapter 5.5.4]{Gardiner.2009}. Let $\pi:=\mathrm{P}(\lim_{t \to \infty} X(t)=0|X(0)=X_0)=\mathrm{P}(\text{"Strain 2 reaches fixation"}|X(0)=X_0)$ be the probability that, starting on the CL, the second strain fixates. It obeys
\begin{align*}
	v(X) \pi'(X) + \frac{1}{2}D(X)\pi''(X)=0.
\end{align*}
Under the boundary conditions $\pi(0)=1$, $\pi(1)=0$ and starting at $X_0$ the solution is
\begin{align}\label{Formulafixationprob}
	\pi(X_0)= \frac{\int_{X_0}^{1} \nu(y)dy}{\int_0^1 \nu(y) dy}, \quad \text{where } \nu(y) = \exp \left(- 2\int_0^y \frac{v(z)}{D(z)}dz \right).
\end{align}
In case the two strains only differ in $\zeta$ or in $\omega$ and not in both, the integrals can be explicitly evaluated, and in that, fixation probability can be explicitly calculated.

\begin{prop}\label{prop: fix probs SIQS}
\begin{enumerate}[label=(\alph*)]
\item If $\zeta_1=\zeta_2$ but $\omega_1\not=\omega_2$, the probability that the second strain, starting at $X_0 \in (0,1)$, outcompetes the first strain 
reads
\begin{align*}
    \pi(X_0)= \frac{\mathrm{e}^{\mathcal{A}}-\mathrm{e}^{\mathcal{A}X_0}}{\mathrm{e}^{\mathcal{A}}-1},
\end{align*}
where $\mathcal{A}=\frac{ \zeta_1^2 (\omega_2-\omega_1)}{\left( \omega_1+(1-\zeta_1)\zeta_1 \right)\left(\omega_2+(1-\zeta_1)\zeta_1\right)}$. In particular $\pi(X_0)>1-X_0 \iff \omega_2>\omega_1$

\item If $\omega_1=\omega_2$ but $\zeta_1 \neq \zeta_2$, then 
\begin{align*}
    \pi(X_0)&= \frac{\mathfrak{A}(1-\zeta_1) - \mathfrak{A}(\alpha(X_0))}{\mathfrak{A}(1-\zeta_1) - \mathfrak{A}(1-\zeta_2)},
\end{align*}
where 
\begin{align*}
    \mathfrak{A}:[1-\zeta_2,1-\zeta_1] &\to \mathbb{R}_{>0}:
    s \mapsto (\omega_1 + \zeta_1 s)^{\frac{-\zeta_2}{\zeta_2-\zeta_1}} (\omega_1 + \zeta_2 s)^{\frac{\zeta_1}{\zeta_2-\zeta_1}}.
\end{align*}
$\mathfrak{A}$ is concave if $\zeta_2 > \zeta_1$ and convex if $\zeta_1 > \zeta_2$.
\end{enumerate}
\end{prop}
A proof is given in SI, Section 4.2. 
The result is nicely in line with simulations of the individual-based SIQS process (see Fig.~\ref{fig: Pi empiric}).

\begin{rem}
(a) In the neutral case ($\omega_1=\omega_2$, $\zeta_1=\zeta_2$), the offspring of any randomly chosen individual has the same probability to eventually take over the population. Thus, $\pi(X_0)=1-X_0$.\\
(b) The theorem indicates that $\pi(X_0)>1-X_0$ for $X_0\in(0,1)$ if either 
$\omega_2>\omega_1$ (and $\zeta_1=\zeta_2$) or  $\zeta_2>\zeta_1$ (and $\omega_1=\omega_2$). More quiescence is an advantage.\\
(c) In case of $\omega_1>\omega_2$ and $\zeta_1<\zeta_2$ such that there is $X\in(0,1)$ with $v(X)=0$ (the drift term changes sign), it depends on the 
initial value $X_0$ if $\pi(X_0)$ is larger or below the neutral case ($\pi(X_0)=1-X_0$).
\end{rem} 

For the sake of completeness, we also state the fixation probability for the SIS model (Theorem~\ref{approxSIS}), also see~\cite{Kogan.2014}. 

\begin{prop}\label{prop: fix prob SIS}
	In the $SIS$-model, the fixation probability of strain 2 is
	\begin{align}\label{eq: Pi Kogan}
		\pi_{SIS}(X_0) = (1-X_0) \left(1+\frac{1-a}{1+a}X_0\right).
	\end{align}
\end{prop}
\begin{proof}
	By standard integration $\nu(y)=1+\frac{1-a}{a}y$ and 
	\begin{align*}
		\frac{\int_{X_0}^{1} \nu(y)dy}{\int_0^1 \nu(y) dy}=\frac{a(1-X_0)^2+1-X_0^2}{1+a}\frac{1+a}{2a}=(1-X_0) \left(1+\frac{1-a}{1+a}X_0\right)
	\end{align*} after some manipulation.
\end{proof}
As expected, the fixation probability scales with the initial proportion of strain 2, $1-X_0$. If $X_0=1/2$, then both strains start with the same amount of infected.

\section{Adaptive Dynamic in case of weak selection}
We will now broaden our perspective and consider the long-term evolution of the traits within the framework of Adaptive Dynamic \cite{geritz1997dynamics,diekmann2004beginner,Muller.2013}. The idea behind Adaptive Dynamics is to consider one strain $I_1$, called the resident, in its endemic equilibrium. The traits of the strain, $\beta, \gamma, \omega_1, \zeta_1$, define the properties of the environment, that is, the number of susceptible, infected, quiescent. Once in a while, a rare mutant strain $I_2$ (particularly with slightly different timing in quiescence) tries to invade the equilibrium population. If this strain is able to displace the resident, it becomes the new resident, and evolution made a (tiny) step.\\
Herein, we cannot use the standard setting of Adaptive Dynamic. Usually, it is assumed that either the resident or the mutant is superior, and the inferior strain always is outcompeted. Only the parameters define which strain is superior and in that, outcompetes the other strain. No randomness is involved, such that a deterministic pairwise invasibility plot (PIP) can be constructed.\\
This is different in the current situation, where we only can talk about invasion probabilities. In that, we make three primary assumptions that parallel those used by classical Adaptive Dynamics. First, mutants are rare: We start off far on the right of the CL; $X_0$ close to $1$. Second, the mutant has parameters close to the resident and will either replace the resident strain or become extinct. This assumption is naturally fulfilled in our process on the CL. In the previous section, we explored the probability of a successful invasion, that is, fixation of the mutant strain $2$. Of course, the probability of a successful invasion scales with the initial frequency of the mutant $1-X_0$; nonetheless, we obtain non-trivial results when calculating stationary distributions for the limit $X_0\to1$. The third key assumption is that the time scale at which mutants reach fixation (or extinction) is much faster than the time scale at which additional mutations occur. This assumption ensures that we remain within the framework of two-strain competition, without the presence of three or more strains simultaneously. Importantly, here we are dealing not with deterministic invasions as is typical in adaptive dynamics, but with probabilities for a successful invasion. This leads to a directed random walk in the trait space \cite{Dieckmann.1996}. The proper ESS in classical Adaptive Dynamics is replaced by an invariant measure of the random walk on the trait space. \par\medskip
In what follows, we focus separately on each of the traits $\omega$ and $\zeta$ and begin with a compact discretization of the state space. Next, we will seek stationary distributions for the resulting directed random walk and then refine the discretization to move towards a continuum limit probability density function.

\subsection{Adaptive Dynamic for $\omega$ in the SIQS model - analytic results}
Fix $\zeta_1=\zeta_2=:\zeta \in (0,1)$ and consider only $\omega$.  Define a discrete-time Markov-chain $(W_t)_{t\in \mathbb{N}}$ for the evolution of $\omega$ in the following way: 
For some $M>0$, $n\in\mathbb{N}$, $\delta=M/n$, the trait space is $\{\omega_0=0$,  
$\omega_1=\delta$, $\omega_2=2\delta,\ldots, \omega_{n-1}= (n-1)\delta$, $\omega_n=M \}$ and the transition probabilities are
\begin{align*}
    P(W_{t+1} = \omega + \delta | W_t=\omega) &= \frac{1}{2}\pi_i(X_0)\\
    P(W_{t+1} = \omega - \delta | W_t=\omega) &= \frac{1}{2}\Tilde{\pi}_i(X_0)\\
    P(W_{t+1} = \omega | W_t=\omega) &= 1-\frac{1}{2}\left(\pi_i(X_0)+\Tilde{\pi}_{i}(X_0) \right).
\end{align*}
$\pi_i(X_0)$ is the probability of a successful invasion (= fixation of strain 2) of a strain with higher $\omega$-value, given by Prop.~\ref{prop: fix probs SIQS} with parameters $\omega_1=\omega_i$, $\omega_2=\omega_i+\delta$. Likewise, $\Tilde{\pi}_i$ is the probability of a successful invasion of a strain with lower $\omega$-value, namely with parameters $\omega_1=\omega_i$, $\omega_2=\omega_i-\delta$. In Section~\ref{sec: Dynamics at CL} we saw that $\Tilde{\pi}_i(X_0)\leq\pi_i(X_0)$. Due to symmetry reasons $\pi_{i-1}(X_0)=1-\Tilde{\pi}_{i}(1-X_0)$ holds. The Markov chain is irreducible on a finite state space and thus has a unique stationary distribution $W^\ast=(W_0^\ast,...,W_M^\ast)$ characterized by the detailed balance equation, resulting in 
\begin{align*}
    W_i^\ast=\frac{\pi_{i-1}(X_0)}{\Tilde{\pi}_{i}(X_0)}\cdot...\cdot\frac{\pi_{0}(X_0)}{\Tilde{\pi}_{1}(X_0)} W_0^\ast,
\end{align*}
where $W_0^\ast$ has to be chosen s.t. $\| W^\ast\|_1=1$. 
In fact, using Proposition \ref{prop: fix probs SIQS},
\begin{align*}
    \frac{\pi_{i-1}(X_0)}{\Tilde{\pi}_{i}(X_0)} &= \frac{\pi_{i-1}(X_0)}{1-\pi_{i-1}(1-X_0)}
    = \mathrm{e}^{\mathcal{A}(\omega_{i-1},\omega_i)X_0} \\
    &= \exp\left( \frac{ \zeta^2 (\omega_i-\omega_{i-1})}{\left( \omega_{i-1}+(1-\zeta)\zeta \right)\left(\omega_i+(1-\zeta)\zeta\right)} X_0 \right).
\end{align*}
Therewith
\begin{align*}
    W_i^\ast&=\frac{\pi_{i-1}(X_0)}{\Tilde{\pi}_{i}(X_0)}\cdot...\cdot\frac{\pi_{0}(X_0)}{\Tilde{\pi}_{1}(X_0)} W_0^\ast
    =\exp\left( \frac{ \zeta^2 (\omega_i-\omega_{0})}{\left( \omega_{i}+(1-\zeta)\zeta \right)\left(\omega_0+(1-\zeta)\zeta\right)} X_0 \right)W_0^\ast\\
    &=\exp\left( \frac{\zeta}{1-\zeta} \frac{ \omega_i}{ \omega_{i}+(1-\zeta)\zeta } X_0 \right)W_0^\ast.
\end{align*}
Not only gives this easy access to the limit $X_0 \to 1$ (number of mutants approaches 0) but it shows that the stationary distribution is - up to scaling - actually independent of the chosen grid-size $\delta$. This result stems from the specific form of $\frac{v(X)}{D(X)}$, which forms a telescopic sum when adding up the values from the discretized trait space. It does not hold in general for Markov processes constructed in this way from a Fokker-Planck equation. Embedding the discrete stationary distribution $W^\ast$ in $[0,M]$ and letting $X_0 \to 1$ and $n \to \infty$,  it is approximating a continuous probability density function $p_{W}$, the so-called continuum limit:
\begin{figure}[tb]
    \begin{subfigure}[b]{0.45\textwidth}
        \includegraphics[width=\textwidth]{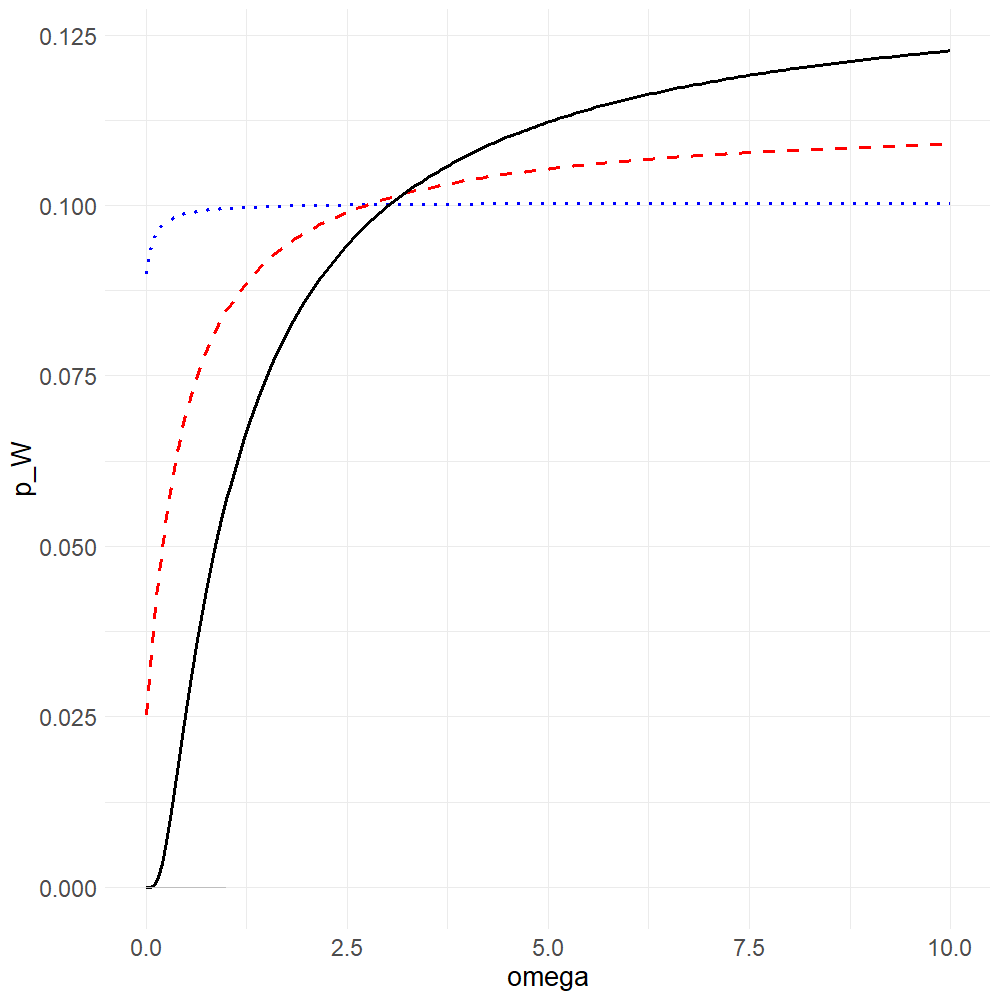}
        \caption{Continuum limit $p_W$  for stationary distributions of $W$. Parameters are $\zeta=0.2,0.6,0.95$ (dotted, dashed, solid) and $M=10$}
        \label{fig: RAD omega}
    \end{subfigure}
    \hfill
    \begin{subfigure}[b]{0.45\textwidth}
        \includegraphics[width=\textwidth]{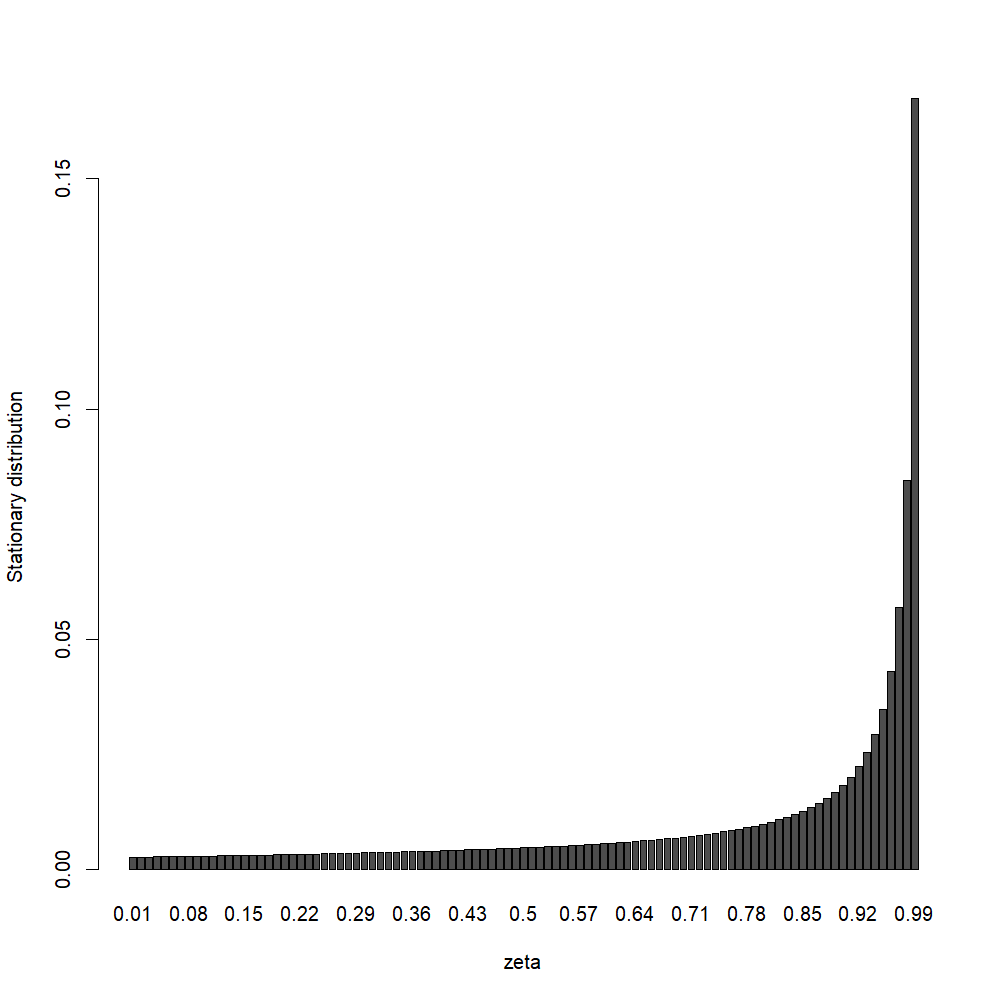}
        \caption{Discrete stationary distribution $Z^\ast$ of $\zeta$ for $n=100, \omega=2$. $Z^\ast \xrightarrow{d}\delta_1$ for $n\to\infty$.}
        \label{fig: RAD zeta}
    \end{subfigure}
    \caption{Invariant measures on the trait space for $\omega$ resp.\ $\zeta$ in case of the SIQS model.}\label{invMeasSIQS}
\end{figure}

\begin{align*}
    p_W(\omega) &= \lim_{\delta \to 0} P(W^\ast = \omega) = \exp \left( \frac{\zeta}{1-\zeta} \frac{ \omega}{ \omega+(1-\zeta)\zeta }\right) C(M)^{-1} \\
    \text{and } C(M) :&=  \int_0^M \exp\left( \frac{\zeta}{1-\zeta} \frac{ s}{ s+(1-\zeta)\zeta } \right)ds.
\end{align*}
In Fig.~\ref{fig: RAD omega}, the $\zeta$-dependence of the limiting distribution $p_W$ is displayed. When $\zeta$ is small, the distribution is nearly uniform, whereas for larger values of $\zeta$, a bias for larger values of $\omega$ can be observed. This is intuitive: When the time in $Q$ increases, the frequency of jumps becomes more important.
The cut-off at $\omega=M=10$ is arbitrary; larger values of $\omega$ are always favored. However, the marginal benefit is diminishing rapidly. Whilst very low values of $\omega$ are improbable to survive, a further increase when $\omega$ is already at a moderate value gains less and less advantage. This finding can be understood based on the rapid-switching approximation theorem~\ref{rapidSwitchTheo} which shows that the SIQS model can be well approximated by an SIS model with rescaled rate constants if $\omega_i$ and $c_i$ are large; in these rescaled rate constants, 
not $\omega_i$ and $c_i$ but only $\zeta_i=\omega_i/(\omega_i+c_i)$ enter. In that, the specific value of $\omega$ becomes  less and less important if 
$\omega$ becomes large while $\zeta$ is kept fix.

\subsection{Interlude: Adaptive Dynamic for $\gamma$ in the SIS model}
As a preparation of the considerations about the evolution 
of the trait $\zeta$ for 
the SIQS model in Sect.~\ref{SIQSzetaAD} we consider a simple SIS model with infection rate $\beta$ and recovery rate $\gamma$. We fix $\mathcal{R}_0=\beta/\gamma$. For the competition of two strains (with rate constants $\beta_i$, $\gamma_i$) Theorem~\ref{approxSIS} holds true where $a=\gamma_1/\gamma_2$. Therewith, the fixation probability $\pi_{SIS}(X)$ is given by Prop.~\ref{prop: fix prob SIS}. \par\medskip 
With these preliminaries, we consider the evolution of $\gamma$ as a Markov process $(A_t)_{t\in\mathbb{N}}$. For $n\in\mathbb{N}, \delta=1/n$, the trait space is $\{ \gamma_0, \gamma_1=(1-\delta)\gamma_0,\gamma_2=(1-2\delta)\gamma_0,...,\gamma_{n-1}=(1-(n-1)\delta)\gamma_0 \}$ and the transitions are 
\begin{align*}
        P(A_{t+1} = \gamma + \delta | A_t=\gamma) &= \frac{1}{2}\pi_i(X_0)\\
    P(A_{t+1} = \gamma - \delta | A_t=\gamma) &= \frac{1}{2}\Tilde{\pi}_i(X_0)\\
    P(A_{t+1} = \gamma | A_t=\gamma) &= 1-\frac{1}{2}\left(\pi_i(X_0)+\Tilde{\pi}_{i}(X_0) \right),
\end{align*}
with the fixation probabilities for the SIS model  $\pi_i(X_0)=\pi_{SIS}(X_0)$ with $a_i=\gamma_{i+1}/\gamma_i$ and $\Tilde{\pi}_i(X_0)$ with $\Tilde{a}_i=\gamma_{i-1}/\gamma_i$. We obtain the stationary distribution $A^\ast=(A_0^\ast,...,A_{n-1}^\ast)$ for $X_0\to1$
\begin{align*}
        \frac{\pi_{i-1}(X_0)}{\Tilde{\pi}_{i}(X_0)} &= \frac{\pi_{i-1}(X_0)}{1-\pi_{i-1}(1-X_0)}
    =\frac{1+a_{i-1}+(1-a_{i-1})X_0}{1+a_{i-1}-(1-a_{i-1})X_0}  
    \xrightarrow{X_0\to1} \frac{1}{a_{i-1}}=\frac{\gamma_{i-1}}{\gamma_{i}}\\
    &\implies A_i^\ast = \frac{\gamma_0}{\gamma_i}A_0^\ast = \frac{1}{1-i/n}A_0^\ast
\end{align*}
with $A_0^\ast=\left( \sum_{j=0}^{n-1}\frac{1}{1-j/n} \right)^{-1}$. We again embed the discrete distribution $A^\ast$ in $[0,1]$ by $A^\ast([0,\gamma_{n-1}])=A^\ast_{n-1}$, $A^\ast( (\gamma_{i+1},\gamma_{i}] )=A^\ast_i$ for $0\leq i \leq n-2$ and $A^\ast((\gamma_0,1])=0$. Denoting both distributions as $A^\ast$ is a slight abuse of notation as they are strictly speaking different objects. Let $F_{A^\ast}:\mathbb{R}\to[0,1]$ be the cdf of the embedded distribution. Then for $0<y\leq1$
\begin{align*}
    1-F_{A^\ast}(y) = A^\ast((y,\infty))=\frac{\sum_{j=0}^{\lfloor (1-y)n \rfloor} A^\ast_j}{\sum_{j=0}^{n-1} A^\ast_j} = \frac{\sum_{j=0}^{\lfloor (1-y)n \rfloor} \frac{1}{n-j}}{\sum_{j=0}^{n-1} \frac{1}{n-j}}
    =\frac{\sum_{j=\lfloor ny \rfloor}^{ n } \frac{1}{j}}{\sum_{j=1}^{n} \frac{1}{j}}
    \xrightarrow{n\to\infty}0
\end{align*}
where the limit holds since $\sum_{j=\lfloor ny \rfloor}^{ n } \frac{1}{j}<\lfloor y^{-1}\rfloor-1$ and the harmonic series diverges. We can infer that the cdf $F_{A^\ast}$ is weakly convergent to $\textbf{1}_{\{0\}}$ and thus by the Helly-Bray theorem the stationary distribution is converging in distribution to the Dirac-measure at $0$:
\begin{align*}
    A^\ast\xrightarrow{d}\delta_0.
\end{align*}
Evolution favours the slower strain. Of course, $\gamma=0$ does not give a well-defined $ SIS$-process; this limit is purely symbolic. As with $\gamma\rightarrow 0$ the time scale becomes slower and slower, also the time till fixation becomes longer and longer. In that, the trait $\gamma$ slowly approximate zero, but never reaches zero. Anyhow, for very slow strains, most likely further mechanisms which are neglected by the simple SIS model gain importance. For $\gamma$ very small, the SIS model might not be appropriate any more.

\subsection{Adaptive Dynamic for $\zeta$ in the SIQS model - formal results}
\label{SIQSzetaAD}

In the present section we focus on $\zeta$, and fix $\omega_1=\omega_2=:\omega>0$. 
We find again, as $\gamma$ in the SIS model, 
also the invariant distribution of $\zeta$ will degenerate to a Dirac-measure.\\
Consider a discrete-time Markov-chain $(Z_t)_{t\in\mathbb{N}}$, defined on the trait space $\{\zeta_0=0 , \zeta_1=\delta, ..., \zeta_{n-1}= (n-1)\delta \}$, where $\delta=1/n$,  with transitions
\begin{align*}
    P(Z_{t+1} = \zeta + \delta | Z_t=\zeta) &= \frac{1}{2}\pi_i(X_0)\\
    P(Z_{t+1} = \zeta - \delta | Z_t=\zeta) &= \frac{1}{2}\Tilde{\pi}_i(X_0)\\
    P(Z_{t+1} = \zeta | Z_t=\zeta) &= 1-\frac{1}{2}\left(\pi_i(X_0)+\Tilde{\pi}_{i}(X_0) \right).
\end{align*}
The fixation probabilities $\pi_i(X_0)$ of a strain with higher $\zeta$-value are  given in Prop.~\ref{prop: fix probs SIQS}. The fixation probability of a strain with lower $\zeta$-value is, by symmetry, $\Tilde{\pi}_i(X_0)=1-\pi_{i-1}(1-X_0)$ and Section~\ref{sec: Dynamics at CL} revealed that $\Tilde{\pi}_i(X_0)\leq\pi_{i-1}(1-X_0)$. 

We aim at a stationary distribution and proceed in the same way as before. Using the l'Hopital rule on the explicit form of $\pi_i(X_0)$ yields
\begin{align*}
    \frac{\pi_{i-1}(X_0)}{\Tilde{\pi}_{i}(X_0)} &= \frac{\pi_{i-1}(X_0)}{1-\pi_{i-1}(1-X_0)}
    = \frac{\mathfrak{A}(1-\zeta_{i-1}) - \mathfrak{A}(\alpha(X_0))}{\mathfrak{A}(\alpha(1-X_0)) - \mathfrak{A}(1-\zeta_i)}\\
    &\xrightarrow{X_0\to1} \frac{\frac{d}{dX_0}\mathfrak{A}(\alpha(X_0))\big|_{X_0=1}}{\frac{d}{dX_0}\mathfrak{A}(\alpha(1-X_0))\big|_{X_0=1}}\\
    &= \frac{1-\zeta_{i-1}}{1-\zeta_i}
    \left[ \frac{\omega+\zeta_i(1-\zeta_i)}{\omega+\zeta_i(1-\zeta_{i-1})} \right]^{\delta^{-1}\zeta_{i}}
    \left[\frac{\omega+\zeta_{i-1}(1-\zeta_{i-1})}{\omega+\zeta_{i-1}(1-\zeta_{i})}\right]^{\delta^{-1}\zeta_{i-1}}.
\end{align*}
Note that for $\omega\to\infty$, we can recover the same expression for the fraction $\frac{\pi_{i-1}(X_0)}{\Tilde{\pi}_{i}(X_0)}$ as before with $a=(1-\zeta_i)/(1-\zeta_{i-1})$. Unfortunately, unlike in the previous cases, 
\begin{align*}
    \prod_{j=1}^{i} \frac{\pi_{j-1}}{\Tilde{\pi}_{j}} \neq \frac{\pi_0}{\Tilde{\pi_{i}}}
\end{align*}
however, for large $n$, the right-hand side is a good approximation. We arrive at
\begin{align*}
    Z^\ast_i \thickapprox \frac{\pi_0}{\Tilde{\pi}_i} = 
    \frac{1}{1-\zeta_i}\frac{\omega + \zeta_i(1-\zeta_i)}{\omega+\zeta_i} Z^\ast_1,\quad 2\leq i \leq n-1,
\end{align*}
and $Z^\ast_1$ chosen s.t. $\| Z^\ast \|_1=1$. The rest goes along the lines of the calculations for $A_t$. Embedding again the discrete distribution $Z^\ast$ in $[0,1]$ by $Z^\ast([0,\zeta_{1}])=Z^\ast_1$, 
$Z^\ast( (\zeta_{i-1}, \zeta_{i}])=Z^\ast_i$ for $2\leq i \leq n-1$ 
and $Z^\ast((\zeta_{n-1},1])=0$, further let $F_{W^\ast}:\mathbb{R}\to[0,1]$ be the cdf of the embedded distribution. Then for $0\leq y < 1$
\begin{align*}
    F_{Z^\ast}(y) &= Z^\ast((-\infty,y])=\frac{\sum_{j=1}^{\lfloor yn \rfloor} Z^\ast_j}{\sum_{j=1}^{n-1} Z^\ast_j} 
\end{align*}
Noting that $\frac{\omega}{\omega+1} < (1-\zeta_i)Z_i^\ast/Z_1^\ast < 1$ we find weak convergence of the cdf to $\textbf{1}_{\{1\}}$:
\begin{align*}
    F_{Z^\ast}(y) < \frac{\sum_{j=1}^{\lfloor yn \rfloor} 
            \frac{1}{n-j} }
            {\sum_{j=1}^{n-1} \frac{1}{n-j} \frac{\omega}{\omega+1}}
            = \frac{\sum_{j=\lfloor (1-y)n \rfloor}^{n-1} 
            \frac{1}{j} }
            {\frac{\omega}{\omega+1} \sum_{j=1}^{n-1} \frac{1}{j} }
            \xrightarrow{n\to\infty}0,
\end{align*}
implying that the embedded stationary distribution is, for $n\to\infty$, converging in distribution to the Dirac-measure at 1:
\begin{align*}
    Z^\ast \xrightarrow{d} \delta_1.
\end{align*}
This limit is symbolic as $\zeta=1$ does not give a well-defined $SIQS$-process. In Fig.~\ref{fig: RAD zeta} $Z^\ast$ is shown for $n=100$. The stationary distribution is heavily concentrated near $1$. The tail is vanishing: $P(Z^\ast \leq \zeta) \xrightarrow{n \to \infty} 0$ for any $\zeta<1$. This strong bias towards quiescent behavior indicates that a strain not maximizing $\zeta$ will almost surely not survive. The analogue can be observed in the $SIS$-model, where evolution drives the strains to become slower and slower. We have seen that quiescence can be interpreted as a mechanism for a strain to slow down its activity, i.e., decrease the infection and recovery rate. 

It seems natural to combine $W$ and $Z$ to a two-dimensional random walk on the grid $\{\zeta_1,...,\zeta_{n-1} \} \times \{\omega_0,...,\omega_{n}\}$. This, however, gains no further insight; the marginal distributions of the resulting stationary distribution on the grid behave qualitatively just like the one-dimensional stationary distributions.

\subsection{Adaptive Dynamics of $\zeta$ in the SIQS model including a non-zero death term}
The evolution in $\zeta$ did hint that it is -- the longer the better -- preferential to be in the Q state. This is, of course, not realistically and is also caused by the fact that quiescence has no disadvantage (if $\hat\mu=0$, the reproduction number is independent on the switching rates).  If, however, much time is spend in the Q state, effects as natural death (or recovery) cannot be neglected. In the present section, we aim to include that effect in allowing that $\hat\mu>0$. To remain in the quasi-neutral case $\mathcal{R}_{0,1} = \mathcal{R}_{0,2}$, we  compensate the additional outflow from the Q state by $\beta_2=\beta_1\frac{\gamma+\vartheta_2 \hat{\mu}}{\gamma+\vartheta_1 \hat{\mu}}$. Both strains have equal fitness but different timing due to quiescence.
Unfortunately, this generalization significantly worsens the algebra and the resulting terms are too cumbersome to be of any practical use. We suffice with evaluating numerically the stationary distributions of the Markov chain induced by Adaptive Dynamics in the discretized trait space of $\zeta$ from \ref{SIQSzetaAD}, through use of the computer algebra package maxima~\cite{Maxima.2023}. 

The procedure is as following: Again, we discretize the $\zeta$-axis to $\{\zeta_0=0 , \zeta_1=\delta, ..., \zeta_{n-1}= (n-1)\delta \}$, and use the algorithm presented in ... to calculate the drift and diffusion terms $v_{\mu}(\zeta_i)$ and $D_{\mu} (\zeta_i)$ at the grid points. Define $(Z_{t}^{\mu})_{t\in\mathbb{N}}$ just as in \ref{SIQSzetaAD} but with the transition probabilities $\pi_i^\mu, \Tilde{\pi_i}^\mu$ induced by the $\hat{\mu}$-dependent drift and diffusion terms. The stationary distribution $Z^{\mu\ast}=(Z_1^{\mu\ast},...,Z_n^{\mu\ast})$ is given by the detailed balance condition
\begin{align*}
	Z_i^\ast=\frac{\pi^\mu_{i-1}(X_0)}{\Tilde{\pi^\mu}_{i}(X_0)}\cdot...\cdot\frac{\pi^\mu_{0}(X_0)}{\Tilde{\pi}^\mu_{1}(X_0)} Z_1^\ast.
\end{align*}

Recalling (\ref{Formulafixationprob}) and l'Hopital, we can calculate the stationary distribution straight from the drift and diffusion terms: 
\begin{align*}
	\frac{\pi^\mu_{i-1}(X_0)}{\Tilde{\pi}^\mu_{i}(X_0)} 
	&= \frac{\pi^\mu_{i-1}(X_0)}{1-\pi^\mu_{i-1}(1-X_0)}
	= \frac{ \int_{X_0}^1 \nu (y)dy /  \int_{0}^1 \nu (y)dy }{ 1 - \int_{1-X_0}^1 \nu (y)dy/ \int_{0}^1 \nu (y)dy}\\
	&= \frac{  \int_{X_0}^1 \nu (y)dy}{  \int_{0}^{1-X_0} \nu (y)dy}
	\xrightarrow{X_0 \to 1} \frac{\nu(1)}{\nu(0)} = \nu(1) = \exp \left( -2\int_0^1 v_\mu(z)/D_\mu(z) dz\right)
\end{align*}
Note that we took the limit $X_0\to1$. This means the number of mutants introduced to the population approaches $0$.

Interestingly, the resulting stationary distributions are no longer degenerate. 
In Fig.~\ref{fig: Sim Stationary} the stationary distribution is given under high population turnover $\mu=0.08$ (a) and low population turnover $\mu=0.01$ (b).
Bear in mind that the additional outflow from $Q$ is compensated for by a higher infection rate. The distribution has a single peak away from $1$. In fact, the probability plummets rapidly to $0$ as $\zeta$ approaches 1, whereas it remains positive on the lower tail. As expected, with lower turnover, the peak is more pronounced and shifted to the right. In Fig.~\ref{fig: Sim Stationary omega} the stationary distribution is given with high frequency of jumps $\omega_1 = \omega_2 = 1$ (a) and low frequency of jumps $\omega_1 = \omega_2 = 10$ (b). We recognize the same shape: A single peak with vanishing upper tail and positive lower tail. Anew, the maximum is more pronounced and shifted to the right when $\omega$ is large.

\begin{figure}[htbp]
	\begin{subfigure}[b]{0.45\textwidth}
		\includegraphics[width=\textwidth]{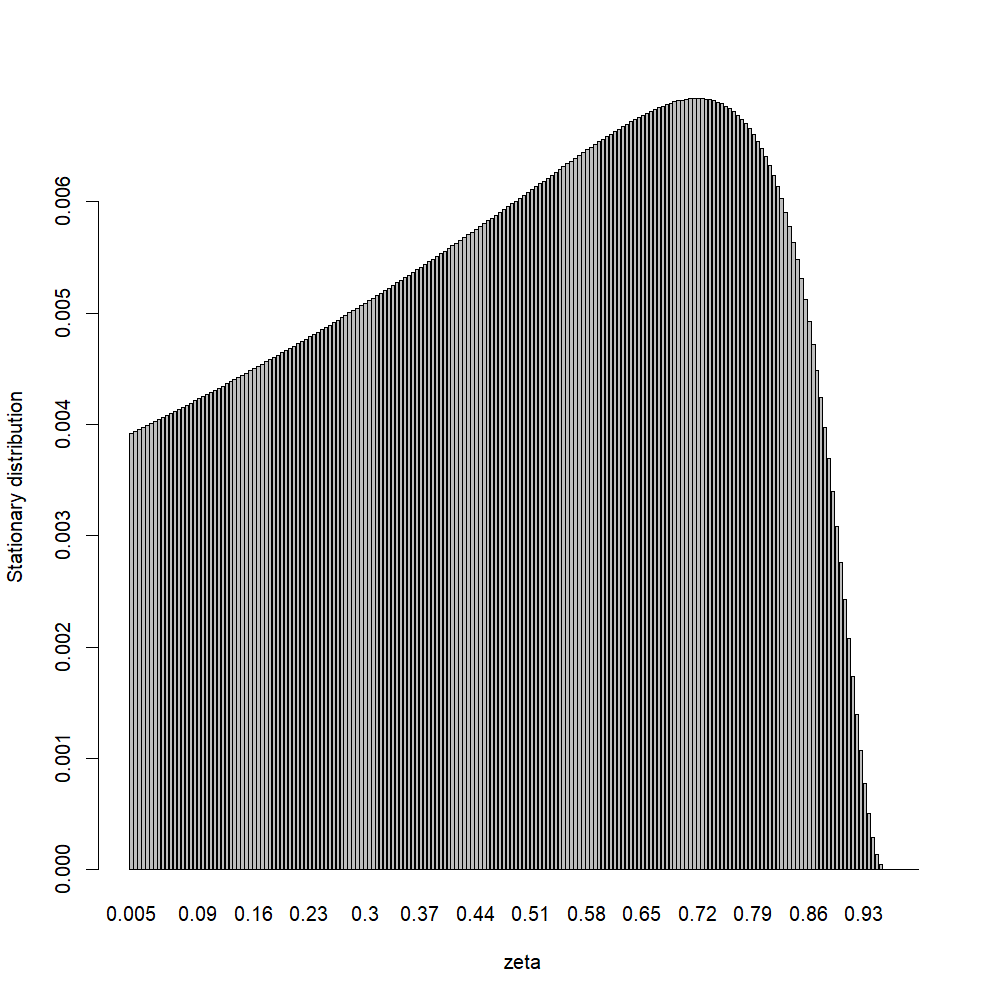}
		\caption{$\mu=0.08$ }
		\label{fig: Sim Stationary distribution mu8}
	\end{subfigure}
	\hfill
	\begin{subfigure}[b]{0.45\textwidth}
		\includegraphics[width=\textwidth]{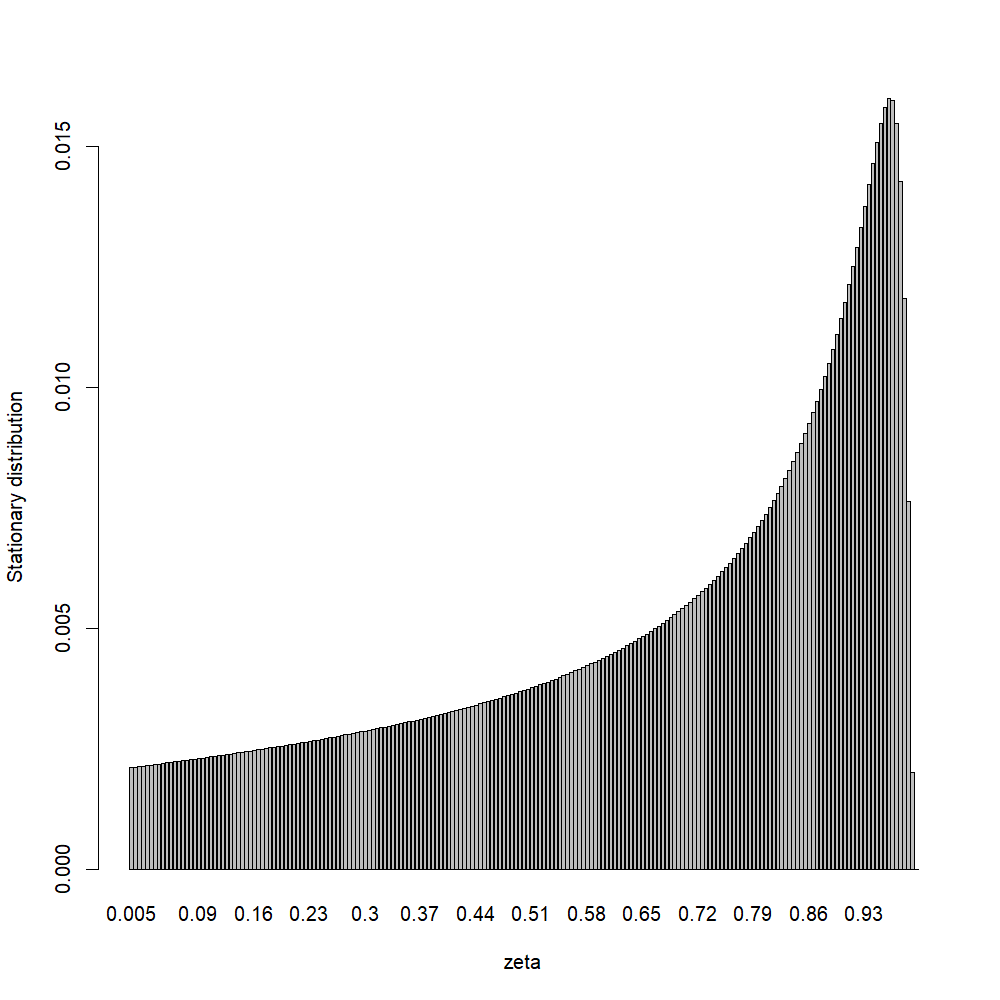}
		\caption{ $\mu=0.01$}
		\label{fig: Sim Stationary distribution mu1}
	\end{subfigure}
	\caption{ Stationary distributions for $\zeta$ under high population turnover (a) and low population turnover (b). The other parameters are $n=200$, $\omega_1=\omega_2=2$, $\beta_1=2$, $\gamma=1$.  Here we consider the SIS model which approximates the SIQS model in the rapid switching limit: The parameters $\omega$ and $c$ appear in the rates of the SIS model (see Theorem~\ref{rapidSwitchTheo}).}
	\label{fig: Sim Stationary}
\end{figure}
\begin{figure}[htbp]
	\begin{subfigure}[b]{0.45\textwidth}
		\includegraphics[width=\textwidth]{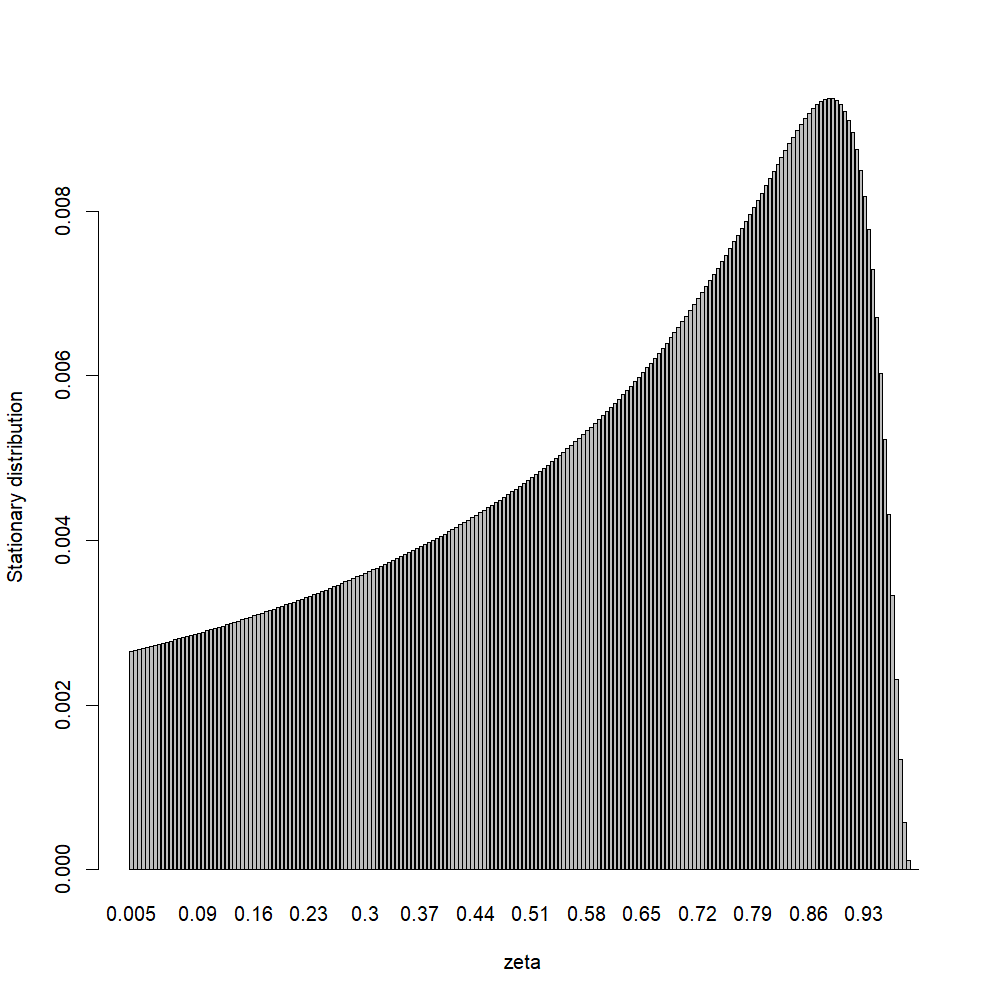 }
		\caption{$\omega_1 = \omega_2 = 10$ }
		\label{fig: Sim Stationary distribution mu8w10}
	\end{subfigure}
	\hfill
	\begin{subfigure}[b]{0.45\textwidth}
		\includegraphics[width=\textwidth]{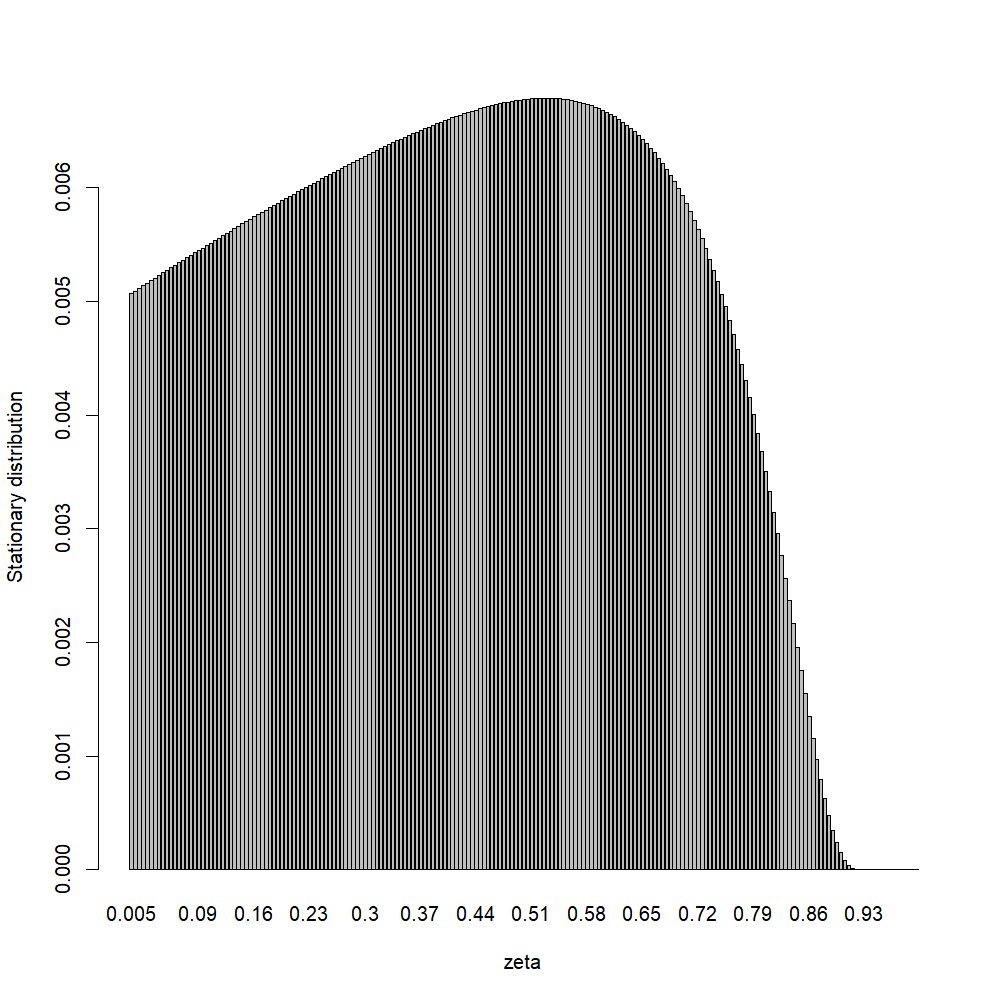}
		\caption{ $\omega_1 = \omega_2 = 1$}
		\label{fig: Sim Stationary distribution mu8w1}
	\end{subfigure}
	\caption{ Stationary distributions for $\zeta$ with high frequency of jumps (a) and low frequency of jumps (b). The other parameters are $n=200$, $\mu=0.08$, $\beta_1=2$, $\gamma=1$.  } %
	\label{fig: Sim Stationary omega}
\end{figure}

In Fig.~\ref{fig: Sim Stationary} the stationary distribution is given under high population turnover $\mu=0.08$ (a) and low population turnover $\mu=0.01$ (b).
Bear in mind that the additional outflow from $Q$ is compensated for by a higher infection rate. The distribution has a single peak away from $1$. In fact, the probability plummets rapidly to $0$ as $\zeta$ approaches 1, whereas it remains positive on the lower tail. As expected, with lower turnover, the peak is more pronounced and shifted to the right. In Fig.~\ref{fig: Sim Stationary omega} the stationary distribution is given with high frequency of jumps $\omega_1 = \omega_2 = 1$ (a) and low frequency of jumps $\omega_1 = \omega_2 = 10$ (b). We recognize the same shape: A single peak with vanishing upper tail and positive lower tail. Anew, the maximum is more pronounced and shifted to the right when $\omega$ is large.

\subsection{Intuition for the mechanism behind the weak selection}\label{innterpret}
Our results show that even in a constant environment, slower traits are superior to faster ones. This result is counter intuitive. Naively, we expect that the faster a strain is, the fitter it is. Unfortunately, the result is a consequence of rather involving and technical arguments and does not lead to a direct deeper insight. In this section, we aim to use some more informal arguments. Thereto, we focus on an SIS model with two different incidence rate constants $\beta_i$ and recovery rate constants $\gamma_i$, where we require the same fitness in the deterministic limit, $R_0=\beta_1/\gamma_1=\beta_2/\gamma_2$. The two-dimensional van Kampen expansion up to second order is given by 
\begin{equation}\label{FP2dimSIS}
	\begin{aligned}
		\partial_t \rho = &-\partial_x \left( \left[ \beta_1(1-x-y)x-\gamma_1x \right] \rho \right)
		-\partial_y \left(\left[ \beta_2(1-x-y)y-\gamma_2y \right]\rho \right)\\
		&+\frac{h}{2}\partial_x^2 \left(\left[ \beta_1(1-x-y)x+\gamma_1x \right] \rho \right)
		+\frac{h}{2}\partial_y^2 \left(\left[ \beta_2(1-x-y)y+\gamma_2y \right]\rho \right).
	\end{aligned}
\end{equation}
That is, the drift (the corresponding ODE) reads
$$\frac d{dt}
\left(\begin{array}{c}
x\\y
\end{array}\right)
= 
\left(\begin{array}{c}
	\gamma_1( (R_0(1-x-y)-1)x\\
	 \gamma_2((R_0(1-x-y)-1)y
\end{array}\right) =: f(x,y).$$

We start off with a state $(x^\#,y^\#)$ on the coexistence line (CL), $x^\#+y^\# = 1-1/R_0$. Noise will carry the state away from the coexistence line. Note that the noise term  
is completely symmetric, such that noise alone does not change the expected location of the state. \\
Anyhow, assume that only one strain becomes perturbed. That is, from $(x^\#, 1-1/R_0-x^\#)$ the state jumps to either 
$$ 
\vec x^\pm_i = \left(\begin{array}{c}
x^\#\\
y^\#
\end{array}\right)\pm \eps \e_i,\qquad i\in\{1,2\}$$
where $\e_1=(1,0)^T$ and $\e_2=(0,1)^T$. Evaluating the drift at these two positions yields 
(note that $R_0(1-x^\#-y^\#)-1=R_0(1-(1-1/R_0))-1=0$ and 
$R_0(1-(x^\#\pm\eps)-y^\#)-1=\mp R_0\eps$)
$$ 
f(\vec x^\pm_1) = 
 \left(\begin{array}{c}
\mp\,\gamma_1\,R_0\,\eps\,(x^\#\pm\eps)\\
\mp\,\gamma_2\,R_0\,\eps\,y^\#\\
\end{array}\right) 
= 
\mp\eps R_0\,\left(\begin{array}{c}
	\gamma_1\,x^\#\\
	\gamma_2\,y^\#
\end{array}\right)
-
	\gamma_1\,R_0\,\eps^2\e_1
$$
and 
$$
f(\vec x^\pm_2) = 
\left(\begin{array}{c}
	\mp\,\gamma_1\,R_0\,\eps\,x^\ast\\
	\mp\,\gamma_2\,R_0\,\eps\, (y^\#\pm\eps)
\end{array}\right)
= 
\mp\eps R_0\,\left(\begin{array}{c}
	\gamma_1\,x^\#\\
	\gamma_2\,y^\#
\end{array}\right)-	\gamma_2\,R_0\,\eps^2\e_2
$$
The vector field which drives back the state consists of two parts: The first order term has the same absolute size for the jump in direction $+\eps \e_i$ and $-\eps\e_i$, but opposite signs. As the noise is symmetric, in average, this term cancels out. There is in average no transport of the state due to this first term.\\
However, the second order term always has a negative sign. That is, if we jump up or down, the strain which is jumping will be decreased in freqeuncy. There is a net average displacement due to this second term. The amount of the decrease is proportional to the removal rate of this strain. A smaller removal rate does mean less decrease. The slower strain has an advantage. Once core message here: The advantage is based on a second order term. This observation might be in line with the fact that we have a weak selection.\par\medskip 

This consideration already moves us closer to a direct interpretation, but we are still not there. We still work with computations, though they made the issue clearer: Direct effects will not explain the effect, we need to consider the interaction of several compartments. Let us decipher these higher order interactions.\par\medskip 

Let us rewrite the terms from above focussing on the stochastic distortion of the first strain, where $x=x^\#\pm\Delta x$ and consequently $s=s_0\mp\Delta s$, $s_0=1-1/R_0$, $\Delta x = \Delta s=\eps$, then the net infection rate (incidence minus recovery) for strain~1 reads
\begin{eqnarray*}
\beta_1 s x -\gamma_1 x 
&=& \pm \gamma_1(R_0(s_0\Delta x-x^\#\Delta s)-\Delta x) )
-\gamma_1R_0\Delta x\Delta s.
\end{eqnarray*} 
The term linear in $\Delta x$ and $\Delta s$ will be in average balanced by jumps in the two directions $\pm \eps$. The interaction term (increase of $x$ by $\eps$ and decrease of $s$ by $\eps$) remains, has this term always has a negative sign.\par\medskip 

To sum up: The strain which is distorted always has an disadvantage in the interaction between number of infecteds and number of susceptibles. Either this strain becomes larger by $\eps$, but these surplus of infected individuals have less susceptibles to infect (also the susceptibles have been decreased by $\eps$). Or, the strain is deceased by $\eps$. Then, there are more susceptibles (a surplus of $\eps$) to infect, but less infectious individuals who are able to infect suscpetibles. The other strain is also affected by the distortion of the first strain, but the effects on the other strain cancel out in average: In the equation of the other strain only a first order term appears.\\ 
While effects which are linear in $\eps$ in average cancel out, the quadratic interaction terms do not. They lead to an disadvantage of the distorted strain, which is proportional to $\eps^2$ and the time scale of the corresponding strain. 
The slower strain has a smaller disadvantage and in that is favoured by selection. \par\medskip

If we consider this argument, it becomes clear that the effects we see here have nothing to do with quiescences {\it per se} but with timing only. The pathogens do not aim to escape some detrimental environment through switching into hibernation. Quiescence is merely a way to slow down processes while keeping the reproduction number untouched. Every other mechanism which slows down the timing of a pathogen while keeping the reproduction number would have the same effect and also would be favoured by evolution.

\section{Discussion}
We discussed a stochastic SIQS model in the framework of Adaptive Dynamics. 
As the model formulated as an individual based process with a finite population is rather challenging to analyze, 
we propose several approximations to reduce the complexity: A fast I-Q-switching limit allows to approximate the SIQS model by an SIS model with adapted/averaged rate constants. 
For the usage by Adaptive Dynamic, we investigate the competition of two strains that only differ in the quiescent behavior, but not in their reproduction number: This neutrality condition
allows to reduce the diffusion limit of the models (SIQS as well as SIS) to a one-dimensional stochastic differential equation. Herein, the fixation probabilities for the strains have been determined, which in turn allowed to show by Adaptive Dynamics that evolution favours the slower trait, and in that, favours quiescence. \par\medskip 

One main outcome of this analysis is that quiescence as a trait is favoured without a switching environment. This finding is in line with previous findings by Blath et al.~\cite{Blath.2020}. The mechanism clearly relies on stochastic effects and disappears if the population size tends to infinity such that stochasticity vanishes.\\
If we look closer at the analysis, we observe the second main finding: 
Quiescence {\it per se} is not the driving force for this trait to be selected: Usually, it is assumed that the quiescent state allows the strain to better withstand extrinsic or intrinsic stochastic fluctuations  which lead to a disadvantage for the strain; in that, quiescence is a response to a fluctuating environment. If we consider the SIS model with re-scaled rates that approximate the SIQS model, also in this setting quiescence is favoured. However, in the SIS model, the Q state does not exist. The nature of the selected trait is different: The main effect is a decrease of the time scale of a strain (infection rate as well as recovery rate).\\ 
The interpretation of quiescence as a mean to control the time scale of a strain resembles the way how the ion current in axons is controlled by ion channels: Though the single channels only can assume one of the two states, open or closed, by stochastic   switching between the two states, a certain fraction of all ion channels are open, and in that, the ion current can be controlled in a continuous manner.  This is a parsimonious design, which avoids to construct sophisticated and complex ion channels where each individual channel is able to control the ion flux in a continuous manner.\\
In the same way, a single infected individuals either is infectious or quiescent. On the population level however, the incidence and the recovery rate, that is, the time scale of the infection process, can continuously be controlled by the switching rates between Q and I. It is not quiescence, but slowness,  which is selected for. \par\medskip 

By a rather informal summary of the approximation results we have indicated that the selection is based on interactions of second order terms in stochastic fluctuations.  
The core mechanism is based on the fact that the noise acts independently on the frequency of both strains. If the frequency of a strain is decreased, the susceptible class is increased. That is, there are more susceptibles but less individuals of this strain to infect them. This is a disadvantage for this strain. Or, if a strain is increased, the susceptibles are decreased. That is, there are more individuals in the distorted strain, but they have less individuals to infect. Also this interaction is an disadvantage. Any stochastic fluctuation has a disadvantage for the fluctuating strain. The size of this disadvantage depends on the velocity of the strain: The faster, the larger is the disadvantage. In that, slower strains have a smaller disadvantage (or, slower strains have an advantage above faster strains). We conjecture that this mechanism is rather universal and can be also found in other biological systems. Whenever two strains interact only indirectly by competing for an environmental resource (in our case: the susceptible individuals), this effect is more likely to play a role. For example, we strongly suspect that this effect can also be found in competition models such as the Volterra model. \par\medskip 

Another consequence of the analysis clearly indicates that the evolutionary dynamics is slow (we have no selective sweeps but only weak selection). And, we do not have an ESS which can  be localized in the trait space as we are used from classical Adaptive Dynamics with a deterministic PIP. In our case, the ESS is replaced by a  broad equilibrium distribution of traits. This finding resembles the difference between a linear ODE with a stable stationary state and an Ornstein-Uhlenbeck process with some small noise  (see, e.g.~\cite{Grasman1999}). \\
It is therefore to be expected that in observations the difference in the degree of quiescence between different strains is large. If we consider the different pathogens causing malaria, some do quiescence, and some not. This difference can be very well caused by different environmental conditions, but it is also plausible that it is a consequence of our findings. \par\medskip 

The evolution of the time scale of actions mainly are  discussed as a life history trait which forms a co-factor or co-evolves with some specific trait of interest, such as cooperation. 
For example, spatial diffusion is considered by Lee et al.~\cite{lee2022slow}. Usually the trait which disperse faster out-compete slower diffusing traits. In that paper, it is shown that slow cheaters can invade fast cooperators. In a similar spirit, it is shown that a strong Allee effect is able to favour slow strains~\cite{taylor2005allee}. 
Within an homogeneous population, quiescence, seedbanks and age structure is considered as a life-history trait, again as a co-trait to cooperation~\cite{sellinger2019better,muller2022life,john2023age}. Depending on whether the population is constant or growing exponentially, the coevolution of cooperation and a slower trait (more dormancy, deeper seed banks) stabilizes cooperation.  
These results indicate that slow time scales can be advantageous under evolutionary pressure. However, the time scale of action {\it per se} is less investigated. 
\par\medskip

All in all, the interpretation of quiescence as a method to slow down the time scale of actions has its own right, particularly in a constant environment. For fluctuating environments, the classical bet-hedging interpretation clearly is appropriate. However, the broader view on quiescence and dormancy introduced here is a promising approach for biological systems where there is no obviously fluctuating environment, or where the fluctuation in the environment and the switching rates between active and dormant states possess very different time scales.

\addcontentsline{toc}{section}{References}
\bibliography{Literaturverzeichnis_neu.bib}

\begin{thebibliography}{10}

\bibitem{blath2021branching}
Jochen Blath, Felix Hermann, and Martin Slowik.
\newblock A branching process model for dormancy and seed banks in randomly
  fluctuating environments.
\newblock {\em Journal of Mathematical Biology}, 83(2):17, 2021.

\bibitem{Blath.2020}
Jochen Blath and Andr{\'a}s T{\'o}bi{\'a}s.
\newblock Invasion and fixation of microbial dormancy traits under competitive
  pressure.
\newblock {\em Stochastic Processes and their Applications},
  130(12):7363--7395, 2020.

\bibitem{Browning2021}
Alexander~P. Browning, Jesse~A. Sharp, Tarunendu Mapder, Christopher~M. Baker,
  Kevin Burrage, and Matthew~J. Simpson.
\newblock Persistence as an optimal hedging strategy.
\newblock {\em Biophysical Journal}, 120(1):133--142, January 2021.

\bibitem{Chakravorty2022}
Srishti Chakravorty, Behdad Afzali, and Majid Kazemian.
\newblock Ebv-associated diseases: Current therapeutics and emerging
  technologies.
\newblock {\em Frontiers in Immunology}, 13, October 2022.

\bibitem{claessen2014bacterial}
Dennis Claessen, Daniel~E Rozen, Oscar~P Kuipers, Lotte S{\o}gaard-Andersen,
  and Gilles~P Van~Wezel.
\newblock Bacterial solutions to multicellularity: a tale of biofilms,
  filaments and fruiting bodies.
\newblock {\em Nature Reviews Microbiology}, 12(2):115--124, 2014.

\bibitem{Gowda2018}
Kiran~K. Dayananda, Rajeshwara~N. Achur, and Channe~D. Gowda.
\newblock Epidemiology, drug resistance, and pathophysiology of plasmodium
  vivax malaria.
\newblock {\em Journal of Vector Borne Diseases}, 55(1):1, 2018.

\bibitem{Dieckmann.1996}
U.~Dieckmann and R.~Law.
\newblock The dynamical theory of coevolution: a derivation from stochastic
  ecological processes.
\newblock {\em Journal of mathematical biology}, 34(5-6):579--612, 1996.

\bibitem{diekmann2004beginner}
Odo Diekmann et~al.
\newblock A beginner's guide to adaptive dynamics.
\newblock {\em Banach Center Publications}, 63:47--86, 2004.

\bibitem{ethier2009markov}
Stewart~N Ethier and Thomas~G Kurtz.
\newblock {\em Markov processes: characterization and convergence}.
\newblock John Wiley \& Sons, 2009.

\bibitem{gaal2010exact}
Bernadett Ga{\'a}l, Jonathan~W Pitchford, and A~Jamie Wood.
\newblock Exact results for the evolution of stochastic switching in variable
  asymmetric environments.
\newblock {\em Genetics}, 184(4):1113--1119, 2010.

\bibitem{Gaal.2010}
Bernadett Ga{\'a}l, Jonathan~W. Pitchford, and A.~Jamie Wood.
\newblock Exact results for the evolution of stochastic switching in variable
  asymmetric environments.
\newblock {\em Genetics}, 184(4):1113--1119, 2010.

\bibitem{Gardiner.2009}
Crispin~W. Gardiner.
\newblock {\em Stochastic methods: A handbook for the natural and social
  sciences}.
\newblock Springer complexity. Springer, Berlin, 4th ed. edition, 2009.

\bibitem{geritz1997dynamics}
Stefan~AH Geritz, Johan~AJ Metz, {\'E}va Kisdi, and G{\'e}za Mesz{\'e}na.
\newblock Dynamics of adaptation and evolutionary branching.
\newblock {\em Physical Review Letters}, 78(10):2024, 1997.

\bibitem{Givon2004extracting}
Dror Givon, Raz Kupferman, and Andrew Stuart.
\newblock Extracting macroscopic dynamics: model problems and algorithms.
\newblock {\em Nonlinearity}, 17(6):R55, 2004.

\bibitem{Grasman1999}
Onno A.~Herwaarden Grasman, Johan.
\newblock {\em Asymptotic Methods for the Fokker-Planck Equation and the Exit
  Problem in Applications}.
\newblock Springer Series in Synergetics Ser. Springer Berlin / Heidelberg,
  Berlin, Heidelberg, 1999.
\newblock Description based on publisher supplied metadata and other sources.

\bibitem{guckenheimer2013nonlinear}
John Guckenheimer and Philip Holmes.
\newblock {\em Nonlinear oscillations, dynamical systems, and bifurcations of
  vector fields}, volume~42.
\newblock Springer Science \& Business Media, 2013.

\bibitem{Hadley2016}
Graham~R. Hadley, Julie~A. Gayle, Juan Ripoll, Mark~R. Jones, Charles~E.
  Argoff, Rachel~J. Kaye, and Alan~D. Kaye.
\newblock Post-herpetic neuralgia: a review.
\newblock {\em Current Pain and Headache Reports}, 20(3), February 2016.

\bibitem{Heinrich.2018}
Lukas Heinrich, Johannes M{\"u}ller, Aur{\'e}lien Tellier, and Daniel
  {\v{Z}}ivkovi{\'c}.
\newblock Effects of population- and seed bank size fluctuations on neutral
  evolution and efficacy of natural selection.
\newblock {\em Theoretical Population Biology}, 123:45--69, 2018.

\bibitem{john2023age}
Sona John and Johannes M{\"u}ller.
\newblock Age structure, replicator equation, and the prisoner’s dilemma.
\newblock {\em Mathematical Biosciences}, 365:109076, 2023.

\bibitem{Joyner2016}
Chester Joyner, Alberto Moreno, Esmeralda V.~S. Meyer, Monica Cabrera-Mora,
  Jessica~C. Kissinger, John~W. Barnwell, and Mary~R. Galinski.
\newblock Plasmodium cynomolgi infections in rhesus macaques display clinical
  and parasitological features pertinent to modelling vivax malaria pathology
  and relapse infections.
\newblock {\em Malaria Journal}, 15(1), September 2016.

\bibitem{Kogan.2014}
Oleg Kogan, Michael Khasin, Baruch Meerson, David Schneider, and Christopher~R.
  Myers.
\newblock Two-strain competition in quasi-neutral stochastic disease dynamics.
\newblock {\em Physical Review E}, 90(4):385, 2014.

\bibitem{Kurtz.1973}
Thomas~G. Kurtz.
\newblock A limit theorem for perturbed operator semigroups with applications
  to random evolutions.
\newblock {\em Journal of Functional Analysis}, 12(1):55--67, 1973.

\bibitem{lee2022slow}
Hyunseok Lee, Jeff Gore, and Kirill~S Korolev.
\newblock Slow expanders invade by forming dented fronts in microbial colonies.
\newblock {\em Proceedings of the National Academy of Sciences},
  119(1):e2108653119, 2022.

\bibitem{lennon2011microbial}
Jay~T Lennon and Stuart~E Jones.
\newblock Microbial seed banks: the ecological and evolutionary implications of
  dormancy.
\newblock {\em Nature reviews microbiology}, 9(2):119--130, 2011.

\bibitem{Majda2001}
Andrew~J. Majda, Ilya Timofeyev, and Eric {Vanden Eijnden}.
\newblock A mathematical framework for stochastic climate models.
\newblock {\em Communications on Pure and Applied Mathematics}, 54(8):891--974,
  2001.

\bibitem{Maxima.2023}
Maxima.
\newblock Maxima, a computer algebra system, 2023.

\bibitem{merrick2021hypnozoites}
Catherine~J Merrick.
\newblock Hypnozoites in plasmodium: do parasites parallel plants?
\newblock {\em Trends in parasitology}, 37(4):273--282, 2021.

\bibitem{Muller.2013}
Johannes M{\"u}ller and Christina Kuttler.
\newblock {\em Methods and Models in Mathematical Biology: Deterministic and
  Stochastic Approaches}.
\newblock Lecture Notes on Mathematical Modelling in the Life Sciences.
  {Springer-Verlag Berlin and Heidelberg GmbH {\&} Co. K}, Munich, 2013.

\bibitem{muller2022life}
Johannes M{\"u}ller and Aur{\'e}lien Tellier.
\newblock Life-history traits and the replicator equation.
\newblock {\em Mathematical Biosciences}, 349:108826, 2022.

\bibitem{Mueller2013}
J.~Müller, B.A. Hense, T.M. Fuchs, M.~Utz, and Ch. Pötzsche.
\newblock Bet-hedging in stochastically switching environments.
\newblock {\em Journal of Theoretical Biology}, 336:144--157, November 2013.

\bibitem{Nadolny2005}
Sten Nadolny.
\newblock {\em The discovery of slowness}.
\newblock Paul Dry Books, Philadelphia, 2005.

\bibitem{phillips2017malaria}
MA~Phillips, JN~Burrows, C~Manyando, RH~van Huijsduijnen, WC~Van~Voorhis, and
  T~Wells.
\newblock Malaria. nature reviews. disease primers, 3, 17050, 2017.

\bibitem{rittershaus2013normalcy}
Emily~SC Rittershaus, Seung-Hun Baek, and Christopher~M Sassetti.
\newblock The normalcy of dormancy: common themes in microbial quiescence.
\newblock {\em Cell host \& microbe}, 13(6):643--651, 2013.

\bibitem{usman:arxiv}
Usman Sanusi, Sona John, Johannes Mueller, and Aurélien Tellier.
\newblock Stochastic time to extinction of an siqs epidemic model with
  quiescence.
\newblock {\em arXiv}, page 2503.06178, 2025.

\bibitem{sellinger2019better}
Thibaut Sellinger, Johannes M{\"u}ller, Volker H{\"o}sel, and Aur{\'e}lien
  Tellier.
\newblock Are the better cooperators dormant or quiescent?
\newblock {\em Mathematical biosciences}, 318:108272, 2019.

\bibitem{sorrell2009evolution}
Ian Sorrell, Andrew White, Amy~B Pedersen, Rosemary~S Hails, and Mike Boots.
\newblock The evolution of covert, silent infection as a parasite strategy.
\newblock {\em Proceedings of the Royal Society B: Biological Sciences},
  276(1665):2217--2226, 2009.

\bibitem{Spriewald2020}
Stefanie Spriewald, Eva Stadler, Burkhard~A. Hense, Philipp~C. Münch, Alice~C.
  McHardy, Anna~S. Weiss, Nancy Obeng, Johannes Müller, and Bärbel Stecher.
\newblock Evolutionary stabilization of cooperative toxin production through a
  bacterium-plasmid-phage interplay.
\newblock {\em mBio}, 11(4), August 2020.

\bibitem{stumpf2002herpes}
Michael~PH Stumpf, Zo{\"e} Laidlaw, and Vincent~AA Jansen.
\newblock Herpes viruses hedge their bets.
\newblock {\em Proceedings of the National Academy of Sciences},
  99(23):15234--15237, 2002.

\bibitem{taylor2005allee}
Caz~M Taylor and Alan Hastings.
\newblock Allee effects in biological invasions.
\newblock {\em Ecology letters}, 8(8):895--908, 2005.

\bibitem{tellier2019}
Aur{\'e}lien Tellier.
\newblock Persistent seed banking as eco-evolutionary determinant of plant
  nucleotide diversity: novel population genetics insights.
\newblock {\em New Phytologist}, 221(2):725--730, 2019.

\bibitem{thorley2013pathogenesis}
David~A Thorley-Lawson, Jared~B Hawkins, Sean~I Tracy, and Michael Shapiro.
\newblock The pathogenesis of epstein--barr virus persistent infection.
\newblock {\em Current opinion in virology}, 3(3):227--232, 2013.

\bibitem{venable1988selective}
D~Lawrence Venable and Joel~S Brown.
\newblock The selective interactions of dispersal, dormancy, and seed size as
  adaptations for reducing risk in variable environments.
\newblock {\em The American Naturalist}, 131(3):360--384, 1988.

\bibitem{verin2018host}
M{\'e}lissa Verin and Aur{\'e}lien Tellier.
\newblock Host-parasite coevolution can promote the evolution of seed banking
  as a bet-hedging strategy.
\newblock {\em Evolution}, 72(7):1362--1372, 2018.

\bibitem{Yin.2013}
George Yin.
\newblock {\em Continuous-time Markov Chains and Applications: A Two-time-scale
  Approach}.
\newblock Springer-Verlag, 2013.

\bibitem{zhang2016understanding}
Zheren Zhang, Dennis Claessen, and Daniel~E Rozen.
\newblock Understanding microbial divisions of labor.
\newblock {\em Frontiers in microbiology}, 7:2070, 2016.

\end{thebibliography}
\bibliographystyle{plain}

\begin{appendix}

	\maketitle

\end{comment}

\newpage 
\section{Supplementary information}

\subsection{Proofs}
\subsubsection{Reproduction number}\label{reproNumberProof}
We consider the SIQS model. Recall the transitions 
of this model: \par\medskip 

  \begin{center}
        \begin{tabular}{||c | c | c||} 
         \hline
         Event & Type of transition & Rate \\ [0.5ex] 
         \hline\hline
         Infection with strain 1 &  $I_1 \to I_1+1$ & $\beta S I_1/N$ \\ 
         \hline
         Infection with strain 2 &  $I_2 \to I_2+1$ & $\beta S I_2/N$ \\ 
         \hline
         Recovery of $I_1$ &  $I_1 \to I_1-1$ & $\gamma_1 I_1$ \\
         \hline
         Recovery of $I_2$ & $I_2 \to I_2-1$ & $\gamma_2 I_2$  \\
          \hline
         Deactivation of $I_1$ & $(I_1,Q_1) \to (I_1-1,Q_1+1)$ & $\hat{\omega}_1 I_1$  \\
           \hline
         Deactivation of $I_2$ & $(I_2,Q_2) \to (I_2-1,Q_2+1)$ & $\hat{\omega}_2 I_2$  \\
          \hline
         Reactivation of $Q_1$ & $(I_1,Q_1) \to (I_1+1,Q_1-1)$ & $\hat{c}_1 Q_1$  \\
           \hline
         Reactivation of $Q_2$ & $(I_2,Q_2) \to (I_2+1,Q_2-1)$ & $\hat{c}_2 Q_2$  \\ 
         \hline
		 Death/recovery of $Q_1$ & $Q_1 \to Q_1-1$ & $\hat \mu Q_1$  \\ 
           \hline
         Death/recovery of $Q_2$ & $Q_2 \to Q_2-1$ & $\hat \mu Q_2$  \\ [1ex] 
         \hline
        \end{tabular}
		\end{center} 
  \par\medskip

In order to determine the reproduction number we follow the fate of one infected individual (we thus drop the index in rates and states). 
We linearize at the uninfected equilibrium, that is, we assume that $S=N$ for all times. The focal individual has state $X_t$ 
at time $t$ which starts with 
state $I$, and jumps back and forth between $I$ and $Q$, until it jumps into the absorbing state $S$. 
As the Markov chain is the same for both strains, we suppress the index $i$ and have 
$$ 
I\rightarrow Q \mbox{ at rate } \hat\omega,\qquad 
Q\rightarrow I \mbox{ at rate } \hat c,\qquad 
I\rightarrow S \mbox{ at rate } \gamma,\qquad 
Q\rightarrow S \mbox{ at rate } \hat \mu$$
and 
$$ R_0 = \int_0^\infty \beta\, \chi(X_t=I)\, dt.$$
Thereto it is sufficient to consider the master equations for $p_I(t)=P(X_t=I)$ and  $p_Q(t)=P(X_t=Q)$, 
$$ \frac d {dt}\left(\begin{array}{c}p_I\\p_Q\end{array}\right) = A  \left(\begin{array}{c}p_I\\p_Q\end{array}\right),\qquad 
\left(\begin{array}{c}p_I(0)\\p_Q(0)\end{array}\right)=\left(\begin{array}{c}1\\0\end{array}\right).\qquad 
A = \left(\begin{array}{cc}
- \hat\omega- \gamma  &   \hat c\\
\hat\omega               &   -\hat c-\hat \mu
\end{array}\right).$$
Thus, 
\begin{align*}
R_0 
&=  \int_0^\infty \beta\, p_I(t))\, dt
=  \int_0^\infty \beta\, \e_1^Te^{At}\, \e_1\, dt 
= -\,\e_1^TA^{-1}\e_1 \\
&= \beta\,\frac{\hat c+\hat \mu}{(\hat\omega+\gamma)(\hat c+\hat\mu)-\hat\omega\hat c}
= \beta\,\frac{\hat c+\hat \mu}{\hat \mu\,(\hat\omega+\gamma)+\gamma\hat c}
= \frac{\beta_i}{\gamma_i+\hat \mu\,\frac{\hat \omega_i}{\hat c_i+\hat \mu}}. 
\end{align*}

\subsubsection{Master equation and Fokker-Planck equation}\label{stochEq}
We assume $\gamma_1=\gamma_2=\gamma$ and $\hat\mu=0$. 
Then, the Master equation of the two-strain SIQS model is given by the time-evolution of $P_{m,k,n,l}(t)=P(I_1(t)=m, Q_1(t)=k, I_2(t)=n, Q_2(t)=l)$ and reads
\begin{equation}\label{eq: Master SIIQQS}
    \begin{aligned}
        \Dot{P}_{m,k,n,l}(t) = -&P_{m,k,n,l}(t) \Bigg[ \frac{N-m-k-n-l}{N}(\beta_1 m+\beta_2 n) + \gamma n+\gamma m\\ 
		+& \hat{\omega}_1 m + \hat{\omega}_2 n + (\hat{c}_1+\mu) k + (\hat{c}_2+\mu) l \Bigg] \\
        +&P_{m-1,k,n,l}(t) \left[ \beta \frac{N-m-k-n-l+1}{N}(m-1) \right]\\
        +&P_{m,k,n-1,l}(t) \left[ \beta \frac{N-m-k-n-l+1}{N}(n-1)  \right]\\
        +&P_{m+1,k,n,l}(t) \left[ \gamma (m+1) \right]
        +P_{m,k,n+1,l}(t) \left[ \gamma (n+1)  \right]\\
        +&P_{m+1,k-1,n,l}(t) \left[ \hat{\omega}_1 (m+1)  \right]
        +P_{m,k,n+1,l}(t) \left[ \hat{\omega}_2 (n+1)  \right]\\
        +&P_{m-1,k+1,n,l}(t) \left[ (\hat{c}_1+\mu) (k+1)  \right]
        +P_{m,k,n,l+1}(t) \left[ (\hat{c}_2+\mu) (l+1)  \right].
    \end{aligned}
\end{equation}

Through van Kampen system size expansion and subsequent Taylor expansion to second order, the Fokker-Planck equation in four dimensions is obtained:
\begin{equation}\label{eq: FPorig SIIQQS}
    \begin{aligned}
    \partial_t \rho (x,u,y,v,t)=
    &-\partial_x \left[ \beta x (1-x-u-y-v)-\gamma x \right] \rho\\
    &-\partial_y \left[ \beta y (1-x-u-y-v)-\gamma y \right] \rho\\
    &- \left( \partial_x-\partial_u \right) \left[ \hat{c}_1 u - \hat{\omega}_1 x \right] \rho 
	\\
    &- \left( \partial_y-\partial_v \right) \left[ \hat{c}_2 v - \hat{\omega}_2 y \right] \rho
	\\
    &+\frac{1}{2N} \partial_x^2 \left[ \beta x (1-x-u-y-v)+\gamma x \right] \rho\\
    &+\frac{1}{2N} \partial_y^2 \left[ \beta y (1-x-u-y-v)+\gamma y \right] \rho\\
    &+\frac{1}{2N} \left( \partial_x-\partial_u \right)^2 \left[ \hat{c}_1 u + \hat{\omega}_1 x \right] \rho
	 \\
    &+\frac{1}{2N} \left( \partial_y-\partial_v \right)^2 \left[ \hat{c}_2 v + \hat{\omega}_2 y \right] \rho
    \end{aligned}
\end{equation}

\subsection{Large switching rate limit}
\label{largeSwitchLimitParticle}

\begin{theorem}\label{rapidSwitchTheoSI}
Assume that $\hat\omega_i$, $\hat c_i\rightarrow\infty$ such that $\frac{\hat \omega_i}{\hat c_i + \hat\omega_i}$ converge to $\zeta_i\in(0,1)$.
Then, the 2-strain SIQS model converges in distribution to a 2-strain SIS model with transition rates
{    \begin{center}
        \begin{tabular}{||c | c | c||} 
         \hline
         Event & Type of transition & Rate \\ [0.5ex] 
         \hline\hline
         Infection with strain 1 &  $\tilde I_1 \to \tilde I_1+1$ & $\beta\,\zeta_1 S \tilde I_1/N$ \\ 
         \hline
         Infection with strain 2 &  $\tilde I_2 \to \tilde I_2+1$ & $\beta\,\zeta_2 S \tilde I_2/N$ \\ 
         \hline
         Recovery of an infected individual (strain 1) &  $\tilde I_1 \to \tilde I_1-1$ & $(\gamma_1(1-\zeta_1)+\hat\mu\zeta_1) \tilde I_1$ \\
         \hline
         Recovery of an infected individual (strain 2)  & $\tilde I_2 \to \tilde I_2-1$ & $(\gamma_2(1-\zeta_2)+\hat\mu\zeta_2) \tilde I_2$  \\
          \hline
        \end{tabular}
		\end{center} 
		}\par\medskip\noindent 
(where $S=N-\tilde I_1-\tilde I_2$) in the sense that  $I_i+Q_i \overset{d}{\to} \tilde I_i$.
\end{theorem} 
{\bf Proof: }
We start with a proof for the case of one-strain and then argue that it generalizes to the two-strain case (also see~\cite{usman:arxiv}.
We refine the notion of convergence for the proof.\\
\textbf{Claim:} Let $(I,Q)_{t\geq0}$ be a two-strain $SIQS$-process with rates $\beta, \gamma$ and $\omega=\hat{\omega}/\varepsilon$, $c=\hat{c}/\varepsilon$. Then, for $\varepsilon\to 0$, $I+Q\overset{d}{\to} \Tilde{I}$ after an initial time layer of size $\mathcal{O}( \varepsilon)$, where $(\Tilde{I})_{\geq0}$ constitutes a $SIS$-process with rates 
\begin{align*}
    \Tilde{\beta} = (1-\zeta)\beta, &\quad \Tilde{\gamma} = (1-\zeta)\gamma
\end{align*}
\textbf{Proof of the claim:}
This proof is based on the monograph by Yin \cite{Yin.2013}, which is concerned time scale separation for Markov processes. 
Denote 
$$r_l(t):=P(I+Q=l)=\sum_{i=0}^l s_i^{(l)}:=\sum_{i=0}^l P(I=i,Q=l-i).$$
With the understanding that $s_i^{(l)}=0$ outside the feasible index combinations $0\leq i \leq l \leq N$, the master equation can be rewritten in terms of the sub-chains to
\begin{align*}
    \varepsilon \frac{d}{dt}s_i^{(l)}
    = &-\left(\varepsilon\beta \frac{i (N-l)}{N} + \varepsilon\gamma i + \varepsilon\mu (l-i) + \hat{\omega}i + \hat{c}(l-i)  \right)s_i^{(l)}\\
    &+\varepsilon\beta \frac{(i-1) (N-l+1)}{N}s_{i-1}^{(l-1)} + \varepsilon\gamma (i+1) s_{i+1}^{(l+1)} + \varepsilon\mu (l-i+1)s_{i}^{(l+1)}\\
    &+\hat{\omega}(i+1)s_{i+1}^{(l)} + \hat{c}(l-i+1)s_{i-1}^{(l)}
\end{align*}
By the Fenichel Theory (E.g. chapter 5 in \cite{Muller.2013} and citations therein), for small $\varepsilon$, the solution for $s_i^{(l)}$ settles on the slow manifold after an initial time layer of size $\mathcal{O}(\varepsilon)$. The slow manifold is given by 
\begin{align*}
    0 \overset{!}{=} \lim_{\varepsilon\to 0} \frac{d}{dt} s_i^{(l)} = \big( -\hat{\omega}i + \hat{c}(l-i) \big)s_i^{(l)}  +\hat{\omega}(i+1)s_{i+1}^{(l)} + \hat{c}(l-i+1)s_{i-1}^{(l)}.
\end{align*}
The detailed balance condition dictates
\begin{align*}
    s_i^{(l)}=\frac{\hat{c}(l-i+1)}{\hat{\omega}i} s_{i-1}^{(l)}
    =\prod_{k=1}^{i} \frac{\hat{c}(l-k+1)}{\hat{\omega}k} s_{0}^{(l)}
    =\left( \frac{\hat{c}}{\hat{\omega}} \right)^{i} \frac{l!}{(l-i)!i!} s_{0}^{(l)}
    =\left( \frac{\hat{c}}{\hat{\omega}}\right)^{i} \binom{l}{i} s_{0}^{(l)}
\end{align*}
and by definition of $r_l$
\begin{align*}
    s_{0}^{(l)}=r_l \left[ \sum_{i=0}^l \left( \frac{\hat{c}}{\hat{\omega}}\right)^{i} \binom{l}{i} \right]^{-1},
\end{align*}
which leads to
\begin{align*}
    \sum_{i=0}^l i s_i^{(l)}
    &=r_l\frac{\sum_{i=0}^l i \left( \frac{\hat{c}}{\hat{\omega}}\right)^{i} \binom{l}{i}}{\sum_{i=0}^l \left( \frac{\hat{c}}{\hat{\omega}}\right)^{i} \binom{l}{i}}
    = r_l \frac{\sum_{i=0}^l i x^{i} \binom{l}{i}}{\sum_{i=0}^l x^{i} \binom{l}{i}}\Bigg|_{x=\frac{\hat{c}}{\hat{\omega}}}
    = r_l \frac{x \frac{d}{dx}(1+x)^l}{(1+x)^l}\Bigg|_{x=\frac{\hat{c}}{\hat{\omega}}}\\
    &= r_l l\frac{x(1+x)^{l-1}}{(1+x)^l}\Bigg|_{x=\frac{\hat{c}}{\hat{\omega}}}
    = r_l l \frac{\frac{\hat{c}}{\hat{\omega}}}{1+\frac{\hat{c}}{\hat{\omega}}}
    = r_l l (1-\zeta)
\end{align*}
and similarly $\sum_{i=0}^l (l-i) s_i^{(l)}= r_l l \zeta$. Now consider the ODE for the subchains $r_l$ (Note that the jumps between $I$ and $Q$ do not appear here as they don't leave $r_l$) and plug in the expression for $\sum_{i=0}^l i s_i^{(l)}$:
\begin{align*}
    \frac{d}{dt}r_l = &-\sum_{i=0}^l \left(\beta \frac{i (N-l)}{N} + \gamma i + \mu (l-i)  \right)s_i^{(l)}\\
    &+\sum_{i=0}^l \beta \frac{(i-1) (N-l+1)}{N}s_{i-1}^{(l-1)} + \sum_{i=0}^l \gamma (i+1) s_{i+1}^{(l+1)} + \sum_{i=0}^l \mu (l-i+1)s_{i}^{(l+1)}\\
    = &- \left(\beta \frac{ (N-l)}{N} + \gamma \right) \sum_{i=0}^l i s_i^{(l)} -\mu\sum_{i=0}^l  (l-i)s_i^{(l)} \\
    &+ \beta \frac{ (N-l+1)}{N}\sum_{i=0}^{l-1} i s_i^{(l-1)}
     + \gamma \sum_{i=0}^{l+1} i s_i^{(l+1)}
     + \mu \sum_{i=0}^{l+1} (l-i)s_i^{(l+1)} \\
    = &- \left(\beta \frac{ (N-l)}{N} + \gamma \right) (1-\zeta)r_l l -\mu \zeta r_l l\\
    &+ \beta \frac{ (N-l+1)}{N}  (1-\zeta)r_{l-1} (l-1)
     + \gamma (1-\zeta)r_{l+1} (l+1)
     + \mu  \zeta r_{l+1} (l+1)\\
     = &- \left(\Tilde{\beta} \frac{ (N-l)}{N} + \Tilde{\gamma} +\Tilde{\mu} \right) r_l l 
    + \Tilde{\beta} \frac{ (N-l+1)}{N}  r_{l-1} (l-1)
     + (\Tilde{\gamma}+ \Tilde{\mu}) r_{l+1} (l+1)
\end{align*}
That is, $r_l$ satisfies for small $\varepsilon$ approximately the master equation for a $SIS$-process with parameters $\Tilde{\beta}=(1-\zeta)\beta, \Tilde{\gamma}=(1-\zeta)\gamma+\zeta\mu$, and the proof is concluded.\par\medskip

\textbf{Extension to two strains:} From this proof we can infer the generalization to the two-strain model: Since there is no direct interaction between the sub-chains $\{ I_1+Q_1=l\}$ and $\{ I_2+Q_2=k\}$, their quasi-equilibrium-distributions $P_{r^{(1)}_l}(I_1=i,Q_1=l-i)$ is the same as if the other strain was absent. Again, since there aren't any interaction terms, plugging the terms into the master equation of $r_l$, just as we've done in the one-strain model, yields the result.
\par\qed\par\medskip

\subsection{Kogan's reduction method}
\label{kogan}
We first describe the steps required to obtain the reduced Fokker-Planck equation, 
and subsequently we go through the steps, one by one. 

\subsubsection{Algorithmic description of Kogan's method}
In this description we closely follow \cite{Kogan.2014}, \cite[subsection 4.4]{Majda2001}, and \cite{Heinrich.2018}. 
To start with, we assume that the Fokker-Planck equation can be written as $\rho_t = L_0\rho+\eps L_1\rho+\eps^2L_2\rho+$higher order term, where the operators $L_i$ have the following form
    \begin{equation}
            \begin{aligned}\label{formL}
            L^0 &= \sum_{i,j=1}^{3} \Psi_{ij}(X) \partial_{Y_i}Y_j + \sum_{i\geq j = 1}^3 \Gamma_{ij}(X) \partial_{Y_i Y_j}^2\\
            L^1 &= \partial_X \left[ \sum_{i=1}^{3} \Xi_i(X) Y_i \right] + \sum_{i=1}^3 \partial_{Y_i}(...)
    \end{aligned}
    \end{equation}
and that there exists a one-dimensional line CL parameterized by $X$ where the noise vanishes, i.e. $CL = \{ (X,Y):Y=0 \}$.
The method consists of the following steps. It is instructive to omit some details here and to provide just a sketch, as the procedure will be exercised in full detail below. 
\begin{enumerate}
    \item We use the form of the Fokker-Planck equation 
        $\partial_t \rho=(L_0+\varepsilon L_1 + \varepsilon^2 L_2 )\rho$ 
    and change to the slow time scale $\tau=\varepsilon^2 t$. Then, use the Ansatz 
    \begin{align*}
        \rho(X,Y,t)&=\rho_0(X,Y,\tau)+\varepsilon \rho_1(X,Y,\tau) + \varepsilon^2\rho_2(X,Y,\tau).
    \end{align*}
    Comparing orders of $\varepsilon$ results in the 3 equations 
    \begin{align*}
     L_0 \rho_0 &= 0,\\
     L_0 \rho_1 &= - L_1 \rho_0,\\
     L_0 \rho_2 &= \partial_\tau \rho_0 - L_1 \rho_1 - L_2 \rho_0.
    \end{align*}
        \item The first equation yields $\rho_0(X,Y,\tau) = f(X,\tau) \mathcal{N}(X) \mathrm{e}^{-\frac{1}{2}\sum_{i,j=1}^d C_{i,j}Y_iY_j}$, implying that near the CL, the noise follows a Gaussian distribution with $X$-dependend mean and variance. $\mathcal{N}(X)$ is the corresponding normalizing constant. The slow evolution of X is encoded in the $Y$-independent $f$. Integration of the third equation gives the equation
        \begin{align*}
           \int L_0 \rho_2 dY + \int L_1 \rho_1 dY +\int L_2 \rho_0 dY  
             = \int\partial_\tau \rho_0 dY =  \partial_{\tau} f,
        \end{align*}
        where integration is carried out over $\R^d$  with respect to $Y_1,...,Y_d$.
        \item         \begin{enumerate}
            \item By assumption, $L_0$ contains only full derivatives w.r.t. $Y$, so $\int L_0 \rho_2 dY=0$
            \item $\int L_2 \rho_0 dY$ boils down to calculating second moments of the Gaussian from $\rho_0$.
        \end{enumerate}
    \item This is the most delicate step. $\int L_1 \rho_1 dY$ involves the calculation of $\rho_1$ at first sight, which in higher dimensions is a difficult task. Kogan proposes to find a function $F(X,Y)=A_1Y_1+...+A_dY_d$ s.t.
    \begin{align*}
        \int L_1 \rho_1 dY = \partial_X \int F(X,Y) L_1 \rho_0 dY.
    \end{align*}
    With $F$, the integral can be expressed using the adjoint $L_1^+$ w.r.t. $Y$ and $\rho_0$:
    \begin{align*}
        \int L_1 \rho_1 dY = -\partial_X \int \left( L_1^+ F(X,Y) \right) \rho_0 dY.
    \end{align*}
	Thereto we observe that in relevant cases $L^{(1)} = \partial_X\sum_i \sum_i Y_i + \mbox{total derivatives in }Y_i$.
	If we are able to solve $L_0^+F(X,Y) = \sum_i \sum_i Y_i$, we can write, using $L_0 \rho_1 = - L_1 \rho_0$
	    \begin{align*}
        \int L_1 \rho_1 dY = \partial_X \int L_0^b+(F(X,Y)) \rho_1 dY = -\partial_X \int F(X,Y) L_0\rho_1 dY 
		= \partial_X \int F(X,Y) L_1 \rho_0 dY.
    \end{align*}
    \item Adding up the results from steps 3 and 4 gives the desired approximate Fokker-Planck equation. 
\end{enumerate}
The computations are mostly straightforward but rather involving, and have been checked or performed using 
the computer algebra package maxima~\cite{Maxima.2023}.

\paragraph{Step 1: Operators $L_0$, $L_1$ and $L_2$}
We transform the coordinate system for the  Fokker-Planck equation ~\eqref{eq: FPorig SIIQQS} such that 
one coordinate $X$ is directed along the line of stationary points, and the other three 
coordinates ($Y_1$, $Y_2$, $Y_3$) span the space perpendicular to the line of stationary points. 
The small parameter is again $\varepsilon=\frac{1}{\sqrt{N}}$. We express the system in the slow and fast variables 
(recall $\vartheta_i = = \hat\omega_i/\hat c_i$)
\begin{align*}
    X &:= \frac{1}{2r}\left((1+\vartheta_1) x - (1+\vartheta_2) y +r\right) \\
    Y_1 &:= \frac{1}{2 \varepsilon}\left((1+\vartheta_1) x + (1+\vartheta_2) y - r\right) \\
    Y_2 &:= \frac{1}{2 \varepsilon}\left(u+v-\vartheta_1 x -\vartheta_2 y\right) \\
    Y_3 &:= \frac{1}{2 \varepsilon}\left(u-v-\vartheta_1 x +\vartheta_2 y\right).
\end{align*}
 Vice versa it is
\begin{align*}
    x &= \frac{r X+\varepsilon Y_1}{1+\vartheta_1} \\
    u &= \vartheta_1 \frac{r X+\varepsilon Y_1}{1+\vartheta_1} +\varepsilon \left(Y_2+Y_3\right) \\
    y &= \frac{r(1-X)+\varepsilon Y_1}{1+\vartheta_2}  \\
    v &= \vartheta_2 \frac{r(1-X)+\varepsilon Y_1}{1+\vartheta_2} + \varepsilon \left(Y_2-Y_3\right)
\end{align*}
As the system is 4-dimensional, we have three fast variables.
The slow variable $X$ measures the progress along the CL, whereas the fast variables $Y_1,Y_2,Y_3$ represent shot noise. The goal is to reduce the 4-dimensional system to a 1-dimensional Fokker-Planck equation for $X$, which will then be used to infer information such as fixation probabilities and mean time to fixation. It is convenient to induce scaled variables $\omega_i$ and $c_i$, 
$$
\omega_i = \hat\omega_i/(\beta r),
\qquad 
c_i = \hat c_i/(\beta r).
$$
In terms of the transformed variables we get in the slow timescale $\Tilde{\tau} = \beta r \varepsilon^2 t$ 
\begin{align*}
    \partial_{\Tilde{\tau}} &\rho (X,Y_1,Y_2,Y_3,\Tilde{\tau}) =\\
    &\left[\frac{\varepsilon}{r}\partial_{X}+\partial_{Y_1}-\zeta_1\partial_{Y_2}-\zeta_1\partial_{Y_3}\right] 
    \left[ (Y_1+Y_2)( X+\varepsilon r^{-1} Y_1)\rho\right]\\
    &\left[-\frac{\varepsilon}{r}\partial_{X}+\partial_{Y_1}-\zeta_2\partial_{Y_2}+\zeta_2\partial_{Y_3}\right]
    \left[(Y_1+Y_2)((1-X)+\varepsilon r^{-1}Y_1)\rho\right]\\
    &-\left[\frac{\varepsilon}{r}\partial_{X}+\partial_{Y_1}-\partial_{Y_2}-\partial_{Y_3}\right]
    \left[\frac{\omega_1+c_1}{2}(Y_2+Y_3)\rho\right]\\
    &-\left[-\frac{\varepsilon}{r}\partial_{X}+\partial_{Y_1}-\partial_{Y_2}+\partial_{Y_3}\right]
    \left[\frac{\omega_2+c_2}{2}(Y_2-Y_3)\rho\right]\\
    &+\frac{1}{4}\left[\frac{\varepsilon}{r}\partial_{X}+\partial_{Y_1}-\zeta_1\partial_{Y_2}-\zeta_1\partial_{Y_3}\right]^2
    \left[(1+\vartheta_1)((1-r)-\varepsilon(Y_1+Y_2))( X+\varepsilon r^{-1}Y_1)\rho\right]\\
    &+\frac{1}{4}\left[-\frac{\varepsilon}{r}\partial_{X}+\partial_{Y_1}-\zeta_2\partial_{Y_2}+\zeta_2\partial_{Y_3}\right]^2
    \left[(1+\vartheta_2)((1-r)-\varepsilon(Y_1+Y_2))((1-X)+\varepsilon r^{-1}Y_1)\rho\right]\\
    &+\frac{1}{4}\left[\frac{\varepsilon}{r}\partial_{X}+\partial_{Y_1}-\partial_{Y_2}-\partial_{Y_3}\right]^2
    \left[(1+\vartheta_1) \omega_1(r X+\varepsilon Y_1)+\varepsilon(\frac{c_1+\omega_1}{2})(Y_2+Y_3))\rho\right]\\
    &+\frac{1}{4}\left[-\frac{\varepsilon}{r}\partial_{X}+\partial_{Y_1}-\partial_{Y_2}+\partial_{Y_3}\right]^2
    \left[(1+\vartheta_2)\omega_2(r(1-X)+\varepsilon Y_1)+\varepsilon(\frac{c_2+\omega_2}{2})(Y_2-Y_3))\rho\right]
\end{align*}

Next, expand w.r.t. $\varepsilon$ s.t.
\begin{align}\label{expansionrho}
    \partial_{\Tilde{\tau}} \rho = \varepsilon^0 L_0 \rho + \varepsilon^1 L_1 \rho + \varepsilon^2 L_2 \rho.
\end{align}
and obtain (note that full second derivatives in $Y_i$ in $L_1$ and $L_2$ 
will later be integrated away and therefore are not necessary to state)
\begin{equation}\label{eq: L0 SIQS}
    \begin{aligned}
    L_0\rho=&\partial_{Y_1} \left[ \bigg((Y_1+Y_2) - \frac{\omega_1+c_1}{2}(Y_2+Y_3) - \frac{\omega_2+c_2}{2}(Y_2-Y_3)\bigg)  \rho\right]\\
            +&\partial_{Y_2} \left[\bigg(- (\zeta_1 X+\zeta_2 (1-X))(Y_1+Y_2) + \frac{\omega_1+c_1}{2}(Y_2+Y_3) + \frac{\omega_2+c_2}{2}(Y_2-Y_3)\bigg) \rho \right]\\
            +&\partial_{Y_3} \left[\bigg( - (\zeta_1 X+\zeta_2 (1-X))(Y_1+Y_2) + \frac{\omega_1+c_1}{2}(Y_2+Y_3) - \frac{\omega_2+c_2}{2}(Y_2-Y_3)\bigg) \rho \right]\\
    +&\frac{1}{4}\left[\partial_{Y_1}-\zeta_1\partial_{Y_2}-\zeta_1\partial_{Y_3}\right]^2
    \left[(1+\vartheta_1)(1-r) X\rho\right]\\
    +&\frac{1}{4}\left[\partial_{Y_1}-\zeta_2\partial_{Y_2}+\zeta_2\partial_{Y_3}\right]^2
    \left[(1+\vartheta_2)(1-r)(1-X)\rho\right]\\
    +&\frac{1}{4}\left[\partial_{Y_1}-\partial_{Y_2}-\partial_{Y_3}\right]^2
    \left[(1+\vartheta_1) \omega_1 r X\rho\right]\\
    +&\frac{1}{4}\left[\partial_{Y_1}-\partial_{Y_2}+\partial_{Y_3}\right]^2
    \left[(1+\vartheta_2)\omega_2 r(1-X)\rho\right]
    \end{aligned}       
\end{equation}
and
\begin{equation}\label{eq: L1 SIQS}
    \begin{aligned}
        L_1 = -&\frac{1}{r} \partial_X\left[\left( (1-2X)(Y_1+Y_2)+\frac{\omega_1}{2\zeta_1}(Y_2+Y_3)-\frac{\omega_2}{2\zeta_2}(Y_2-Y_3) \right)\rho \right]\\
        +&\left[ 2\partial_{Y_1}-(\zeta_1+\zeta_2)\partial_{Y_2}+(\zeta_2-\zeta_1)\partial_{Y_3} \right] \left[ r^{-1} (Y_1+Y_2)Y_1 \rho \right]\\
        +& \frac{1}{2} \partial_{XY_1}^2 \left[\bigg( (1+\vartheta_1)(\gamma/(\beta r)+\omega_1)X-(1+\vartheta_2)(\gamma/(\beta r)+\omega_2)(1-X)\bigg) \rho \right]\\
        -& \frac{1}{2}\left[ \partial_{X}(\partial_{Y_2}+\partial_{Y_3}) \right] \left[ (1+\vartheta_1)(\omega_1+\zeta_1\gamma/(\beta r))X \rho \right]\\
        +& \frac{1}{2}\left[ \partial_{X}(\partial_{Y_2}-\partial_{Y_3}) \right] \left[ (1+\vartheta_2)(\omega_2+\zeta_2\gamma/(\beta r))(1-X) \rho \right]\\
        +& \partial_{Y_1 Y_1}(...) + \partial_{Y_1 Y_2}(...) +  \partial_{Y_1 Y_3}(...) +  \partial_{Y_2 Y_2}(...)  + \partial_{Y_2 Y_3}(...)  \partial_{Y_3 Y_3}(...) 
    \end{aligned}
\end{equation}
and (recall that $r=1-1/R_0=1-\gamma/\beta$) 
\begin{equation}\label{eq: L2 SIQS}
    \begin{aligned}
        L_2 = &\partial_{XX}^2 \frac{1}{4r} \left[ (1+\vartheta_1)(\omega_1+\gamma/(\beta r))X +(1+\vartheta_2)(\omega_2+\gamma/(\beta r))(1-X)  \right]\\
        &+ \partial_{Y_1}(...) + \partial_{Y_2}(...) +  \partial_{Y_3}(...)
    \end{aligned}
\end{equation}

\paragraph{Step 2: Equation for $f(X,\tau)$}
Recall that 
\begin{align}
    \label{1equation} L_0 \rho_0 &= 0\\
    \label{2equation} L_0 \rho_1 &= - L_1 \rho_0\\
    \label{3equation} L_0 \rho_2 &= \partial_{\Tilde{\tau}} \rho_0 - L_1 \rho_1 - L_2 \rho_0
\end{align}
From (\ref{1equation}) we infer $\rho_0 = f(X,\Tilde{\tau}) \mathcal{N}(X) \exp \left( {\sum_{i,j=1}^3 c_{ij}(X) Y_i Y_j} \right)$, that is, for fixed $X$ the fast variables follow a 3-dimensional Gaussian distribution with $X$-dependend mean and variance. We will disregard here the explicit calculation of the coefficients $c_{ij}(X)$. As we are interested in the evolution of the slow variable, we aim at $f(X,\Tilde{\tau})$. Plugging in $\rho_0$ with yet unknown $f$ and integrating, (\ref{3equation}) gives
\begin{align}
    \partial_{\Tilde{\tau}} f  
    = \int L_0 \rho_2 dY - \int L_1 \rho_1 dY - \int L_2 \rho_0 dY
\end{align}
where $\int ... dY$ is shorthand for $\int_{\R^3} ... d(Y_1, Y_2, Y_3)$. Solving these three integral will leave us with the 1-dimensional Fokker-Planck equation for $X$.

\paragraph{Step 3: Handling the integral over $L_0$ and $L_2$}
As $L_0$ contains only full derivatives w.r.t. the $Y_i$'s, the first integral from above equation is zero:
\begin{align}\label{eq: int L0rho2 SIQS}
    \int L_0 \rho_2 dY = 0.
\end{align}
For the same reason, in the third integral we only need the derivatives w.r.t. $X$ from $L_2$. Since they don't contain fast variables (\ref{eq: L2 SIQS}), the integral's solution can be read off to be:
\begin{align}\label{eq: int L2rho0 SIQS}
    \int L_2 \rho_0 dY = 
       \partial_{XX}^2 \frac{1}{4r} \left[ (1+\vartheta_1)(\omega_1+\gamma/(\beta r))X +(1+\vartheta_2)(\omega_2+\gamma/(\beta r))(1-X)  \right] f(X, \Tilde{\tau}).
\end{align}

\paragraph{Step 3: Handling the integral over $L_1$}
{\it Part 1:} The tricky step is solving the integral $\int L_1 \rho_1 dY$, for which one would need to solve (\ref{2equation}) for $\rho_1$ - an intricate problem in our four-dimensional system. Kogan circumvents this direct computation in determining a function 
$F(X,Y)=A_1(X)Y_1+A_2(X)Y_2+A_3(X)Y_3$ such that 
$$\int L_1 \rho_1 dY = \partial_X \int F(X,Y) L_1 \rho_0 dY.$$
As we will see, $ L_1^+ F$ is a linear combination of $\{ Y_{i}Y_j :i,j=1,...,3\}$.\par\medskip 

The central task is to explicitly determine this function $F$ s.t. $L_0^+ F = \Xi_1 Y_1 + \Xi_2 Y_2 +\Xi_3 Y_3$. 
We test the form $F(X,Y)=A_1(X)Y_1+A_2(X)Y_2+A_3(X)Y_3$. Then, $L_0^+ F = \Xi_1 Y_1 + \Xi_2 Y_2 +\Xi_3 Y_3$ is equivalent to solving the linear problem $\Psi^TA=-\Xi$, 
where $\Psi=(\Psi_{ij})_{i,j=1,...,3}$ and $\Xi=(\Xi_1,...,\Xi_3)^T$ from (\ref{formL}). The calculation can be done via Maxima. We find
\begin{align*}
    A_1 & = A_2 = \frac{1}{r}\,\, \frac{1-2X+\zeta_1 X -\zeta_2 (1-X)}{1-\zeta_1X-\zeta_2(1-X)} \quad \text{and} \quad
    A_3 = \frac{1}{r}.
\end{align*}
Summary: With 
$$ F(X,Y) = 
 \frac{1}{r}\,\, \frac{1-2X+\zeta_1 X -\zeta_2 (1-X)}{1-\zeta_1X-\zeta_2(1-X)}\,\,(Y_1+Y_2) + \frac 1 r \, Y_3$$
 we find $L_0^+F(X,Y)= \sum_{i=1}^3 \Xi_i Y_i$ and thus 
\begin{align*}
    \int (L_1^+ F) \rho_0 dY = \int F (L_1 \rho_0 )dY
\end{align*}

{\it Part 2:} Next, we determine $ \int (L_1^+ F) \rho_0 dY $ in integrating the r.h.s. of 
\begin{align*}
    \int (L_1^+ F) \rho_0 dY = \int F (L_1 \rho_0 )dY.
\end{align*}
The integral on the r.h.s. involves four types of terms: $F\partial_{Y_i}(...)$, $F\partial_X(...)$, $F\partial_{XY_i}(...)$ and $F\partial_{Y_iY_j}(...)$. These can be handled as follows. For suitable functions $g(X,Y),h(X,Y)$, integration by parts for the partial derivatives $\partial_{Y_i}$
\begin{align*}
    \int h(X,Y) \left[ \partial_{Y_i}g(X,Y)\rho_0 \right] dY 
    &= - \int \left[\partial_{Y_i} h(X,Y) \right] g(X,Y) \rho_0 dY,\\\text{and}
        \int h(X,Y) \left[ \partial_{Y_iY_j}^2g(X,Y)\rho_0 \right] dY 
    &=  \int \left[\partial_{Y_iY_j}^2 h(X,Y) \right] g(X,Y) \rho_0 dY,
\end{align*}
the product rule for $\partial_X$ and a combination of both for $\partial_{X Y_1}^2$
\begin{align*}
    \int h(X,Y) \left[ \partial_{XY_i}^2g(X,Y)\rho_0 \right] dY 
    &= - \int \left[\partial_{Y_i} h(X,Y) \right] \partial_X g(X,Y) \rho_0 dY\\
    &= -\partial_X \int \left[\partial_{Y_i} h(X,Y) \right]  g(X,Y) \rho_0 dY \\
    &\quad+ \int \left[\partial_{XY_i}^2 h(X,Y) \right] g(X,Y) \rho_0 dY
\end{align*}
Using the explicit form of $L_1$ (\ref{eq: L1 SIQS}) we get
\begin{equation}
    \begin{aligned}
        \int (&L_1^+ F) \rho_0 dY \\
        = \frac{1}{r}&\int (\partial_XF) \left( (1-2X)(Y_1+Y_2)+\frac{\omega_1}{2\zeta_1}(Y_2+Y_3)-\frac{\omega_2}{2\zeta_2}(Y_2-Y_3) \right)\rho_0dY\\
        -\frac{1}{r}\partial_X&\int F \left( (1-2X)(Y_1+Y_2)+\frac{\omega_1}{2\zeta_1}(Y_2+Y_3)-\frac{\omega_2}{2\zeta_2}(Y_2-Y_3) \right)\rho_0dY\\
         - \frac{1}{r}&\int \left( 2A_1-(\zeta_1+\zeta_2)A_2+(\zeta_2-\zeta_1)A_3 \right) (Y_1+Y_2)Y_1 \rho  dY\\
         + \frac{1}{2}&\int (\partial_X A_1)\left( (1+\vartheta_1)(\gamma/(\beta r)+\omega_1)X-(1+\vartheta_2)(\gamma/(\beta r)+\omega_2)(1-X)  \right)\rho_0dY\\
         - \frac{1}{2}\partial_X&\int A_1\left( (1+\vartheta_1)(\gamma/(\beta r)+\omega_1)X-(1+\vartheta_2)(\gamma/(\beta r)+\omega_2)(1-X)  \right)\rho_0dY\\
         + \frac{1}{2} &\int -\partial_X(A_2+A_3) (1+\vartheta_1)(\omega_1+\zeta_1\gamma/(\beta r))X \rho + \partial_X(A_2-A_3) (1+\vartheta_2)(\omega_2+\zeta_2\gamma/(\beta r))(1-X) \rho dY\\
         -\frac{1}{2}\partial_X &\int -(A_2+A_3) (1+\vartheta_1)(\omega_1+\zeta_1\gamma/(\beta r))X \rho + (A_2-A_3) (1+\vartheta_2)(\omega_2+\zeta_2\gamma/(\beta r))(1-X) \rho dY
    \end{aligned}
\end{equation}
this can be simplified a little using $A_1=A_2$ and $\partial_X A_3 =0$ to
\begin{equation}\label{eq: L1+F vorletzte gleichung SIQS}
    \begin{aligned}
        \int (&L_1^+ F) \rho_0 dY \\
        = \frac{1}{r}&\int (\partial_XA_1) (Y_1+Y_2) \left( (1-2X)(Y_1+Y_2)+\frac{\omega_1}{2\zeta_1}(Y_2+Y_3)-\frac{\omega_2}{2\zeta_2}(Y_2-Y_3) \right)\rho_0dY\\
         -  \frac{1}{r}&\int \left( 2A_1-(\zeta_1+\zeta_2)A_2+(\zeta_2-\zeta_1)A_3 \right) (Y_1+Y_2)Y_1 \rho_0 dY\\
         + \frac{\gamma}{2\beta r}&\int (\partial_X A_1)\left( (1+\vartheta_1)(1-\zeta_1) X-(1+\vartheta_2)(1-\zeta_2)(1-X)  \right)\rho_0dY\\
        -\frac{1}{r}\partial_X&\int F \left( (1-2X)(Y_1+Y_2)+\frac{\omega_1}{2\zeta_1}(Y_2+Y_3)-\frac{\omega_2}{2\zeta_2}(Y_2-Y_3) \right)\rho_0dY\\
         - \frac{\gamma}{2\beta r}\partial_X&\int A_1\left( (1+\vartheta_1)(1-\zeta_1) X-(1+\vartheta_2)(1-\zeta_2)(1-X)  \right)\rho_0dY\\
         +\frac{1}{2}\partial_X &\int A_3 \left( (1+\vartheta_1)(\omega_1+\zeta_1\gamma/(\beta r))X - (1+\vartheta_2)(\omega_2+\zeta_2\gamma/(\beta r)(1-X) \right) \rho_0 dY\\
    \end{aligned}
\end{equation}
Now, note that all six integrands are either of zeroth or second degree w.r.t. the $Y_i$. For the zeroth degree terms of course there is nothing to integrate. We are left with integrals of the form $\int (...)Y_{i}Y_{j}\rho_0dY$. Replicating Proposition A.3 from Heinrich \cite{Heinrich.2018}, this can be reduced to a linear algebra problem. 
\begin{lemma}
    With the notation from (\ref{formL}) define
    \begin{align*}
      M(X) :=  -\begin{pmatrix}
            2\Psi_{1,1} & \Psi_{2,1} & \Psi_{3,1} & 0 & 0 & 0 \\
            2\Psi_{1,2} & (\Psi_{1,1}+\Psi_{2,2}) & \Psi_{3,2} & 2\Psi_{2,1} & \Psi_{3,1} & 0 \\
            2\Psi_{1,3} & \Psi_{2,3} & (\Psi_{1,1}+\Psi_{3,3}) & 0 & \Psi_{2,1} & 2\Psi_{3,1} \\
            0 & \Psi_{1,2} & 0 & 2\Psi_{2,2} & \Psi_{3,2} & 0 \\
            0 & \Psi_{1,3} & \Psi_{1,2} & 2\Psi_{2,3} & (\Psi_{2,2}+\Psi_{3,3}) & 2\Psi_{3,2} \\
            0 & 0 & \Psi_{1,3} & 0 & \Psi_{2,3} & 2\Psi_{3,3}
            \end{pmatrix}
\end{align*}
and
\begin{align*}
    \Gamma(X) = \begin{pmatrix}
        2 \Gamma_{1,1} & \Gamma_{1,2} & \Gamma_{1,3} & 2\Gamma_{2,2} &\Gamma_{2,3} & 2\Gamma_{3,3}
    \end{pmatrix}^T.
\end{align*}
Assume that $M$ is invertible. Consider the space $\mathcal{H}_2=\{ \sum_{i,j=1}^3 c_{ij}(X) Y_iY_j  \}$ of all homogeneous polynomials of degree 2 in $Y_1,Y_2,Y_3$ and use $(Y_1^2,Y_1Y_2, Y_1Y_3,Y_2^2,Y_2Y_3,Y_3^2)$ as a basis. Then, for any $h\in\mathcal{H}_2$, 
\begin{align}\label{eq: Int L1rho1}
    \int h(X,Y) \rho_0(X,Y,\Tilde{\tau}) dY = -f(X,\Tilde{\tau}) \langle \Gamma (X) ,  z \rangle,
\end{align}
where $z$ is a solution to $Mz=\Vec{h}$ and $\Vec{h}$ is the vector-representation of
$h$ w.r.t. to the basis. 
\end{lemma}
\begin{proof}
    Due to the structure of $L_0$ (\ref{formL}), for any $h\in\mathcal{H}_2$ we have $L_0^+h = h_0 + c_0 \in \mathcal{H}_2 \oplus \mathcal{H}_0 $. More precisely, $h_0=M(X)\Vec{h}$ as well as $c_0=\langle \Gamma (X) ,  \Vec{h} \rangle$. If $z$ is a solution to $Mz=\Vec{h}$, then
\begin{align*}
    \int h(X,Y) \rho_0(X,Y,\Tilde{\tau}) dY 
    &= \int Mz \rho_0(X,Y,\Tilde{\tau}) dY \\
    &= \int \big( L_0^+z\big)\rho_0(X,Y,\Tilde{\tau}) dY - \int \big( \langle \Gamma (X) ,  z \rangle \big)\rho_0(X,Y,\Tilde{\tau}) dY \\
    &= 0-f(X,\Tilde{\tau}) \langle \Gamma (X) ,  z \rangle.
\end{align*}
\end{proof}
For our case, numerical analysis reveals that, for our parameter ranges, $M$ is indeed invertible and thus such $z$ can be found.
Now it is a matter of collecting the coefficients of $Y_1^2,Y_1Y_2, Y_1Y_3,Y_2^2,Y_2Y_3,Y_3^2$ of the integrands in (\ref{eq: L1+F vorletzte gleichung SIQS}) and solving the resulting linear problem. This can be done by Maxima. We do not present the (lengthy) result, 
but refer to the maxima script (subsection~\ref{maxCode}).

\paragraph{Step 4: The explicit equation for $f(X,\tau)$}
Combining (\ref{eq: int L0rho2 SIQS}), (\ref{eq: int L2rho0 SIQS}) and the solution to (\ref{eq: Int L1rho1}) for our case, obtained by Maxima~\cite{Maxima.2023} (see SI), finally gives
\begin{equation}
    \begin{aligned}
        \partial_{\Tilde{\tau}} &f(X,\Tilde{\tau}) =\\ 
        \partial_X &\left[ \frac{\beta r^2}{\gamma} (1-\zeta_1)(1-\zeta_2)\frac{\zeta_1\zeta_2 (\omega_1-\omega_2) \alpha(X)+\omega_1\omega_2(\zeta_1-\zeta_2) }{ \alpha(X)^2 \left[\omega_1+\zeta_1 \alpha(X) \right]\left[\omega_2+\zeta_2 \alpha(X) \right]}  X(1-X) f(\Tilde{\tau},X) \right]\\
        +\partial_{XX} &\left[ \frac{\beta r^2}{\gamma}  \frac{(1-\zeta_1)(1-\zeta_2)}{\alpha(X)}X(1-X) f(\Tilde{\tau}, X) \right]
    \end{aligned}
\end{equation}
where $\alpha(X) := 1-\zeta_2+(\zeta_2-\zeta_1)X$. After a finale rescale of time $\tau = \gamma/(\beta r^2) \Tilde{\tau}$. 
Recall that $\Tilde{\tau} = \beta r \varepsilon^2 t$, such that $\tau= = \gamma \varepsilon^2 t/r$ the one-dimensional Fokker-Planck is independent of $r$ and reads
\begin{equation}
    \begin{aligned}\label{FP1dimAppendix}
     \partial_{\tau} f(X,\tau) = -\partial_X v(X) f(X,\tau) + \frac{1}{2} \partial_{X}^2 D(X) f(X,\tau)
    \end{aligned}
\end{equation} 
with effective drift and diffusion terms
\begin{align}\label{Drift1dimSI}
    v(X) &= -(1-\zeta_1)(1-\zeta_2)\frac{\zeta_1\zeta_2 (\omega_2-\omega_1) \alpha(X)
    +\omega_1\omega_2(\zeta_2-\zeta_1) }{ \alpha(X)^2 \left[\omega_1+\zeta_1 \alpha(X) \right]\left[\omega_2+\zeta_2 \alpha(X) \right]}  X(1-X)\\
    D(X) &= \frac{2(1-\zeta_1)(1-\zeta_2)}{ \alpha(X)}X(1-X).\label{Noise1dim}
\end{align}

The initial condition is given by the Dirac measure $f(X,0)=\delta(X-X_0)$, where $X_0$ is the fraction of the type-1 strain at the initial time. 

\subsubsection{Kogan's reduction for an SIS model}\label{KoganSIS}
For the convenience of the reader, we re-state Kogan's result for the tow-strain SIS-model, and apply his result to the SIS model which approximates the SIQS model.
\begin{theorem}[Kogan et al.~\cite{Kogan.2014}]
Consider a two-strain SIS model: For a finite population size $N$, we have $I_1$, $I_2$, and transitions $I_i\rightarrow I_i+1$ at rate $\beta_i\, I_iS/N$ where $S=N-I_1-I_2$, 
and $I_i\rightarrow I_i-1$ at rate $\gamma_i\, I_i$. Assume $R_0=\beta_1/\gamma_1=\beta_2/\gamma_2$, such that a line of stationary points appears in the deterministic limit ($N\rightarrow\infty$). 
The diffusion approximation allows an approximation of the process along the line of stationary points, given by 
$$\partial_{\tilde \tau} f(X,\tilde \tau) 
= 
\frac{2a(1-a)}{N} \partial_X\bigg(\frac{1-(X/r)^2}{[(1+a+(X/r)(1-a)]^2}\,f(X,\tilde \tau)\bigg) 
+
\frac{2ar)}{N} \partial_{XX}\bigg(\frac{1-(X/r)^2}{(1+a+(X/r)(1-a)}\,f(X,\tilde \tau)\bigg),
$$
where $\tilde \tau=\gamma_1\, t$, $X\in[-r,r]$, $r=1-1/R_0$ and $a=\gamma_2/\gamma_1$. 
\end{theorem}
We use $\eps^2=1/N$, and apply the transformation $\tilde X = (X/r+1)/2\in[0,1]$. Therewith, 
$$\partial_\tau f(\tilde X,\tau) 
= \eps^2 \frac{a(1-a)}{r} \partial_{\tilde X}\bigg(\frac{\tilde X(1-\tilde X)}{[a+(1-a)\tilde X]^2}\,f(\tilde X,t)\bigg) 
+\eps^2 \frac a r  \partial_{\tilde X\tilde X}\bigg(\frac{\tilde X(1-\tilde X)}{a+(1-a)\tilde X}\,f(\tilde X,t)\bigg).
$$
Rescaling time $\tau=\eps^2 \gamma_1 t/r$ gives us 
$$\partial_\tau f(\tilde X,\tau) 
= a(1-a) \partial_{\tilde X}\bigg(\frac{\tilde X(1-\tilde X)}{[a+(1-a)\tilde X]^2}\,f(\tilde X,\tau)\bigg) 
+ a  \partial_{\tilde X\tilde X}\bigg(\frac{\tilde X(1-\tilde X)}{a+(1-a)\tilde X}\,f(\tilde X,\tau)\bigg).
$$
Now we replace the parameters used by Kogan by the parameters of our SIS model, particularly we have 
$$ a = \frac{\gamma_2\zeta_2 + \hat\mu(1-\zeta_2)}{\gamma_1\zeta_1 + \hat\mu(1-\zeta_1)}$$
with  $\tilde \tau=\zeta_1\gamma_1\, t/r$

\subsection{Absorption probabilities}
\label{absorpProof}

\subsubsection{Sign of the drift in the one dimensional SIQS Fokker-Planck equation}\label{driftSign}

\begin{prop}\label{Prop: selectiontypesSI}
	Assume $\zeta_2 > \zeta_1$. If and only if $(1-\zeta_2) < -\frac{(\zeta_2-\zeta_1)\omega_1\omega_2}{\zeta_1\zeta_2(\omega_2-\omega_1)} < (1-\zeta_1)$, there exists $X_0 \in (0,1)$ with $v(X_0)=0$. For any such $X_0$ it is $v'(X_0)>0$. In this sense, the model may exhibit directional or disruptive selection but not balancing selection. 
\end{prop}
{\bf Proof: }
   First, $v(X)=0$ iff $\alpha(X)=-\frac{(\zeta_2-\zeta_1)\omega_1\omega_2}{\zeta_1\zeta_2(\omega_2-\omega_1)}$ and since $\alpha:[0,1] \xrightarrow{} [(1-\zeta_2),(1-\zeta_1)]$ is bijective, the first claim follows. Assuming there is a zero $X_0\in(0,1)$, then $\textbf{sgn}\left(v'(X_0))\right)=\textbf{sgn}\left(\frac{d}{dX}\frac{v(X)}{X(1-X)}\big|_{X=X_0}\right)$.

    Abbreviating the denominator of $v(X)$ in (\ref{Drift1dim}) as $d(\alpha(X))$, the sign of the derivative at $X_0$ can be determined as
    \begin{align*}
        \textbf{sgn}(v'(X_0))
    =\textbf{sgn}\Big(&-\frac{\alpha'(X_0)}{\big(d(\alpha(X_0))\big)^2}
    \Big[ 
        \zeta_1\zeta_2 (\omega_2-\omega_1)\alpha(X_0)\\
        &- \underbrace{\left(\zeta_1\zeta_2 (\omega_2-\omega_1)+(\zeta_2-\zeta_1)\omega_1\omega_2\alpha(X_0)\right)}_{=0 \text{ by def. of }X_0} d'(\alpha(X_0)) \Big]
    \Big)\\
        = \textbf{sgn}\big( &-\alpha'(X_0)(\omega_2-\omega_1)\zeta_1\zeta_2\alpha(X_0)\big) = \textbf{sgn}\big( -(\zeta_2-\zeta_1)(\omega_2-\omega_1) \big).
    \end{align*}
    Since the differences $(\zeta_2-\zeta_1)$ and $(\omega_2-\omega_1)$ must have different signs, $\textbf{sgn}(v'(X_0))>0$.
\par\qed\par\medskip

\subsubsection{Computation of the fixation probabilities} \label{prop: fix probs SIQSAppend}

Recall (see, e.g.~\cite{Gardiner.2009}) that the absorption probability $\pi(X)$  at $X=0$ for 
a stochastic deferential equation with Fokker-Planck equation
$$ u_t(x,t) = -\partial_x(v(x) u(x,t)) + \frac 1  2 \partial_{xx} (D(x)u(x,t))$$
starting with $u(x,0)=\delta_{X_0}(X)$ is given by the ODE 
\begin{align*}
	v(X) \pi'(X) + \frac{1}{2}D(X)\pi''(X)=0,\qquad 
 \pi(0)=1,\quad \pi(1)=0
\end{align*}
with the solution 
\begin{align}\label{FormulafixationprobSI}
	\pi(X_0)= \frac{\int_{X_0}^{1} \nu(y)dy}{\int_0^1 \nu(y) dy}, \quad \text{where } \nu(y) = \exp \left(- 2\int_0^y \frac{v(z)}{D(z)}dz \right).
\end{align}
We use this result to prove the following proposition.

\begin{prop}\label{prop: fix probs SIQS:SI}
\begin{enumerate}[label=(\alph*)]
\item If $\zeta_1=\zeta_2$ but $\omega_1\not=\omega2$, the probability that the second strain, starting at $X_0 \in (0,1)$, outcompetes the first strain 
reads
\begin{align*}
    \pi(X_0)= \frac{\mathrm{e}^{\mathcal{A}}-\mathrm{e}^{\mathcal{A}X_0}}{\mathrm{e}^{\mathcal{A}}-1},
\end{align*}
where $\mathcal{A}=\frac{ \zeta_1^2 (\omega_2-\omega_1)}{\left( \omega_1+(1-\zeta_1)\zeta_1 \right)\left(\omega_2+(1-\zeta_1)\zeta_1\right)}$. In particular $\pi(X_0)>1-X_0 \iff \omega_2>\omega_1$

\item If $\omega_1=\omega_2$ but $\zeta_1 \neq \zeta_2$, then 
\begin{align*}
    \pi(X_0)&= \frac{\mathfrak{A}(1-\zeta_1) - \mathfrak{A}(\alpha(X_0))}{\mathfrak{A}(1-\zeta_1) - \mathfrak{A}(1-\zeta_2)},
\end{align*}
where 
\begin{align*}
    \mathfrak{A}:[1-\zeta_2,1-\zeta_1] &\to \mathbb{R}_{>0}:
    s \mapsto (\omega_1 + \zeta_1 s)^{\frac{-\zeta_2}{\zeta_2-\zeta_1}} (\omega_1 + \zeta_2 s)^{\frac{\zeta_1}{\zeta_2-\zeta_1}}.
\end{align*}
$\mathfrak{A}$ is concave if $\zeta_2 > \zeta_1$ and convex if $\zeta_1 > \zeta_2$.
\end{enumerate}
\end{prop}
{\bf Proof: } We compute the fixation probability if the strains only differ in either $\zeta$ or $\omega$, and not in both parameters at a time. As indicated, this task boils down to the evaluation 
fo the integrals 
\begin{align}\label{FormulafixationprobAppend}
	\pi(X_0)= \frac{\int_{X_0}^{1} \nu(y)dy}{\int_0^1 \nu(y) dy}, \quad \text{where } \nu(y) = \exp \left(- 2\int_0^y \frac{v(z)}{D(z)}dz \right).
\end{align}
The drift term $\nu(x)$ and the noise term $D(x)$ are given in equations~\eqref{Drift1dim}, \eqref{Noise1dim}.

\textbf{ First case: $\zeta_1=\zeta_2$.}Then,  $\alpha(X)=(1-\zeta_1)$ does not depend on $X$, and hence 
    \begin{align*}
        -2\frac{v(X)}{D(X)} &= \frac{ \zeta_1^2(\omega_2-\omega_1) }{\left[\omega_1+\zeta_1 (1-\zeta_1)\right]\left[\omega_2+\zeta_1 (1-\zeta_1) \right]}=\mathcal{A}.
    \end{align*}
    Thus $\nu(y)=\mathrm{e}^{\mathcal{A}y}$, which immediately gives $\pi$, 
    \begin{align*}
    	\pi(X_0)= \frac{\mathrm{e}^{\mathcal{A}}-\mathrm{e}^{\mathcal{A}X_0}}{\mathrm{e}^{\mathcal{A}}-1},
    \end{align*} 

    For differentiation of $\pi$ it is by the product rule 
\begin{align*}
    \frac{d}{d\zeta_1} \pi(X_0) 
    &= \left(\frac{d}{d\zeta_1}\mathcal{A}\right) \left( 1- \mathrm{e}^{\mathcal{A}} \right)^{-2} \mathrm{e}^{\mathcal{A}} \left[ 1 - \mathrm{e}^{-X_0 \mathcal{A}}\left( 1-X_0(1-\mathrm{e}^{\mathcal{A}}) \right)  \right].
\end{align*}
Estimating
\begin{align*}
    1 - \mathrm{e}^{-X_0 \mathcal{A}}\left( 1-X_0(1-\mathrm{e}^{\mathcal{A}}) \right)  
    > 1 - \mathrm{e}^{-X_0 \mathcal{A}}\left( 1+\mathcal{A} X_0 \right) > 0 ,
\end{align*}
we find that the sign of $\frac{d}{d\zeta_1} \mathcal{A}$ determines the sign of $\frac{d}{d\zeta_1} \pi(X_0)$. This of course is also true for derivatives with respect to $\omega_1, \omega_2$. Determining the signs of $\frac{d}{d\zeta_1} \mathcal{A}, \frac{d}{d\omega_1} \mathcal{A}, \frac{d}{d\omega_2} \mathcal{A}$ is standard:
\begin{align*}
    \frac{1}{\omega_2-\omega_1}\frac{d}{d\zeta_1} \mathcal{A} &= \frac{2\zeta_1}{\left[\omega_1+\zeta_1 (1-\zeta_1)\right]\left[\omega_2+\zeta_1 (1-\zeta_1) \right]} - \frac{\zeta_1^2(1-2\zeta_1)(\omega_1+\omega_2+2\zeta_1(1-\zeta_1))}{\left[\omega_1+\zeta_1 (1-\zeta_1)\right]^2\left[\omega_2+\zeta_1 (1-\zeta_1) \right]^2}\\
    &= \zeta_1\frac{2\omega_1\omega_2 + \zeta_1(\omega_1+\omega_2)+2\zeta_1^3(1-\zeta_1)}{\left[\omega_1+\zeta_1 (1-\zeta_1)\right]^2\left[\omega_2+\zeta_1 (1-\zeta_1) \right]^2} > 0\\
    \frac{d}{d\omega_1} \mathcal{A} &= -\frac{ \zeta_1^2}{\omega_1 + \zeta_1(1-\zeta_1)^2} < 0\\
    \frac{d}{d\omega_2} \mathcal{A} &= \frac{ \zeta_1^2}{\omega_2 + \zeta_1(1-\zeta_1)^2} > 0
\end{align*}

\textbf{Second case: $\zeta_1 \neq \zeta_2$ but $\omega_1=\omega_2$.}

    \begin{align*}
        -\int_0^y 2\frac{v(X)}{D(X)}dX &=  \int_0^y \frac{ \omega_1^2(\zeta_2-\zeta_1) }{\alpha(X) \left[\omega_1+\zeta_1 \alpha(X) \right]\left[\omega_1+\zeta_2 \alpha(X) \right]}dX\\
        &= \int_{\alpha(0)}^{\alpha(y)}\frac{\omega_1^2}{s \left[ \omega_1+\zeta_1 s\right]\left[  \omega_1+\zeta_2 s \right]}ds\\
        &= \frac{1}{\zeta_2-\zeta_1}\left[ \zeta_1 \log(\omega_1+ \zeta_1s)-\zeta_2 \log(\omega_1 +\zeta_2s) +(\zeta_2-\zeta_1)\log(s) \right]\Bigg|_{\alpha(0)}^{\alpha(y)}.
    \end{align*}
    Due to the denominator in (\ref{FormulafixationprobAppend}), prefactors in $\nu$ need not to be calculated and can be lumped into some constant  $C$ (independent of $y$).
    \begin{align*}
        \nu(y)=C  \Big[ \alpha(y) (\omega_1+\zeta_2\alpha(y))^{-\frac{\zeta_2}{\zeta_2-\zeta_1}} (\omega_1+\zeta_1\alpha(y))^{\frac{\zeta_1}{\zeta_2-\zeta_1}} \Big]
    \end{align*}
    Substituting $s=\alpha(y)$ just as done before gives
    \begin{align*}
        \int_{X_0}^1 \nu(y) dy = C \int_{\alpha(X_0)}^{\alpha(1)} s (\omega_1+\zeta_2s)^{-\frac{\zeta_2}{\zeta_2-\zeta_1}} (\omega_1+\zeta_1s)^{\frac{\zeta_1}{\zeta_2-\zeta_1}}\, ds.
    \end{align*}
    It turns out that $\mathfrak{A}(s)=(\omega_1 + \zeta_1 s)^{\frac{-\zeta_2}{\zeta_2-\zeta_1}} (\omega_1 + \zeta_2 s)^{\frac{\zeta_1}{\zeta_2-\zeta_1}}$ is actually an anti-derivative of the integrand:
    \begin{align*}
        \frac{d}{ds} \mathfrak{A}(s)
        &= \frac{\zeta_1\zeta_2}{\zeta_2-\zeta_1} \left( \omega_1 + \zeta_1 s \right)^{\frac{\zeta_2}{\zeta_2-\zeta_1}} \left( \omega_1 + \zeta_2 s \right)^{\frac{-\zeta_1}{\zeta_2-\zeta_1}} \left( (\omega_1 + \zeta_1 s)^{-1}- (\omega_1 + \zeta_2 s)^{-1} \right)\\
        &= \frac{\zeta_1\zeta_2}{\zeta_2-\zeta_1} \left( \omega_1 + \zeta_1 s \right)^{1+\frac{\zeta_1}{\zeta_2-\zeta_1}} \left( \omega_1 + \zeta_2 s \right)^{1+\frac{\zeta_2}{\zeta_2-\zeta_1}} \frac{(\zeta_2-\zeta_1)s}{(\omega_1 + \zeta_1 s)(\omega_1 + \zeta_2 s)}\\
        &= \zeta_1\zeta_2 s (\omega_1+\zeta_2s)^{-\frac{\zeta_2}{\zeta_2-\zeta_1}} (\omega_1+\zeta_1s)^{\frac{\zeta_1}{\zeta_2-\zeta_1}}.
    \end{align*}
    Plugging the integration bounds into $\mathfrak{A}$ yields $\pi$,
    \begin{align*}
    	\pi(X_0)&= \frac{\mathfrak{A}(1-\zeta_1) - \mathfrak{A}(\alpha(X_0))}{\mathfrak{A}(1-\zeta_1) - \mathfrak{A}(1-\zeta_2)},
    \end{align*} 

    In case both $\omega_1 \neq \omega_2$ and $\zeta_1 \neq \zeta_2$, there is no elementary anti-derivative. One needs to resort to numerical calculation of 
    \begin{align*}
        \pi(X_0)&= \frac{\int_{\alpha(X_0)}^{1-\zeta_1} \Tilde{\nu}(s)ds}{\int_{1-\zeta_2}^{1-\zeta_1}  \Tilde{\nu}(s) ds}
          \end{align*}
        where  $\Tilde{\nu}(y)=C \mathrm{e}^{\mathcal{A}\alpha^{-1}(s)} s (\omega_1+\zeta_2s)^{\frac{\zeta_2}{\zeta_2-\zeta_1}} (\omega_1+\zeta_1s)^{-\frac{\zeta_1}{\zeta_2-\zeta_1}}$.
  
\par\qed\par\medskip

\subsection{Maxima script for the dimension reduction of the SIQS model}
\label{maxCode}

{\small 
\begin{lstlisting}
/* SIMULATION STUDY */

/*  original Fokker Plank */
/*  Dx1 is a symbol for \partial_{x_1} etc.;  Dx1**2 indicates \partial^2_{x_1}. */
/*  Defining c1,c2 in terms of zeta1, zeta2 makes sure the */
/*  results are expressed in zeta1, zeta2.  */
assume(r>0, beta1>gamma, w1>0, sigma>0, zeta1>0, zeta2>0, gamma>mu, mu>0);


c1: (w1/zeta1-w1-mu)$
c2: (w2/zeta2-w2-mu)$
theta1: zeta1/(1-zeta1)$
theta2: zeta2/(1-zeta2)$
gamma: beta1*(1-r)-theta1*mu$
beta2: beta1*(gamma+theta2*mu)/(gamma+theta1*mu)$

FPorig:
    -Dx*(beta1*x*(1-x-u-y-v)-gamma*x)*rho
    -Dy*(beta2*y*(1-x-u-y-v)-gamma*y)*rho
    -(Dx-Du)*(c1*u-w1*x)*rho
    -(Dy-Dv)*(c2*v-w2*y)*rho
    +Du*(mu*u)*rho
    +Dv*(mu*v)*rho
    +h/2*Dx**2*(beta1*x*(1-x-u-y-v)+gamma*x)*rho
    +h/2*Dy**2*(beta2*y*(1-x-u-y-v)+gamma*y)*rho
    +h/2*(Dx-Du)**2*(c1*u+w1*x)*rho
    +h/2*(Dy-Dv)**2*(c2*v+w2*y)*rho
    +h/2*Du**2*(mu*u)*rho
    +h/2*Dv**2*(mu*v)*rho$

/* Check CL*/
ratsimp( subst([h=0, u=theta1*x, y=(r-(1+theta1)*x)/(1+theta2), 
                      v=theta2*(r-(1+theta1)*x)/(1+theta2)], FPorig));

/* transformation itself (forward and back) */
termX: 1/(2*r)*((1+theta1)*x-(1+theta2)*y+r)$
termY1: eps**(-1)*1/2*((1+theta1)*x+(1+theta2)*y -r)$
termY2: eps**(-1)*1/2*(u+v-theta1*x -theta2*y)$
termY3: eps**(-1)*1/2*(u-v-theta1*x +theta2*y)$

termx: (r*X+eps*Y1)/(1+theta1)$
termu: eps*(Y2+Y3)+theta1*((r*X+eps*Y1)/(1+theta1))$
termy: -((r*X-eps*Y1-r)/(1+theta2))$
termv: eps*(Y2-Y3)-theta2*((r*X-eps*Y1-r)/(1+theta2))$

/* check back-transformation; the next terms should simplify to zero. */
ratsimp(subst([x=termx,u=termu,y=termy,v=termv],termX) - X);
ratsimp(subst([x=termx,u=termu,y=termy,v=termv],termY1) - Y1);
ratsimp(subst([x=termx,u=termu,y=termy,v=termv],termY2) - Y2);
ratsimp(subst([x=termx,u=termu,y=termy,v=termv],termY3) - Y3);

/* Define trafo of derivatives */
termDx:diff(termX,x)*DX+diff(termY1,x)*DY1+diff(termY2,x)*DY2+diff(termY3,x)*DY3$
termDu:diff(termX,u)*DX+diff(termY1,u)*DY1+diff(termY2,u)*DY2+diff(termY3,u)*DY3$
termDy:diff(termX,y)*DX+diff(termY1,y)*DY1+diff(termY2,y)*DY2+diff(termY3,y)*DY3$
termDv:diff(termX,v)*DX+diff(termY1,v)*DY1+diff(termY2,v)*DY2+diff(termY3,v)*DY3$

/* trafo of Fokker-Plank in (X,Y1,Y2,Y3)*/
FP:subst([x=termx,u=termu,y=termy,v=termv,Dx=termDx,Du=termDu,
                             Dy=termDy,Dv=termDv,h=eps**(2)],FPorig)$

/* ###########################################
Insert specific choices of parameters here.
On my GPU, it was not possible to have BOTH zeta1 and zeta2 symbolic.
############################################# */
/*FP: subst([r=1/2, mu=1/10, beta1=2, w1=3,w2=3,zeta1=1/2,zeta2=zeta2],FP)$*/
/*FP: subst([beta1=2,r=1/2,mu=1/10,w1=2,w2=2,zeta1=zeta1,zeta2=zeta1+5/100],FP)$*/
FP: subst([beta1=2,r=1/2,mu=8/100,w1=2,w2=2,zeta1=0.6,zeta2=zeta2],FP)$

/* define L0, L1, L2. */
L0:ratcoeff(FP,eps,0)$
L1:ratcoeff(FP,eps,1)$
L2:ratcoeff(FP,eps,2)$

Psi: matrix(
 [ratcoeff(L0,DY1*Y1*rho),ratcoeff(L0,DY1*Y2*rho),ratcoeff(L0,DY1*Y3*rho)],
 [ratcoeff(L0,DY2*Y1*rho),ratcoeff(L0,DY2*Y2*rho),ratcoeff(L0,DY2*Y3*rho)],
[ratcoeff(L0,DY3*Y1*rho),ratcoeff(L0,DY3*Y2*rho),ratcoeff(L0,DY3*Y3*rho)]
)$
Xi: [ratcoeff(L1,DX*Y1*rho),ratcoeff(L1,DX*Y2*rho),ratcoeff(L1,DX*Y3*rho)]$

/* Calculate F such that L0+(F) = DX L1*/
Eq1: Psi[1,1]*x1 + Psi[2,1]*x2 + Psi[3,1]*x3$
Eq2: Psi[1,2]*x1 + Psi[2,2]*x2 + Psi[3,2]*x3$
Eq3: Psi[1,3]*x1 + Psi[2,3]*x2 + Psi[3,3]*x3$
sol: linsolve([Eq1 = -Xi[1], Eq2 = -Xi[2], Eq3 = -Xi[3]], [x1,x2,x3])$
A1: rhs(sol[1])$
A2: rhs(sol[2])$
A3: rhs(sol[3])$
F: A1*Y1+A2*Y2+A3*Y3$

/*Define M*x = vector using (Y1**2, Y1Y2, Y1Y3, Y2**2, Y2Y3, Y3**2) as basis*/
M1: -2*Psi[1,1]*x1   - Psi[2,1]*x2                        - Psi[3,1]*x3;
M2: -2*Psi[1,2]*x1   - (Psi[1,1]+Psi[2,2])*x2        - Psi[3,2]*x3                  -2*Psi[2,1]*x4  - Psi[3,1]*x5;
M3: -2*Psi[1,3]*x1   - Psi[2,3]*x2                        - (Psi[1,1]+Psi[3,3])*x3                          - Psi[2,1]*x5                   -2*Psi[3,1]*x6;
M4:                         - Psi[1,2]*x2                                                              -2*Psi[2,2]*x4   - Psi[3,2]*x5 ;
M5:                         - Psi[1,3]*x2                         - Psi[1,2]*x3                 -2*Psi[2,3]*x4  - (Psi[2,2]+Psi[3,3])*x5   -2*Psi[3,2]*x6;
M6:                                                                     -  Psi[1,3]*x3                                          - Psi[2,3]*x5                   -2*Psi[3,3]*x6;
/* Remainder terms */
Gamma11: 2*ratcoeff(L0,DY1*DY1*rho)$
Gamma12: ratcoeff(L0,DY1*DY2*rho)$
Gamma13: ratcoeff(L0,DY1*DY3*rho)$
Gamma22: 2*ratcoeff(L0,DY2*DY2*rho)$
Gamma23: ratcoeff(L0,DY2*DY3*rho)$
Gamma33: 2*ratcoeff(L0,DY3*DY3*rho)$
/* Check that these are all terms in L0. Should evaluate to zero.*/
ratsimp( Psi[1,1]*DY1*Y1*rho + Psi[1,2]*DY1*Y2*rho +Psi[1,3]*DY1*Y3*rho
            +Psi[2,1]*DY2*Y1*rho + Psi[2,2]*DY2*Y2*rho +Psi[2,3]*DY2*Y3*rho
            +Psi[3,1]*DY3*Y1*rho + Psi[3,2]*DY3*Y2*rho +Psi[3,3]*DY3*Y3*rho
            + Gamma11*DY1*DY1*rho/2
            + Gamma12*DY1*DY2*rho
            + Gamma13*DY1*DY3*rho
            + Gamma22*DY2*DY2*rho/2
            + Gamma23*DY2*DY3*rho
            + Gamma33*DY3*DY3*rho/2
    - L0
);

/* Calculate Integral L1*rho1 dY */
/* The other derivates have to be set to 0 to exact the mixed terms */
L1coeffDX: ratcoeff(subst([DY1=0,DY2=0,DY3=0], L1), DX)$
L1coeffDY1: ratcoeff(subst([DX=0,DY2=0,DY3=0], L1), DY1)$
L1coeffDY2: ratcoeff(subst([DX=0,DY1=0,DY3=0], L1), DY2)$
L1coeffDY3: ratcoeff(subst([DX=0,DY1=0,DY2=0], L1), DY3)$
L1coeffDXDY1: ratcoeff(subst([DY2=0,DY3=0], L1), DX*DY1)$
L1coeffDXDY2: ratcoeff(subst([DY1=0,DY3=0], L1), DX*DY2)$
L1coeffDXDY3: ratcoeff(subst([DY1=0,DY2=0], L1), DX*DY3)$
/* Check that these are all required terms in L1 */
L1notRequired: (  ratcoeff(L1,DY1,2)*DY1**2
                        +ratcoeff(L1,DY2,2)*DY2**2
                        +ratcoeff(L1,DY3,2)*DY3**2
                        + ratcoeff(L1,DY1*DY2)*DY1*DY2
                        + ratcoeff(L1,DY1*DY3)*DY1*DY3
                       + ratcoeff(L1,DY2*DY3)*DY2*DY3
               )$
ratsimp(    L1coeffDX*DX
                +L1coeffDY1*DY1
                +L1coeffDY2*DY2
                +L1coeffDY3*DY3
                +L1coeffDXDY1*DX*DY1
                +L1coeffDXDY2*DX*DY2
                +L1coeffDXDY3*DX*DY3
                +L1notRequired-L1);

/* Evaluate \int L1 rho1 dY = -DX \int F L1 rho0 dY
There are 4 types of terms:
i) \int F*DX(...) dY = DX \int F*(...) dY - \int (DX F)*(...) dY (Product rule)
ii) \int F*DYi(...) dY = -\int (DYi F)*(...) dY = -\int A[i]*(...) dY (Integration by parts)
iii) \int F*DXDYi(...) dY = DX -\int (DX DYi F)*(...) dY 
                               + \int (DX DYi F)*(...) dY (Both i and ii)
iv) \int F*DYi DYj dY = 0 (Integration by parts)
 */
IntegrantL1rho1Drift: -diff(F,X)*L1coeffDX /* i) */
                                    -A1*L1coeffDY1    /* ii) */
                                    -A2*L1coeffDY2    /* ii) */
                                    -A3*L1coeffDY3    /* ii) */
                                    +diff( diff(F,Y1), X)*L1coeffDXDY1 /* iii) */
                                    +diff( diff(F,Y2), X)*L1coeffDXDY2 /* iii) */
                                    +diff( diff(F,Y3), X)*L1coeffDXDY3$ /* iii) */

IntegrantL1rho1Diffusion: F*L1coeffDX     /* i) */
                                        -A1*L1coeffDXDY1 /* ii) */
                                        -A2*L1coeffDXDY2 /* ii) */
                                        -A3*L1coeffDXDY3$ /* ii) */
IntegrantL2rho1Drift: ratcoeff(subst([DY1=0,DY2=0,DY3=0],L2),DX,1)$
IntegrantL2rho1Diffusion: ratcoeff(subst([DY1=0,DY2=0,DY3=0],L2),DX,2)$

/* Use Lemma 1.8 */
/* Driftterm DX(...*f) */
solDrift:linsolve([M1=ratcoeff(IntegrantL1rho1Drift,Y1*Y1*rho),
                            M2=ratcoeff(IntegrantL1rho1Drift,Y1*Y2*rho),
                            M3=ratcoeff(IntegrantL1rho1Drift,Y1*Y3*rho),
                            M4=ratcoeff(IntegrantL1rho1Drift,Y2*Y2*rho),
                            M5=ratcoeff(IntegrantL1rho1Drift,Y2*Y3*rho),
                            M6=ratcoeff(IntegrantL1rho1Drift,Y3*Y3*rho)],
                                                   [x1,x2,x3,x4,x5,x6])$
L1rho1Drift: ratsimp( -rhs(solDrift[1])*Gamma11
                                    -rhs(solDrift[2])*Gamma12
                                    -rhs(solDrift[3])*Gamma13
                                    -rhs(solDrift[4])*Gamma22
                                    -rhs(solDrift[5])*Gamma23
                                    -rhs(solDrift[6])*Gamma33
                                    +subst([Y1=0,Y2=0,Y3=0,rho=1], 
                                              IntegrantL1rho1Drift) )$
/* Diffusionterm DX**2*(...*f) */
solDiffusion:linsolve([M1=ratcoeff(IntegrantL1rho1Diffusion,Y1*Y1*rho),
                            M2=ratcoeff(IntegrantL1rho1Diffusion,Y1*Y2*rho),
                            M3=ratcoeff(IntegrantL1rho1Diffusion,Y1*Y3*rho),
                            M4=ratcoeff(IntegrantL1rho1Diffusion,Y2*Y2*rho),
                            M5=ratcoeff(IntegrantL1rho1Diffusion,Y2*Y3*rho),
                            M6=ratcoeff(IntegrantL1rho1Diffusion,Y3*Y3*rho)],
                            [x1,x2,x3,x4,x5,x6])$
L1rho1Diffusion: ratsimp( -rhs(solDiffusion[1])*Gamma11
                                        -rhs(solDiffusion[2])*Gamma12
                                        -rhs(solDiffusion[3])*Gamma13
                                        -rhs(solDiffusion[4])*Gamma22
                                        -rhs(solDiffusion[5])*Gamma23
                                        -rhs(solDiffusion[6])*Gamma33
                                    +subst([Y1=0,Y2=0,Y3=0,rho=1], IntegrantL1rho1Diffusion) )$
/* L2 */
solDriftL2:linsolve([M1=ratcoeff(IntegrantL2rho1Drift,Y1*Y1*rho),
                            M2=ratcoeff(IntegrantL2rho1Drift,Y1*Y2*rho),
                            M3=ratcoeff(IntegrantL2rho1Drift,Y1*Y3*rho),
                            M4=ratcoeff(IntegrantL2rho1Drift,Y2*Y2*rho),
                            M5=ratcoeff(IntegrantL2rho1Drift,Y2*Y3*rho),
                            M6=ratcoeff(IntegrantL2rho1Drift,Y3*Y3*rho)],[x1,x2,x3,x4,x5,x6])$
L2rho1Drift: ratsimp( -rhs(solDriftL2[1])*Gamma11
                                    -rhs(solDriftL2[2])*Gamma12
                                    -rhs(solDriftL2[3])*Gamma13
                                    -rhs(solDriftL2[4])*Gamma22
                                    -rhs(solDriftL2[5])*Gamma23
                                    -rhs(solDriftL2[6])*Gamma33
                                    +subst([Y1=0,Y2=0,Y3=0,rho=1], 
                                    IntegrantL2rho1Drift) )$
/* Diffusionterm DX**2*(...*f) */
solDiffusionL2:linsolve([M1=ratcoeff(IntegrantL2rho1Diffusion,Y1*Y1*rho),
                                        M2=ratcoeff(IntegrantL2rho1Diffusion,Y1*Y2*rho),
                                        M3=ratcoeff(IntegrantL2rho1Diffusion,Y1*Y3*rho),
                                        M4=ratcoeff(IntegrantL2rho1Diffusion,Y2*Y2*rho),
                                        M5=ratcoeff(IntegrantL2rho1Diffusion,Y2*Y3*rho),
                                        M6=ratcoeff(IntegrantL2rho1Diffusion,Y3*Y3*rho)],
                                        [x1,x2,x3,x4,x5,x6])$
L2rho1Diffusion: ratsimp( -rhs(solDiffusionL2[1])*Gamma11
                                        -rhs(solDiffusionL2[2])*Gamma12
                                        -rhs(solDiffusionL2[3])*Gamma13
                                        -rhs(solDiffusionL2[4])*Gamma22
                                        -rhs(solDiffusionL2[5])*Gamma23
                                        -rhs(solDiffusionL2[6])*Gamma33
                                    +subst([Y1=0,Y2=0,Y3=0,rho=1], 
                                    IntegrantL2rho1Diffusion) )$
/* Minus infront of L1rho1Drift come s from .... */
FP1dim: -DX*(L1rho1Drift+DX*L1rho1Diffusion) +DX*(L2rho1Drift+DX*L2rho1Diffusion)$

/*
/* Drift */
ratcoeff( ratcoeff(FP1dim, DX,1), X, 0);
ratsimp(ratcoeff( ratcoeff(FP1dim, DX,1), X, 1)+ratcoeff( ratcoeff(FP1dim, DX,1), X, 2));
/* Diffusion */
ratcoeff( ratcoeff(FP1dim, DX,2), X, 0);
ratsimp(ratcoeff( ratcoeff(FP1dim, DX,2), X, 1)+ratcoeff( ratcoeff(FP1dim, DX,2), X, 2));
*/
/* Sanity check */
ratsimp(subst([X=0], FP1dim));
ratsimp(subst([X=1], FP1dim));

FP1dimDrift: -ratcoeff(FP1dim,DX,1);
FP1dimDiffusion: 2*ratcoeff(FP1dim,DX,2);


/* Check 1-dimensional Fokker-Planck equation */
/* Divide by r/gamma to account for the rescaling of time */
listofvars(FP1dimDrift);
listofvars(FP1dimDiffusion);
/*  1-dimensional Fokker-Planck equation from the paper */
alpha: (1-zeta2+X*(zeta2-zeta1))$
FP1dimDriftpaper: (
                     -(alpha**2*(w1+zeta1*alpha)*(w2+zeta2*alpha))**(-1) 
                        * X*(1-X)*(1-zeta1)*(1-zeta2)
                        *(zeta1*zeta2*(w2-w1)*alpha+w1*w2*(zeta2-zeta1))
                   );
/* Should be zero */
ratsimp(  FP1dimDriftpaper-(r/gamma)*subst([w1=beta*r*w1,w2=beta*r*w2],FP1dimDrift));
FP1dimDiffusionpaper: (alpha)**(-1)*2*(1-zeta1)*(1-zeta2)*X*(1-X);
/* Should be zero */
ratsimp( FP1dimDiffusionpaper - (r/gamma)*subst([w1=beta*r*w1,w2=beta*r*w2],FP1dimDiffusion));
\end{lstlisting}
}



\end{appendix}

\end{document}